\documentclass[11pt]{article}
\usepackage[utf8]{inputenc}
\pdfoutput=1
\usepackage{url}
\usepackage{amsmath, amssymb, amsfonts}
\usepackage{amsthm}
\usepackage{color}
\usepackage{graphicx, graphics}
\usepackage{enumerate}
\usepackage{setspace}
\usepackage{geometry}
\usepackage{pstricks}
\usepackage{ulem}
\usepackage{authblk}
\usepackage[hang]{subfigure}
\renewcommand{\red}[1]{#1}

\newcommand{\ben}{\vspace{0mm}\begin{equation}}
\newcommand{\een}{\vspace{0mm}\end{equation}}
\newcommand{\be}{\vspace{0mm}\begin{equation*}}
\newcommand{\ee}{\vspace{0mm}\end{equation*}}
\newcommand{\bea}{\vspace{0mm}\begin{equation*}\begin{aligned}}
\newcommand{\eea}{\vspace{0mm}\end{aligned}\end{equation*}}
\newcommand{\bean}{\vspace{0mm}\begin{equation}\begin{aligned}}
\newcommand{\eean}{\vspace{0mm}\end{aligned}\end{equation}}

\newcommand{\Xcal}{\mathcal{X}}
\newcommand{\Ycal}{\mathcal{Y}}

\newcommand{\E}{\mathbb{E}}
\newcommand{\Pro}{\mathbb{P}}
\newcommand{\N}{\mathbb{N}}
\newcommand{\R}{\mathbb{R}}
\newcommand{\xbf}{\mathbf{x}}
\newcommand{\ybf}{\mathbf{y}}
\newcommand{\zbf}{\mathbf{z}}

\newcommand{\nbf}{\mathbf{n}}
\newcommand{\hbf}{\mathbf{h}}

\newcommand{\Zbf}{\mathbf{Z}}
\newcommand{\un}{\mathbf{1}}
\newcommand{\cu}{\underline{c}}

\newcommand{\Bu}{\underline{B}}
\newcommand{\Bb}{\bar{B}}
\newcommand{\Du}{\underline{D}}

\newcommand{\tM}{\widetilde{M}}
\newcommand{\tq}{\widetilde{q}}

\newcommand{\Zm}{\mathcal{Z}}
\newcommand{\Hm}{\mathcal{H}}
\newcommand{\Nm}{\mathcal{N}}
\newcommand{\norm}[1]{||#1||}
\newcommand{\pderiv}[2]{\frac{\partial #1}{\partial #2}}

\newtheorem{thm}{Theorem}[section]
\newtheorem{definition}[thm]{Definition}

\newtheorem{remark}[thm]{Remark}
\newtheorem{prop}[thm]{Proposition}
\newtheorem{lemma}[thm]{Lemma}
\newtheorem{Hyp}{Assumption} 
\newtheorem{subhyp}{}[Hyp]

\begin{document}

\title{Stochastic eco-evolutionary model of a prey-predator community}
\author[1]{Manon Costa \thanks{manon.costa@cmap.polytechnique.fr, corresponding author}}
\author[2]{C\'eline Hauzy}
\author[2]{Nicolas Loeuille}
\author[1]{Sylvie M\'el\'eard}
\affil[1]{CMAP, \'Ecole Polytechnique, CNRS UMR 7641, Route de Saclay, 91128
 Palaiseau Cedex, France}
\affil[2]{Laboratoire EcoEvo, UMR 7625, UPMC, Paris, France}
\date{Februay 03, 2015}
\maketitle

\begin{abstract}
We are interested in the impact of natural selection in a prey-predator community. We introduce an individual-based model of the community that takes into account both prey and predator phenotypes. 
Our aim is to understand the phenotypic coevolution of prey and predators. The community evolves as a multi-type birth and death process with mutations. We first consider the infinite particle approximation of the process without mutation. In this limit, the process can be approximated by a system of differential equations. We prove the existence of a unique globally asymptotically stable equilibrium under specific conditions on the interaction among prey individuals. When mutations are rare, the community evolves on the mutational scale according to a Markovian jump process. This process describes the successive equilibria of the prey-predator community and extends the Polymorphic Evolutionary Sequence to a coevolutionary framework. We then assume that mutations have a small impact on phenotypes and consider the evolution of monomorphic prey and predator populations. The limit of small mutation steps leads to a system of two differential equations which is a version of the canonical equation of adaptive dynamics for the prey-predator coevolution. \red{We illustrate these different limits with an example of prey-predator community that takes into account different prey defense mechanisms. We observe through simulations how these various prey strategies impact the community.}
\end{abstract}


\section{Introduction}
The evolution of a population establishes a link between selected individual characteristics and the environment in which the population lives. Quantifying how the impact of the environment varies along evolutionary trajectories is an important question. Here, we aim at considering how other species interact with the population of interest. These different species compose an ecological community in which each population has a specific role: parasites, predators, resources, etc... The evolution of the different species then modifies the complete interaction network, continuously redefining the selective environment acting on the considered population. The coevolution of different species therefore allows us to consider the feedback loop that links phenotype distributions to environmental variations \cite{intro-dieckman}.\\
In the present paper, we focus on the case of prey-predator communities evolving on similar time scales. As far as ecological dynamics are concerned, there exists an important literature on such predator-prey interactions. In the 1920's, Lotka \cite{lotka} and Volterra \cite{volterra} independently proposed a dynamical system for the ecological dynamics of prey and predators which was then extensively studied (see \cite{TA83m},\cite{HofbauerSigmund},\cite{Murray}).
More recently Marrow, Dieckmann and Law \cite{DML95},\cite{Marrow92},\cite{Marrow96} tackled the question of how natural selection affected the dynamics of such interactions. In the adaptive dynamics framework introduced by Metz, Geritz $\&$ al. and Dieckman and Law \cite{Metz92},\cite{DieckmannLaw96}, these authors developed heuristic tools to study the phenotypic coevolution of monomorphic prey and predator populations and its impact on the network. The survival of prey and predators is strongly conditioned on their respective abilities to defend and hunt. As a result, the understanding of the variety of defense traits and of behavioral and morphological adaptation of predators to these defensive mechanisms has become an important focus for evolutionary ecology (see among others \cite{Strauss02},\cite{Muller04},\cite{lind2013life},\cite{courtois2012differences}).
Considering such coevolutionary dynamics brings up new questions regarding the structure of ecological networks, their stability and the consequences of evolution on their emergent properties (e.g. \cite{Loeuille10},\cite{Dercole06}). For instance, it has been shown that predator-prey coevolution may yield food-web architectures that resemble the ones observed in empirical datasets (\cite{Loeuille05},\cite{rossberg2006food},\cite{caldarelli1998modelling},\cite{drossel2001influence}).
Coevolution of predator-prey interactions may also erode the regulating role of predation \cite{Loeuille04}
and change the overall distribution of energy within the community
\cite{loeuille2006evolution}.
Further models suggest that evolution can select ecological dynamics that are inherently less stable \cite{Loeuille10},\cite{ferriere1993chaotic},\cite{doebeli1995evolution}
or more stable \cite{Abrams00},\cite{abrams1997prey}
than initial systems. It is important to note that the importance of coevolution for ecological network dynamics is not restricted to the realm of mathematical models. Indeed, some of the implications of defense evolution in prey for the stability of ecological dynamics have been reproduced experimentally \cite{yoshida2003rapid},\cite{meyer2006prey}. 
Evolutionary dynamics have also been experimentally reproduced in plant-herbivore systems \cite{Agrawal12}. 
Because the importance of eco-evolutionary dynamics of predator-prey interactions now relies on a strong theoretical background and complementary empirical observations or experimental works, evolution is nowadays largely used in terms of applications. To give just an example, the implications of plant-enemy coevolution for the management of agricultural production has been stressed by many \cite{denison2003darwinian},\cite{thrall2011evolution},\cite{loeuille2013eco}.
\\
In a mathematical setting, Durett and Mayberry \cite{DM10} looked into a specific prey-predator community and considered the phenotypic evolution of prey in a fixed community of predators and vice versa under the assumptions of adaptive dynamics (large population, rare and small mutations). They consider a probabilistic microscopic model of the community, following the rigourous approch developped by Champagnat, Ferri\`ere and M\'el\'eard (\cite{Champagnat06},\cite{CFM08},\cite{CM11}) for the eco-evolutionary dynamics of a population with logistic competition.\\
In this article, we present a stochastic individual-based model for the predator-prey community that evolves as a multi-type birth and death process. The phenotype of an individual is transmitted to its offspring after a potential mutation. The prey phenotypes constrain their defense abilities and influence their reproduction, mortality rate and competition ability. We also consider the evolution of predator phenotypes and model its impact on the predation intensity. \red{ We give an example of prey and predator phenotypes in Section \ref{subsec:example} and we illustrate our results with exact simulations of the individual-based process.} \\
%\bleu{\sout{This article is divided in two main parts. First we consider the evolution of a community composed of $d$ prey sub-populations and $m$ predator sub-populations. Second we introduce mutations during the reproduction events of prey and predators. We assume that the mutation probability is small and use the results of the first part to deduce the behavior of the community including mutations.}}\\
\red{We study the stochastic prey-predator community process in different scalings corresponding to the assumptions of adaptive dynamics: large population, rare mutations and mutations of small impact. Since we assume that mutations are rare, it is important to understand the behavior of the community between two mutations. Therefore we study the evolution of a prey-predator community composed of $d$ prey sub-populations and $m$ predator sub-populations (Section \ref{sec:Model}).
The main question is the composition of this community in a long time scale corresponding to the scale where mutations occur. In the large population limit, the dynamics of the prey-predator community is well approximated by a system of differential equations. In Section~\ref{sec:ODE}, we study the long time behavior of this deterministic system. In particular, we introduce conditions for the existence and uniqueness of a globally asymptotically stable equilibrium. These conditions rely on specific matrices for the interaction between the species. We improve here a result of  Goh, Takeuchi and Adachi (see \cite{Goh78},\cite{TA83m}) in our specific setting. 
The existence of globally stable equilibria is related to optimization problems called Linear complementarity problems. We consider a class of these problems related to the \textit{augmented problems }(see Cottle et al.\cite{cottleLCP}) and extend existing results to our framework.\\
Then we prove in Section \ref{sec:sortie-eq} that the individual-based stochastic process also converges to this equilibrium in finite time and remains close to this equilibrium on a long time scale. In particular we give a result on the exit time of an attractive domain which remains true even for a perturbed process. %(as it will be the case when mutations occur). 
Our result is obtained using the properties of the Lyapunov function associated with the deterministic system as in the work of Champagnat, Jabin et M\'el\'eard \cite{CJM13}. The interest is to highlight the time scale separation between competition phases and mutation occurences. Between two mutations, we can thus
% is the following: if we consider an appropriate scaling of the mutation frequency, we can 
characterize the resident prey-predator community.\\% at the time when a mutation occurs. \\
} 
%In Section~\ref{sec:Model}, we detail the microscopic model and introduce an example that illustrates our results throughout this work. Then, we consider an infinite particle approximation in which the system behaves according to a deterministic system of differential equations. 
%In Section~\ref{sec:ODE}, we study the long time behavior of this deterministic system. In particular, we introduce conditions for the existence and uniqueness of a globally asymptotically stable equilibrium. These conditions rely on specific matrices for the interaction between the species. We improve here a result of  Goh, Takeuchi and Adachi (see \cite{Goh78, TA83m}) in our specific setting. 
%The existence of globally stable equilibria is related to optimization problems called Linear complementarity problems. We consider a class of these problems related to the \textit{augmented problems }(see Cottle et al.\cite{cottleLCP}) and extend existing results to our framework.
%In Section~\ref{sec:sortie-eq}, we derive from the study of the deterministic system, properties on the long time behavior of the stochastic individual-based process. In particular we give a result on the exit time of an attractive domain which remains true even for a perturbed process (as it will be the case when mutations occur). Our result is obtained using the properties of the Lyapunov function associated with the deterministic system as in the work of Champagnat, Jabin et M\'el\'eard \cite{CJM13}.
\red{In Section~\ref{sec:MUT}, we study the impact of rare mutations on the community. The rare mutation framework was first formalized by Champagnat \cite{Champagnat06} for the phenotypic evolution of a population with logistic competition. At each reproduction event, the phenotype of the newborn can be altered by a mutation. We consider the successive invasions of mutants and characterize the survival probability of a mutant trait in a given community. 
In the mutation scale, we prove that the process jumps from a deterministic equilibrium to another one according to the successive mutant invasions. This jump process extends the Polymorphic Evolutionary Sequence to a co-evolutionary framework.}

Finally, we consider the case where mutations have a small impact on phenotypes. Combining these three assumptions (large population, rare mutations and small mutation jumps), we derive a couple of canonical equations describing the coevolution of the prey and predator traits \cite{CM11},\cite{Marrow92},\cite{Marrow96}. 
\section{The model}
\label{sec:Model}
\subsection{The microscopic model}
\label{subsec:micro}
We consider an asexual prey-predator community in which each individual is characterized by its phenotypic traits. At each reproduction event the trait of the parent is transmitted to its offspring. %\red{We will include further on mutations at the birth events.}\\

The interest of this work is the coevolution of prey and predator traits that affects the predation. 
The phenotype $x\in\Xcal$ of a prey individual describes its ability to defend itself against predation. We assume that this trait has an effect on the predation intensity that the prey individual undergoes, but also on its reproduction rate, intrinsic death rate, and ability to compete with other prey individuals.
Such costs may emerge because the energy allocated to defense is diverted from other functions such as growth, maintenance or reproduction (e.g. \cite{herms1992dilemma},\cite{Agrawal12},\cite{lind2013life}). 
The phenotype $y\in\Ycal$ of a predator characterizes its prey consumption rate. This trait affects the predation exerted on prey but also the death rate of the predator. Again, such costs may be explained by differential allocation among life-history traits, but also by behavioral constraints. For instance, increased consumption rate requiring a larger time investment in resource acquisition, it may decrease the vigilance of the predator against its own enemies, creating a mortality cost (see \cite{illius1994costs},\cite{trussell2006fear}).
The trait spaces $\Xcal$ and $\Ycal$ are assumed to be compact subsets of $\R^p$ and $\R^P$ respectively.\\

\noindent The community is composed of $d$ prey types $x_1,\dots,x_d$ and $m$ predator types $y_1,\dots,y_m$. The state of the community is described by the vector of the sub-population sizes. We introduce a parameter $K$ scaling these sub-population sizes (as in \cite{FM04},\cite{CFM08}). To ease the distinction between prey and predator populations we denote by $N^K_i$ the number of prey individuals with trait $x_i$, for $1\le i\le d$, and by $H^K_l$ the number of predators with trait $y_l$, for $1\le l\le m$. Finally the community is represented by the vector
\ben
\label{Z}
\Zbf^K=\frac{1}{K}\bigl(N_1^K,\dots,N_d^K,H_1^K,\dots,H^K_m\bigr),
\een
of the rescaled numbers of individuals holding the different traits.\\
The dynamics of the community follows a continuous time multi-type birth and death process. We first describe the behavior of the prey population. Each prey individual with trait $x$ gives birth to an offspring at rate $b(x)$. The newborn holds the same trait as its parent. The death rate of a prey individual holding trait $x$ is given by
$$\lambda(x,\Zbf^K)=d(x)+\sum_{i=1}^d \frac{c(x,x_i)}{K}N^K_i+\sum_{l=1}^m \frac{B(x,y_l)}{K}H^K_l,
$$ 
where $d(x)$ is the intrinsic death rate of a prey individual with trait $x$, $c(x,x')$ the competition exerted by a prey individual with trait $x'$ on the prey individual with trait $x$ and $B(x,y)$ the intensity of the predation exerted by a predator holding trait $y$ on the prey individual with trait $x$.
In the absence of predators, the prey population evolves as a birth and death process with logistic competition whose behavior was extensively studied by Champagnat, Ferri\`ere and M\'el\'eard \cite{Champagnat06},\cite{CFM08},\cite{CM11}.\\
For the predator population, each predator holding trait $y$ gives birth to a new predator at rate
$$
r\sum_{i=1}^d \frac{B(x_i,y)}{K}N^K_i,
$$
proportional to the predation pressure it exerts on the prey population. The parameter $r$ can be seen as the conversion efficiency of prey biomass into predator biomass. We assume in the following that $r<1$.
In the absence of prey, the predators are unable to reproduce and their population will become extinct rapidly.
Each predator holding trait $y$ dies at rate $D(y)$. The competition between different predators is taken into account through the prey consumption.\\
\noindent The interaction between prey and predators affects the prey death rate and the predator birth rate. This interaction benefits predators but penalizes prey. It creates an asymmetry in the community process and makes it difficult to study: comparisons between two processes whose rates are close, are not possible on a long time scale. We will see in the following how to circumvent this difficulty.
\subsection{An example introducing two types of defenses}
\label{subsec:example}
The diversity of defense strategies observed in nature is overwhelming and the maintenance of such a diversity of strategies is an important focus of evolutionary ecology \cite{ehrlich1964butterflies}. 
Just focusing on one type of consumption interaction, namely plant-herbivore interactions, strategies of defense are morphological (e.g., through spines or \red{trichomes/hair} \cite{zhang2012decreased}), chemical (e.g., the productions of phenols and tannins \cite{becerra2009macroevolutionary}) or through the attraction of enemies of herbivores ("crying for help" \cite{kessler2001defensive}). 
Even when focusing on one defense mechanism, e.g. chemical, the \red{diversity of compounds that are used for defense is very high}, not only in total, but even within species \cite{poelman2008consequences}. 
Modelling such a diversity is challenging and a broad categorization is necessary. 
Here, \red{based on previous empirical or experimental works (see \cite{Strauss02},\cite{Muller04})}, we propose to consider two major classes of defenses, based on their action mode and on the costs they incur: \textit{quantitative defenses} and \textit{qualitative defenses}.\\
\red{
Quantitative defenses correspond to phenotypes that are efficient against a vast number of enemies, but that incur a direct cost in terms of growth or reproduction \cite{Muller04}. Typical examples include structural defenses such as increased toughness \cite{poorter1999comparison}, production of morphological defenses (trichomes, spines) (e.g. \cite{agren1994evolution},\cite{mauricio1997experimental}) or production of digestibility reducing compounds \cite{baldwin1998jasmonate}. In the present work,} we assume that the cost of quantitative defenses affects reproduction (cf. \cite{Muller04},\cite{Strauss02},\cite{lind2013life}).\\
\red{Conversely, qualitative defenses correspond to phenotypes that alleviate consumption by some of the enemies, but incur a cost through another ecological interaction (eg, increased consumption by other enemies or reduced benefits from mutualists \cite{Muller04} ; ``ecological costs'' sensu \cite{Strauss02}).} \red{For instance, alkaloid defenses in plants are efficient against generalist herbivores, but may attract specialists that have evolved to tolerate them or even to use them against their own predators \cite{Muller04}. Other chemical defenses (eg, nicotine) affect the quality of nectar, reducing pollination opportunities \cite{adler2012reliance}. Floral traits such as color or corolla size may reduce the attraction of herbivores, but at the expense of pollinator visitation
\cite{strauss1997floral}. In the present work, qualitative defenses allow a reduction in the effect of one predator, but increase the vulnerability to another predator.}
Because such defense strategies largely impact the similarity of prey niches regarding their enemies \cite{Umea}, we here make the hypothesis that individuals that are closer in terms of qualitative defenses $x$ have a stronger interference competition. 
Such an hypothesis is justified by experimental observations \cite{Agrawal12}, and coherent with the fact that closely related or trait-similar species usually compete more strongly (see \cite{abrams1983theory},\cite{burns2011more}).\\
We take these two types of defenses into account by associating each prey with a two-dimensional trait $x=(q_n,q_a)$ where $q_n\in\R_+$ is the quantity of quantitative defense produced by the prey and $q_a\in \R$ represents its qualitative defense. The allocative trade-off induced by the quantitative defense $q_n$ is represented by an exponential decrease of both the prey birth rate and the predation intensity, at speed $\alpha_n$ and $\beta_n$ respectively. In simulations, we chose a weak allocative trade-off with $\alpha_n=1/10$ and $\beta_n=2$: prey can increase their production of defenses without being too penalized.\\
The predator ability to consume the different qualitative defenses of prey individuals is characterized by two parameters: their preferred qualitative defense $\rho$, and their degree of generalism $\sigma$. \textit{Specialists predators} have a small range $\sigma$ and exert an important predation pressure on the prey populations holding traits close to their preference, while \textit{generalist predators} ($\sigma$ large) consume a large range of qualitative defenses but with less efficiency.  Each predator is then represented by the couple $y=(\rho,\sigma)\in\R\times]0,+\infty[$. The predation intensity decreases with the difference $|\rho-q_a|$ between the preference of predators and the prey qualitative defense. Note that higher generalism incurs a cost in terms of interaction efficiency, as the maximal predation rate is of order $1/\sigma$.\\
In the simulations, we used the following rate functions: for $(q_n,q_a)\in[0,+\infty[\times \R$ and $(\rho,\sigma)\in\R\times ]0,+\infty[$:
\bean
\label{para}
&b(q_n,q_a)=b_0\exp(- \alpha_n q_n), \quad d(q_a,q_n)=d_0,\\
&c(q_n,q_a,q_n',q_a')=c_0\exp\Bigl(-\frac{(q_a-q_a')^2}{2}\Bigr),\\
&B(q_n,q_a,\rho,\sigma)=\exp(- \beta_n q_n)\frac{1}{\sigma}\exp\Bigl(-\frac{(q_a-\rho)^2}{2\sigma^2}  \Bigr),\\
&D(\rho,\sigma)=D.
\eean
\noindent We illustrate this example with exact simulations of the birth and death process introduced above. We are interested in the impact of predators on a prey population using two different qualitative defenses and no quantitative defense: the different prey traits are $x_1=(0,0.8)$ and $x_2=(0,1.7)$. We represent on Figure \ref{fig:1}, the evolution through time of the respective sizes of the prey sub-populations with trait $x_1$ (in green \textcolor{green}{$\times$}), $x_2$ (in red \textcolor{red}{$+$}) and of the predator population holding a trait $(\rho, 0.6)$ for different choices of $\rho$ (in blue \textcolor{blue}{$*$}).\\
\red{When the predator preference differs too much from the prey defense, their population dies out and the two prey populations coexist. In the sequel, we are interested in the cases where the predator population survives.}
We observe three different behaviors. In Figure~\ref{fig:article-1}, the preference of predators is $\rho=0.2$. The three populations coexist on a long time scale. The prey population holding trait $x_2$ has more individuals than the prey population with trait $x_1$ since predation is less important on $x_2$. In Figure~\ref{fig:article-2}, the preference of predators is $\rho=0.7$: predators are well adapted to the trait $x_1$. The predation intensity is so strong on prey holding trait $x_1$ that their population die out. However both populations of predators and prey with trait $x_2$ survive. In Figure~\ref{fig:article-3}, the preference of predators is $\rho=1.26$: they consume both prey populations similarly. We observe that the three populations coexist and that both prey sub-populations have similar small size.\\
\red{As the parameter $\rho$ increases further, we first observe the extinction of the prey population holding trait $x_2$. This is the symmetrical case to \subref{fig:article-2}. Then, we observe similarly to case \subref{fig:article-1} that the three populations coexist.} 

\begin{figure}[h!]
\centering
\subfigure[][$\rho=0.2$]{
\scalebox{0.35}{
\includegraphics{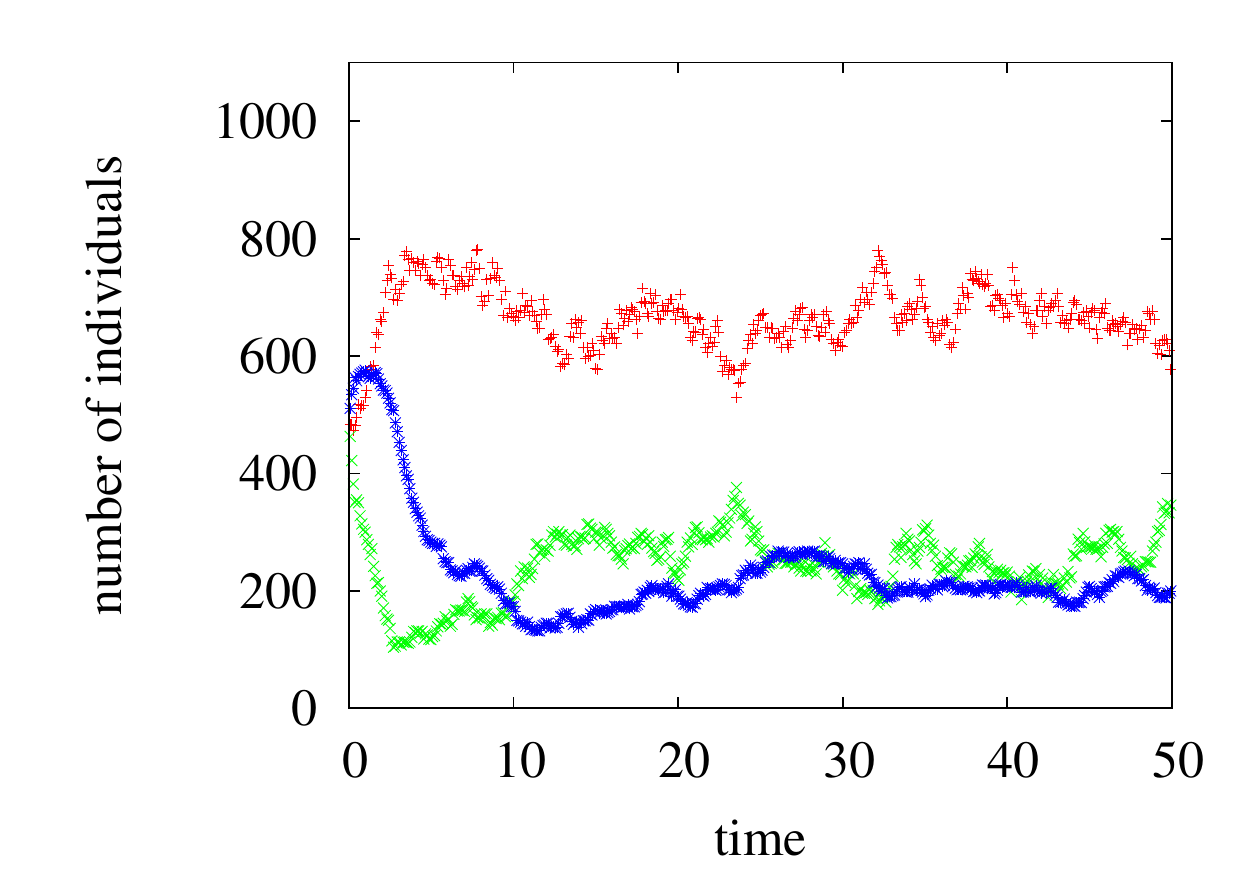}
}
\label{fig:article-1}
}\hfill
\subfigure[][$\rho=0.7$]{
\scalebox{0.35}{
\includegraphics{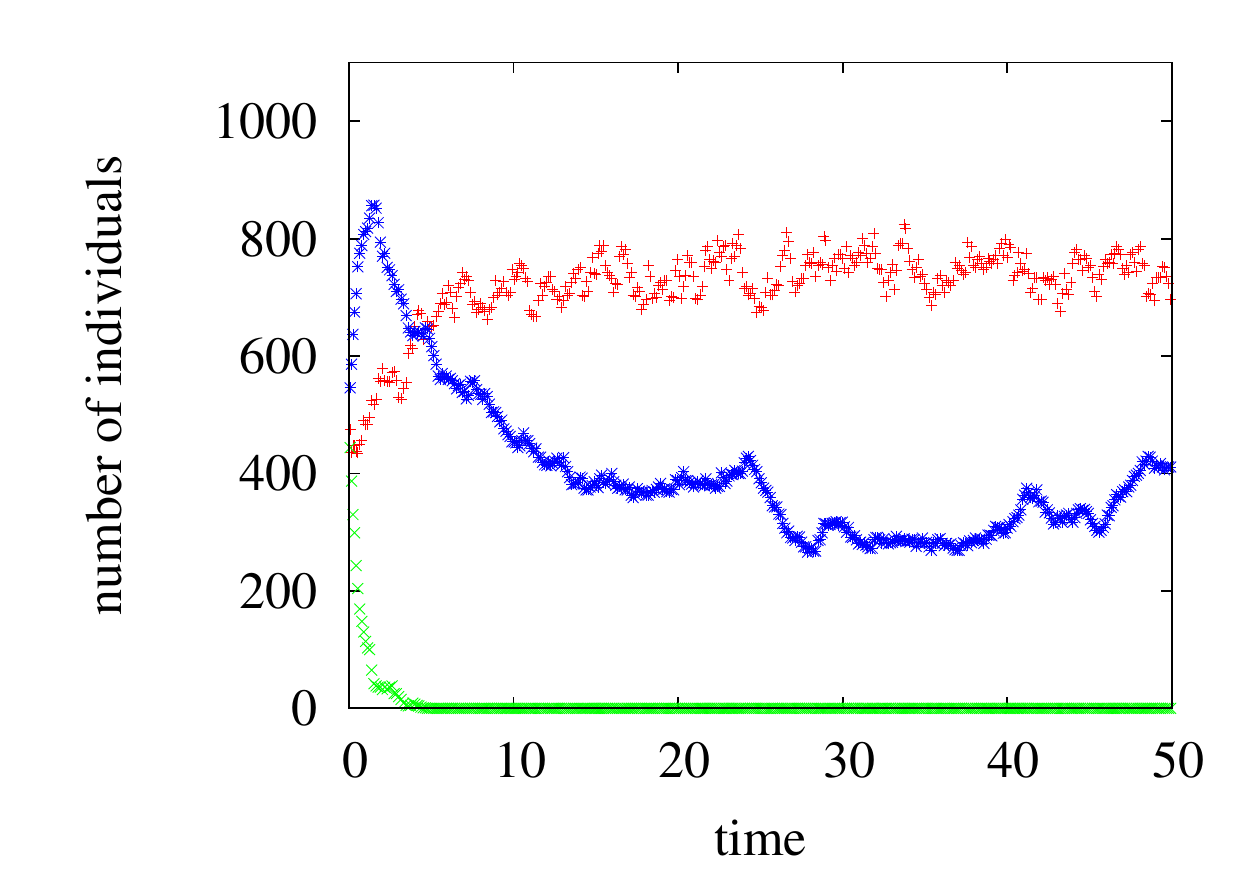}
}
\label{fig:article-2}
}\hfill
\subfigure[][$\rho=1.26$]{
\scalebox{0.35}{
\includegraphics{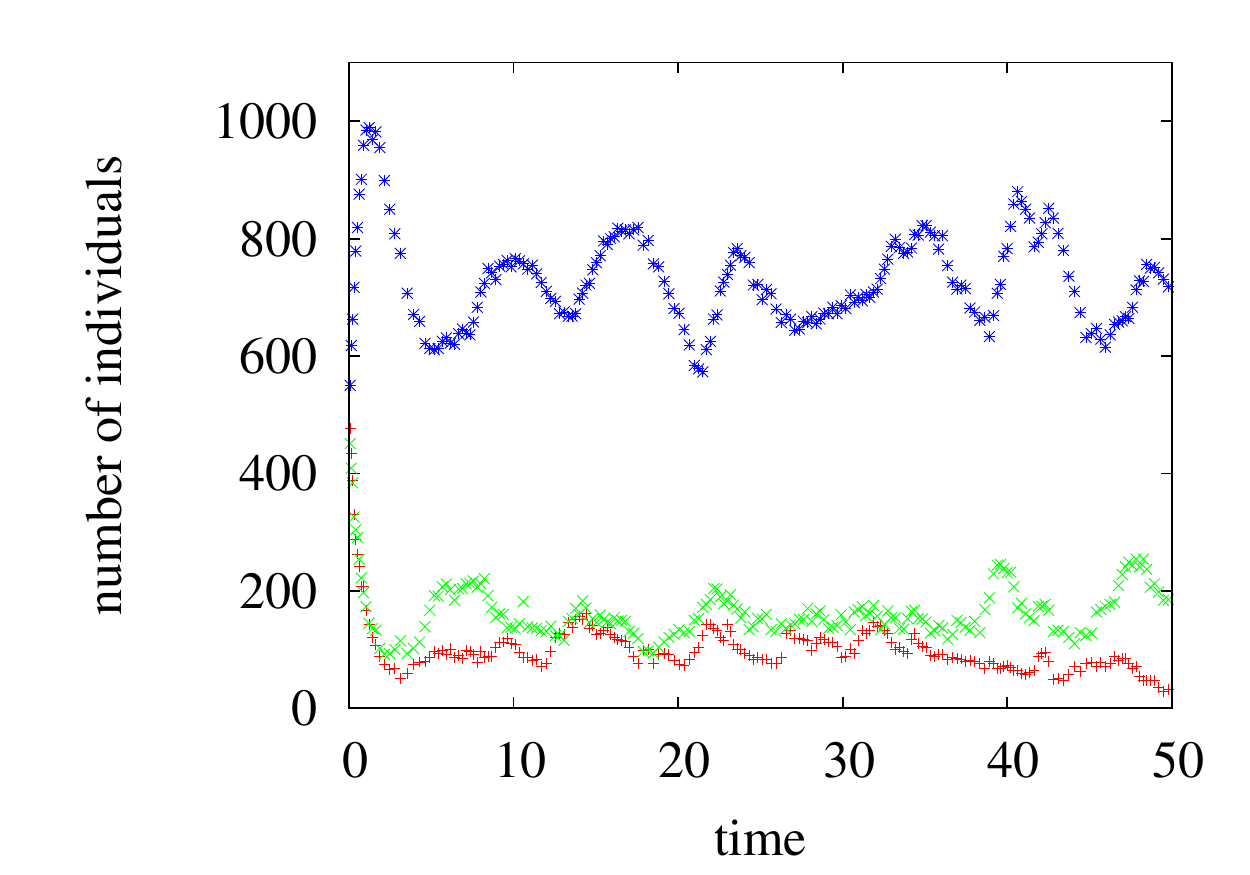}
}
\label{fig:article-3}
}
\caption{We represent the evolution through time of the respective sizes of the prey sub-populations with trait $x_1=(0,0.8)$ (\textcolor{green}{$\times$}), $x_2=(0,1.7)$ (\textcolor{red}{$+$}) and of the predator population with trait $(\rho, 0.6)$ for different choices of $\rho$ (\textcolor{blue}{$*$}). Other parameters are $K=500$, $b_0=2.5$, $d_0=0$, $c_0=1.5$, $D=0.5$, $r=0.8$, $\alpha_n=0.1$, $\beta_n=2$.
}
\label{fig:1}
\end{figure}

%We summarize all these different behaviors in Figure \ref{fig:recap} which includes the case where the preference $\rho$ is to far from the qualitative defense of preys. In this case, the predator population vanishes and both prey populations coexist.} 
%\begin{figure}[h!]
%\centering
%\scalebox{0.5}{
%\input{recap.pdf_t}}
%\caption{}
%\label{fig:recap}
%\end{figure}

\subsection{Existence of the process and uniform bounds of the community size}

\red{The prey-predator community process $\Zbf^K=\frac{1}{K}(N^K_1,\cdots,N^K_d,H^K_1,\cdots,H^K_d)$ introduced above is a Markov process on $(\N/K)^{d+m}$. Its transition rates (or jump rates) are given by the birth and death rates of individuals. \\
%Si on met la suite en appendice, je mettrais une phrase du style
A trajectory of the prey-predator community process can be constructed as solution of a stochastic differential equation driven by Poisson point measures (see \cite{FM04},\cite{CFM08}). This construction is given in Appendix \ref{app:poisson}.\\
The community process is well defined up to the explosion of the number of individuals.
% We prove in this section that the expected size of the community remains finite for all times.
}
We denote by $N^K=\sum_{i=1}^d N_i^K$  the total \red{prey number} and by $H^K=\sum_{l=1}^mH_l^K$ the total number of predators. 
\red{The prey population size $N^K$ jumps of $+1$ each time a prey individual is born and of $-1$ each time a prey individual dies; the predator population size evolves similarly. }
%\sout{\noindent We define two families of independent Poisson point measures on $(\R_+)^2$ with intensity $dsd\theta$: $(R_j)_{1\le j\le d+m}$ for the reproduction events of prey and predators and $(M_j)_{1\le j\le d+m}$ for the death events. 
%Following the pathwise description of the process $\Zbf^K$ introduced in Fournier and Méléard (2004) \nocite{FM04} we can write:}

In the sequel we make the following assumptions:
\begin{Hyp} 
\label{hyp:existence}
The rate functions $b$, $d$, $c$, $B$  and $D$ are continuous, positive and bounded respectively by $\bar{b}$, $\bar{d}$, $\bar{c}$, $\bar{B}$ and $\bar{D}$. Moreover the functions $c$, $B$ and $D$ \red{are bounded below} by positive real numbers $\cu$, $\Bu$ and $\Du$.
\end{Hyp}
\begin{Hyp}
\label{hyp:moment}
The initial condition satisfies $ \sup_K\E\Bigl( ( \frac{N^K(0)}{K})^3+(\frac{H^K(0)}{K} )^3\Bigr) <\infty$.
\end{Hyp}
\red{
The next proposition gives moment properties of the community process and states that the expected population size remains bounded uniformly in $K$ and $t$.}
\begin{prop}
\label{thm:majunif}
Under Assumptions  \ref{hyp:existence} and \ref{hyp:moment}
\begin{description}
\item[(i)]For every $T>0$, 
$$
\sup_K\E\Bigl(\sup_{t\in[0,T]} \bigl( \frac{N^K(t)}{K}\bigr)^3+\bigl(\frac{H^K(t)}{K} \bigr)^3 \Bigr) <\infty.
$$
\item[(ii)] Moreover
$$
\sup_K\sup_{t\ge0}\E\Bigl( ( \frac{N^K(t)}{K}+\frac{H^K(t)}{K})^2\Bigr) <\infty.
$$
\end{description}
\end{prop}
Point (i) justifies the existence of the process $\Zbf^K$ for all times and point (ii) will be used to justify convergence results on long time scales.
The proof of the Proposition is given in Appendix \ref{app:majunif}.
\subsection{Limit in large population}
\label{subsec:largepop}
In this section we study the behavior of the community in a large population limit ($K\to\infty$). We use the same scaling for both populations and establish that the stochastic process $\Zbf^K$ can be approximated by the solution of a deterministic system of differential equations.\\
For $\mathbf{x}=(x_1,\dots,x_d)\in\Xcal^d$ and $\mathbf{y}=(y_1,\dots,y_m)\in\Ycal^m$ we denote by $LVP(\mathbf{x},\mathbf{y})$ the differential system 
\begin{equation}
\label{dpmp}
\left\{ \begin{aligned}
&\frac{dn_i(t)}{dt}= n_i(t) \Bigl( b(x_i)-d(x_i)-\sum_{j=1}^dc(x_i,x_j)n_j(t) -\sum_{l=1}^mB(x_i,y_l) h_l(t) \Bigr),\quad\forall 1\le i\le d,\\
&\frac{dh_l(t)}{dt}= h_l(t)\Bigl(r \sum_{i=1}^dB(x_i,y_l) n_i(t)-D(y_l) \Bigr),\quad\forall 1\le l\le m. 
\end{aligned}\right. 
\end{equation}
A solution of this system is a vector $\zbf=(n_1,\dots,n_d,h_1,\dots,h_m)$.
\begin{prop}
\label{prop:cventempsfini}
Under Assumptions \ref{hyp:existence} and \ref{hyp:moment} and assuming that the sequence of initial conditions $(\Zbf^K(0))_K$ converges in probability toward a deterministic vector $\zbf(0)\in [0,\infty)^{d+m}$, then for every $T>0$ the sequence of processes $(\Zbf^K(t),t\in[0,T])_K$ converges in law in the Skorohod space $\mathbb{D}([0,T],(\R_+)^{d+m})$ toward the unique function $(\zbf(t),t\in[0,T])$ solution of the system $LVP(\xbf,\ybf)$ with initial condition $\mathbf{z}_0$ and satisfying $\sup\limits_{t\in[0,T]} \norm{\zbf(t)} <\infty$.
\end{prop}
\noindent 
\red{The proof follows a classical compactness-uniqueness method developed by Fournier and M\'el\'eard (2004) (Theorem 5.3). First we prove using Proposition \ref{thm:majunif}(i) that the sequence $(\Zbf^K(t),t\in[0,T])_K$ is tight. Then we identify the limit as the unique solution of the system of differential equations $LVP(\xbf,\ybf)$.}

\begin{remark}The extinction of the predator population is not possible in finite time for the solutions of the differential system $LVP(\xbf,\ybf)$. Indeed, if there exists $1\le l\le m$ such that $h_l(0)>0$, then for every $t\ge0$,
$$
\frac{d}{dt} h_l(t) \ge -D(y_l)h_l(t).
$$
Thus $h_l(t)\ge h_l(0)\exp(-D(y_l)t)>0$.\\
Conversely, if there is no predator at time $t=0$, i.e. $\zbf(0)=(\nbf(0),0)$, then the stochastic process $\Zbf^K$ converges toward the solution of a competitive Lotka-Volterra system (denoted by $LVC(\mathbf{x})$) given by:
\bean
 \label{dp}
&\frac{dn_i(t)}{dt}= n_i(t) \Bigl( b(x_i)-d(x_i)-\sum_{j=1}^dc(x_i,x_j)n_j(t)  \Bigr),\quad\forall 1\le i\le d.
\eean
\end{remark}

\section{Long time behavior of the solutions of the deterministic system $LVP$}
\label{sec:ODE}
In this section we study the long time behavior of the solutions to the $LVP(\xbf,\ybf)$ system for fixed $\mathbf{x}=(x_1,\dots,x_d)\in\Xcal^d$ and $\mathbf{y}=(y_1,\dots,y_m)\in\Ycal^m$. To simplify notation, we forget the dependence on traits for the parameters and only use subscripts: for example  $B_{il}=B(x_i,y_l)$.\\
\noindent We are interested in the equilibria of the dynamical system \eqref{dpmp}. 
Hofbauer et Sigmund  proved (Section 5.4, p.47 \cite{HofbauerSigmund}) that the $LVP(\xbf,\ybf)$ systems satisfy the competitive exclusion principle.
This ecological principle states that $m$ different species cannot survive on fewer than $m$ different resources (or in less than $m$ different niches) (see \cite{Amstrong80}). An important consequence is that every asymptotically stable equilibrium $\zbf^*$ of the $LVP(\xbf,\ybf)$ system contains \red{no fewer} prey sub-populations than of predators:
$$\#\{ 1\le i\le d, n^*_i>0\} \ge \#\{ 1\le l\le m, h^*_l>0\}.$$
Therefore the diversity among predators is limited by the diversity among prey.\\
In Subsection \ref{subsec:stab}, we introduce conditions for an equilibrium to be globally asymptotically stable, (i.e. every solution of the system with positive initial condition converges when $t$ goes to $\infty$ toward this equilibrium). This strong notion of stability entails that such an equilibrium is unique. Numerous authors, notably  Goh (\cite{Goh78}), Takeuchi et Adachi (\cite{TA83m}) have already studied this question. We develop here a different approach by improving the Lyapunov function introduced by these authors. \nocite{Goh78,TA83m} 
The interest of this approach is to obtain quantitative information on the behavior of the stochastic process close to the deterministic equilibrium (see Section \ref{sec:sortie-eq}).
Then in Subsection \ref{subsec:existence}, we study the existence of globally asymptotically stable equilibria. This question is related to the existence of solutions to \textit{Linear Complementarity Problems}.
Combining these two results, we derive conditions that ensure the existence of a unique globally asymptotically stable equilibrium for the $LVP$ systems.

\subsection{Condition for global asymptotic stability}
\label{subsec:stab}
We assume the existence of a non-negative equilibrium $\zbf^*=(n^*_1,\dots n^*_d,h^*_1,\dots,h^*_m)$ of the $LVP(\xbf,\ybf)$ system defined in \eqref{dpmp}. We seek conditions on this equilibrium to be globally asymptotically stable.
The global stability relies on the properties of the interaction matrix of the system $LVP$:
\ben
\label{mat-i}
I=\left( \begin{array}{cc}
C&B\\
-rB^T&0 
\end{array}
\right),
\een
where $C=(c_{ij})_{1\le i,j\le d}$ and $B=(B_{il})_{1\le i\le d,1\le l\le m}$. 
We introduce two assumptions on the differential system:
\begin{Hyp}
\label{hyp:ODE}
\begin{subhyp}
\label{hyp:mat}
For every $d\in\N$ and almost every $(x_1,..,x_d)\in\Xcal^d$, the matrix of the competition among prey $C(\mathbf{x})=(c(x_i,x_j))_{1\le i,j\le d}$ satisfies that $C(\mathbf{x})+C(\mathbf{x})^T$ is positive definite.
\end{subhyp}
\begin{subhyp}
\label{hyp:independance}
Let $d,m\in\N$, $\mathbf{x}=(x_1,...x_d)\in\Xcal^d$, and $\mathbf{y}=(y_1,...,y_m)\in\Ycal^m$. Every subsystem of the system $LVP(\xbf,\ybf)$ is non degenerate.
\end{subhyp}
\end{Hyp}
\noindent Assumption \ref{hyp:mat} allows us to define a Lyapunov function for the system $LVP$. \red{As an example, Assumption \ref{hyp:mat} is satisfied for matrices $C=(c_{ij})$ symmetric and strictly diagonally dominant ( $|c_{ii}|>\sum_{j\neq i}|c_{ij}|$). Remark that the competition matrix is symmetric when the competition among preys only depends on the distance between their phenotypes. This is often the case when individuals pay a cost in phenotype matching \cite{Yoder10},\cite{burns2011more}.
Strictly diagonally dominant matrices arise when the competition within the sub-populations is more important than the competition with the other sub-populations. This assumption reflects the impact of the similarity of niches of individuals with close phenotypes \cite{Umea},\cite{burns2011more}.\\
Assumption \ref{hyp:independance} allows to characterize the different equilibria of the $LVP(\xbf,\ybf)$ system with their null and positive components. This assumption reflects that every sub-population plays a different role in the prey-predator community (and in any sub-commnity).}\\

We associate with the equilibrium $\zbf^*$ two subsets containing the subscripts of the traits that disappear in the equilibrium for the prey and predator populations respectively:
\ben
\label{PQ}
P=\{1\le i\le d, n^*_i=0\}\text{ and } Q=\{1\le l\le m, h^*_l=0\}.
\een
The following proposition states conditions for the global asymptotic stability of an equilibrium. 
\begin{prop}
\label{prop:stab}
Let us assume Assumption~\ref{hyp:ODE} and the existence of an equilibrium $\zbf^*$ of the system $LVP(\mathbf{x},\mathbf{y})$ such that
\ben
\label{cond-stab}
\left\{\begin{aligned}
&\forall i\in P,\quad b_i-d_i-\sum_{j=1}^dc_{ij}n_j^* -\sum_{l=1}^mB_{il} h_l^*< 0,\\
&\forall l\in Q,\quad r \sum_{i=1}^dB_{il} n_i^*-D_l< 0,
\end{aligned} \right. 
\een
then this equilibrium is globally asymptotically stable.
Moreover such an equilibrium is unique.
\end{prop}
\noindent Conditions \eqref{cond-stab} ensure that the equilibrium $\zbf^*$ is asymptotically stable. This can be easily obtained by computing the eigenvalues of the Jacobian matrix of the system.
\begin{proof}
We define the function
\ben
\label{Lyap-V}
V(\zbf)=\sum_{i=1}^d r(n_i-n_i^*\log(n_i))+\sum_{l=1}^m (h_l-h_l^*\log(h_l)).
\een
Using the fact that $\zbf^*$ is an equilibrium of the system $LVP(\mathbf{x},\mathbf{y})$, the derivative of $V$ along a solution equals
\bean
\label{dVdt}
\frac{d}{dt}V(\zbf(t)&)=-\frac{r}{2}(\mathbf{n}-\mathbf{n}^*)^T(C+C^T)(\mathbf{n}-\mathbf{n}^*)
\\&+r\sum_{i\in P} n_i(b_i-d_i-\sum_{j=1}^dc_{ij}n_j^* -\sum_{l=1}^mB_{il} h_l^*)+\sum_{l\in Q} h_l( \sum_{j=1}^drB_{jl} n_j^*-D_l).
\eean
Since $\zbf^*$ satisfies \eqref{cond-stab} and by \ref{hyp:mat}, the derivative $\frac{d}{dt}V(\zbf(t)$ is nonnegative but vanishes not only at point $\zbf^*$. In the following we search for a function $\red{W}$ and $\gamma>0$ such that\ben
\label{Lyap-L}
L(\zbf)=V(\zbf)+\gamma \red{W}(\zbf)
\een is a Lyapunov function for the system: for every solution $(\zbf(t); t\ge0)$, the function $L(\zbf(t))$ decreases with time and reaches its only minimum at $\zbf^*$. We set
\ben
\label{Lyap-G}
\red{W}(\zbf)=\sum_{l=1}^m(h_l-h_l^*)\sum_{i=1}^dB_{il}(n_i-n_i^*).
\een
Its derivative along a solution is given by:
$$\begin{aligned}
\frac{d}{dt}&\red{W}(\zbf(t))=
\sum_{l=1}^m h_l r\Bigl( \sum_{i=1}^d B_{il}(n_i-n_i^*) \Bigr)^2
+ \sum_{l\in Q} h_l \Bigl(r \sum_{i=1}^d B_{il}n^*_i-D_l\Bigr)\Bigl(\sum_{j=1}^dB_{jl}(n_j-n_j^*)\Bigr)
\\
- & \sum_{i=1}^d n_i\Bigl(\sum_{l=1}^m B_{il}(h_l-h_l^*) \Bigr)^2
+ \sum_{i\in P} n_i \Bigl(b_i-d_i-\sum_{j=1}^dc_{ij}n_j^* -\sum_{l=1}^m B_{il} h_l^*  \Bigr) \Bigl(\sum_{k=1}^m B_{ik}(h_k-h_k^*)\Bigr)\\
-& \sum_{i=1}^d n_i\sum_{j=1}^d c_{ij}(n_j-n_j^*)\Bigl(\sum_{l=1}^m B_{il}(h_l-h_l^*) \Bigr).
  \end{aligned}$$
%The first term is positive but we can control it thanks to the parameter $\gamma$.\\
The second, third and forth terms are bounded because the solutions of the system are bounded as well. %(Lemme\ref{lem:borne}).\\
The last term can be bounded by :
$$\begin{aligned}
   \sum_{i=1}^d n_i\sum_{j=1}^d c_{ij}(n_j-n_j^*)\sum_{l=1}^m B_{il}(h_l-h_l^*) 
\le
 \sum_{i=1}^d n_i \Bigl( \frac{(\sum_{j=1}^d c_{ij}(n_j-n_j^*))^2}{\Gamma} + \Gamma\Bigl(\sum_{l=1}^m B_{il}(h_l-h_l^*) \Bigr)^2 \Bigr),
  \end{aligned}
$$
where $\Gamma$ will be chosen afterwards.
Together with equation \eqref{dVdt} we can upper bound the derivative of $L$:
\bean
\label{deriv-L}
\frac{d}{dt}L(\zbf(t))\le &-(\mathbf{n}-\mathbf{n}^*)^T(U+U^T)(\mathbf{n}-\mathbf{n}^*)- \gamma(1-\Gamma) \sum_{i=1}^d n_i \Bigl(\sum_{k=1}^m B_{ik}(h_k-h_k^*) \Bigr)^2,\\
&+\sum_{i\in P} n_i(b_i-d_i-\sum_{j=1}^dc_{ij}n_j^* -\sum_{l=1}^mB_{il} h_l^*) (1+\gamma\sum_{k=1}^m B_{ik}(h_k-h_k^*) )\\
&+\sum_{l\in Q} h_k( \sum_{j=1}^dr B_{jl} n_j^*-D_l)(1+\gamma\sum_{i=1}^dB_{il}(n_i-n_i^*))
\eean
where $U=(c_{ij}+ \frac{\gamma}{\Gamma}c_{ij}\sum_{u=1}^d n_u+\gamma\sum_{l=1}^mh_lB_{il} B_{jl}  )_{1\le i, j\le d}$.\\
It remains to choose $\Gamma$ and $\gamma$. We set $\Gamma<1$.
Since the solution $\zbf$ is bounded, it is possible to choose the constant $\gamma$ such that the matrix $U+U^T$ is positive definite and
$$
1+\gamma\sum_{i=1}^dB_{ik}(n_i-n_i^*)>0 ,\quad \forall 1\le k\le m,
\text{ and } 
1+\gamma\sum_{k=1}^m B_{ik}(h_k-h_k^*)>0, \quad \forall 1\le i\le d,
$$
%and such that the matrix $U+U^T$ is positive definite.\\
The derivative of $L( \zbf(t))$ is then non positive and null for the vectors $(u_1,\dots,u_d,v_1,\dots,v_m)$ such that:
\be
\left\{\begin{aligned}
&\forall i\in\{1,...,d\},\quad  u_i=n_i^*,\\
&\forall l\in Q,\quad v_l=h_l^*=0,\\
&\forall i\in\{1,...,d\}, \quad \sum_{l=1}^m B_{il}(v_l-h_l^*) =0.\end{aligned}\right.
\ee
Since $\zbf^*$ is an equilibrium, these conditions are equivalent to 
\be
 \left\{\begin{aligned}
&\forall i\in\{1,...,d\},\quad  u_i=n_i^*,\\
&\forall l\in Q,\quad v_l=h_l^*=0,\\
&\forall i\notin P, \quad b_i-d_i-\sum_{j=1}^d c_{ij}n^*_j -\sum_{l=1}^m B_{il} v_l=0,\end{aligned}\right.
\ee
The vector $(\mathbf{u},\mathbf{v})$ is then an equilibrium $LVP$ having the same null components as $\zbf^*$. Assumption \ref{hyp:independance} ensures that $(\mathbf{u},\mathbf{v})=\zbf^*$.
\end{proof}

\subsection{Existence of globally asymptotically stable equilibria for the system $LVP$}
\label{subsec:existence}
The existence of equilibria of the system $LVP(\xbf,\ybf)$ satisfying \eqref{cond-stab} is related to the existence of solutions to specific optimization problems called Linear Complementarity Problems (LCP) (see \cite{TA82}).

\begin{definition}[Cottle et al. \cite{cottleLCP}]
Given $M\in \R^{u\times u}$ and $q\in\R^u$, the Linear Complementarity Problem associated with $(M,q)$ (denoted by $LCP(M,q)$) seeks a vector $z\in\R^u$ satisfying
\bean
\label{LCP}
\forall 1\le j\le u, \quad &z_j\ge0 \quad \text{and }\quad (Mz+q)_j\ge0,\\
& (Mz+q)^T\cdot z=0.
\eean
\end{definition}
\noindent Note that the last condition can be written $(Mz+q)_j z_j=0$, $\forall 1\le j\le u$.\\
\noindent Let us remark that every equilibrium $\zbf^*\in(\R_+)^{d+m}$ of the system $LVP(\xbf,\ybf)$ satisfying \eqref{cond-stab} is a solution of $LCP(I,R)$ where $u=d+m$, $I$ is the interaction matrix introduced in \eqref{mat-i} and $R=(-(b_1-d_1),\dots,-(b_d-d_d),D_1,\dots,D_m))^T$ is the vector of the growth rates of the sub-populations. Actually, an equilibrium of the system $LVP(\xbf,\ybf)$ satisfying \eqref{cond-stab} is also a solution to $LCP(\widetilde{I},\widetilde{R})$ where 
\ben
\label{tilde}
\widetilde{I}=\left( \begin{array}{c|c}
M&B\\
\hline
-B^T&0 
\end{array}\right),\quad \widetilde{R}=(-(b_1-d_1),\dots,-(b_d-d_d),\frac{D_1}{r},\dots,\frac{D_m}{r})^T.  
\een
\noindent We therefore consider a specific range of LCP related to the shape of the interaction matrix $\widetilde{I}$ which presents a null sub-matrix.
The following result derives easily from existing results (see \cite{cottleLCP}). We detail the proof in Appendix \ref{app:proofALCP}.
\begin{thm}
\label{thm:ALCP}
Let $M\in\mathbb{R}^{d\times d}$ and $q\in\mathbb{R}^d$. For every matrix $B\in(\mathbb{R}_+)^{d\times m}$ and every non-negative vector $D\in \R^{m}$  we define
\begin{equation}
 \label{ALCP}
\widetilde{M}=\left( \begin{array}{c|c}
M&B\\
\hline
-B^T&0 
\end{array}\right)\text{ and } \widetilde{q} = 
\left( \begin{aligned}
&q\\
&D
\end{aligned}\right).
\end{equation}
The problem $LCP(\tM,\tq)$ admits a solution.
\end{thm}
Note that a solution $(\nbf,\hbf)$ of $LCP(\widetilde{I},\widetilde{R})$ is an equilibrium of the $LVP$ system such that
\ben
\label{cond-stab-large}
\left\{\begin{aligned}
&\forall 1\le i\le d, \text{ if } n_i=0 \text{ then } b_i-d_i-\sum_{j=1}^dc_{ij}n_j -\sum_{l=1}^mB_{il} h_l\le 0,\\
&\forall 1\le l\le m, \text{ if } h_l=0 \text{ then }r \sum_{i=1}^dB_{il} n_i-D_l\le 0.
\end{aligned}\right.
\een
These conditions are similar to conditions \eqref{cond-stab} for the global asymptotic stability of an equilibrium, but contain large inequalities. Therefore to obtain the existence of globally asymptotically stable equilibria of the $LVP$ systems we introduce an additional assumption that prevents the quantities involved in conditions \eqref{cond-stab} and \eqref{cond-stab-large} from vanishing. These quantities correspond to the growth rates of prey individuals holding trait $x_i$ and of predators holding trait $y_l$ in a community described by the vector $\zbf^*$. In ecology these quantities are referred to as invasion fitness. We denote the \textit{invasion fitness of a prey individual holding trait $x$ in a community $\zbf^*$} by
\ben
\label{fit-Wp}
s(x;\zbf^*)= b(x)-d(x)-\sum_{i=1}^dc(x,x_i)n^*_i-\sum_{k=1}^m B(x,y_k)h^*_k, \quad\forall x\in\Xcal,
\een
and \textit{invasion fitness of a predator holding trait $y$ in a community $\zbf^*$}  by
\ben
\label{fit-WP}
F(y;\zbf^*)=\sum_{j=1}^drB(x_j,y) n^*_j-D(y) ,\quad \forall y\in Yc.
\een 
\begin{Hyp}
\label{hyp:fitness}
For every $(\mathbf{x},\mathbf{y})\in\Xcal^d\times\Ycal^m$, and every vector $(\mathbf{n},\mathbf{h})$ solution of $LCP(I,R)$, the sets
$
\{x'\in\Xcal, s(x';(\mathbf{n},\mathbf{h}))=0\}
$
and
$
\{y'\in\Ycal, F(y';(\mathbf{n},\mathbf{h})) =0\}
$
have null Lebesgue measure.
\end{Hyp}
In the following we prove that conditions for survival of a small population can be expressed thanks to the fitness functions $s$ and $F$ (we will be interested in the survival of a mutant population). More precisely if a population has a non positive fitness, then it becomes extinct quickly. Otherwise, the population has a chance to invade the resident community. Therefore these fitness functions measure the selective advantage of a trait value in a given community. Assumption \ref{hyp:fitness} is equivalent to assume that every possible trait has either an advantage or a disadvantage in every stable equilibria of the $LVP$ system.\\
\noindent Combining Proposition~\ref{prop:stab} and Theorem \ref{thm:ALCP} we establish that
\begin{thm}
\label{thm:exis-uni-eq}
Under Assumptions \ref{hyp:ODE} and \ref{hyp:fitness}, for almost every $(\mathbf{x},\mathbf{y})\in\Xcal^d\times\Ycal^m$ there exists a unique globally asymptotically stable $LVP(\mathbf{x},\mathbf{y})$. Moreover this equilibrium satisfies \eqref{cond-stab}.
\end{thm}
\noindent In the sequel we denote by $\zbf^*(\xbf,\ybf)=(\mathbf{n}^*(\mathbf{x},\mathbf{y}),\mathbf{h}^*(\mathbf{x},\mathbf{y}))$ the unique globally asymptotically stable equilibrium of the $LVP(\mathbf{x},\mathbf{y})$ system.
\noindent Under the same assumptions we can also establish the existence of a unique globally asymptotically stable equilibrium of the $LVC$ system introduced in \eqref{dp}. We denote by $\bar{\nbf}(\xbf)$ this equilibrium. 

\section{Consequence for the long time behavior of the stochastic process}
\label{sec:sortie-eq}
Let us fix $\xbf\in\Xcal^d$ and $\ybf\in\Ycal^m$ and denote by $\zbf^*=\zbf^*(\xbf,\ybf)$ the unique globally asymptotically stable equilibrium of the system $LVP(\xbf,\ybf)$. 
\red{In this section we study the long time behavior of the prey-predator community process $\Zbf^K$ defined in \eqref{Z}. In Proposition \ref{prop:cventempsfini}, we compare the stochastic process with its deterministic approximation on a finite time interval $[0,T]$, however, on longer time scales the stochastic process may exit the neighbourhood of this approximation. We first prove that $\Zbf^K$ enters in finite time in a neighbourhood of $\zbf^*$. Then, using a probabilistic argument of large deviation, we prove that the trajectory remains in a neighbourhood of $\zbf^*$ during a time of order $\exp(KV)$ for $V>0$. Finally we study the extinction time of small populations which are not adapted in the community.}\\
For every $\varepsilon>0$, we denote by $\mathcal{B}_{\varepsilon}$ the $\R^{d+m}$ sphere of radius $\varepsilon$ centred in $\zbf^*$.
\begin{prop}
\label{prop:cveq}
Let us assume Assumptions \ref{hyp:existence} and \ref{hyp:moment} and that the sequence of initial conditions $\Zbf^K(0)$ converges in probability toward a deterministic vector $\zbf(0)$, then for every $\varepsilon>0$, there exists $t_{\varepsilon}>0$ such that
$$
\lim_{K\to\infty} \Pro( \Zbf^K_{t_{\varepsilon}} \in \mathcal{B}_{\varepsilon})=1.
$$
\end{prop}
\begin{proof}
To prove this result we use classical techniques developed in \cite{Ethier&Kurtz} (Chapter 11, Theorem 2.1) to obtain the convergence in probability uniformly on a time interval of the process $\Zbf^K$: $\forall T>0$, $\forall\varepsilon>0$
$$\lim_{K\to\infty} \Pro\Bigl( \sup_{t\in [0,T]} \norm{\Zbf^K(t)-\zbf(t)}>\varepsilon \Bigr) =1,
$$
where $\zbf(t)$ is the solution of $LVP(\xbf,\ybf)$. The difficulty relies in the fact that the \red{birth and death rates} are only locally Lipschitz functions of the state of the process. However, as the limit function $\zbf(t)$ takes values in a compact set of $\R^{d+m}$, we overcome this difficulty by regularizing the \red{birth and death rates} outside a sufficient large compact set. \\
Moreover there exists a compact set $C$ containing the sequence of initial conditions $(\Zbf^K(0))_{K\ge0}$ with probability converging to $1$. 
We set for every initial condition $z_0\in C$ the last time $t_{\varepsilon}(z_0)$ where the deterministic solution $\zbf(t)$ enters $\mathcal{B}_{\varepsilon}$. This time is finite according to Theorem \ref{thm:exis-uni-eq}.  Since the solutions of the $LVP(\xbf,\ybf)$ system are continuous with respect to their initial condition, the time $t_{\varepsilon}=\sup_{z_0\in C} t_{\varepsilon}(z_0)$ is finite and satisfies that
$
\forall t>t_{\varepsilon}$, $\sup_{\{z_0\in C\}} \norm{\zbf(t)-z^* }<\varepsilon.
$
Combining these two results, we conclude the proof of Proposition \ref{prop:cveq}.
\end{proof}
\bigskip
We then study the time spent by $\Zbf^K$ in the neighbourhood of $\zbf^*$. The estimate of \textit{the exit time of an attractive neighbourhood} gives a good scaling for the introduction of rare mutations in the next section.
This result relies usually on the large deviation theory. However, classical techniques cannot be applied in our setting since the \red{birth and death} rates of $\Zbf^K$ are not bounded uniformly away from zero.
We introduce here a different method which allows to extend the result to perturbations of the process $\Zbf^K$. In particular, we aim at considering small mutant populations that interact with the process $\Zbf^K$ or at modifying the \red{birth and death} rates introduced in Section~\ref{sec:Model}.
\red{Another interest in considering perturbations of the process is the study of the stability or resilience of this prey-predator network (see the seminal work of May \cite{may2001stability} or \cite{thebault2010stability},\cite{ives2007stability} for more recent references)}.\\
\noindent We define a perturbation \red{$\Zm^K=(\mathcal{N}^K_1,\cdots,\mathcal{N}^K_d,\mathcal{H}^K_1,\cdots,\mathcal{H}^K_m)$} of the process $\Zbf^K$ by $2$ families of $d+m$ real-valued random processes $(u^K_i)_{1\le i \le d+m}$ and $(v^K_i)_{1\le i \le d+m}$ predictable with respect to the filtration $\mathcal{F}_t$ generated by the sequence of processes $\Zbf^K$. 
\red{
The sequence $(u^K_i)_{1\le i \le d+m}$ describes the modifications of the birth rates of the prey and the predator populations while the sequence $(v^K_i)_{1\le i \le d+m}$ gives the modifications of the death rates. The modified process evolves as follows:
\begin{itemize}
\item For $1\le i\le d$, the perturbed prey population $\mathcal{N}^K_i$ evolves as a birth and death process with individual birth rate $b(x_i)+u_i^K(t)$ and individual death rate $\lambda(x_i,\Zm^K(t))+v^K_i(t)$ at time $t$.
\item For $1\le l\le m$, the perturbed predator population $\mathcal{H}^K_l$ evolves as a birth and death process with individual birth rate $r\sum_{i=1}^dB(x_i,y_l)\mathcal{N}^K_i+u_{d+l}^K(t)$ and individual death rate $D(y_l)+v^K_{d+l}(t)$ at time $t$.
\end{itemize}
In the case where $u^K_i=v^K_i=0$ for all $1\le i\le d+m$ the process $\Zm^K$ is the prey-predator community process $Z^K$.\\
We assume that the processes $(u^K_i)_{1\le i \le d+m}$ and $(v^K_i)_{1\le i \le d+m}$ are uniformly bounded by $\kappa$.
}
%\bleu{détailler par exemple ce qui se passe si on prend la population résidente et un mutant. Que sont alors les valeurs de $u_i$, et $v_i$.\\}
\begin{thm}
\label{thm:sortie_eq}
For every $\varepsilon$ small enough, there exist a constant $V_{\varepsilon}>0$ and $\varepsilon''<\varepsilon$ such that if $\kappa$ is small enough and $\Zm^K(0)\in\mathcal{B}_{\varepsilon''}$, then the probability that the process $
(\Zm^K(t); t\ge0)$ exits the neighbourhood $\mathcal{B}_{\varepsilon}$ after a time $e^{V_{\varepsilon}K}$
converges to $1$ as $K\to\infty$.
\end{thm}
\noindent The results is obtained using the method developed by Champagnat, Jabin et M\'el\'eard (Proposition 4.2  \cite{CJM13}). We detail the proof in Appendix~\ref{app:proof_sortie_eq} and give hereby the main ideas in the non perturbed setting.
\begin{proof}[Ideas of the proof]
We recall the definition of $P$ and $Q$ in \eqref{PQ} and set 
$$\norm{\zbf-\zbf^*}_{PQ}=\sum_{i\notin P}|n_i-n_i^*|^2 +\sum_{i\in P}|n_i|+\sum_{l\notin Q}|h_l-h_l^*| +\sum_{l\in Q}|h_l|.$$
The Lyapunov function $L$ for the system \eqref{dpmp} defined by \eqref{Lyap-L} with an appropriate choice of $\gamma$ is smooth in the neighbourhood of $\zbf^*$. 
In particular we can define three non negative constants $C$, $C'$ and $C''$ such that
\bean
\label{Prop1}
  \norm{\zbf-\zbf^*}^2&\le  \norm{\zbf-\zbf^*}_{PQ}\le  C\Bigl(L(\zbf)-L(\zbf^*)\Bigr)\le
CC'  \norm{\zbf-\zbf^*}_{PQ},
\eean
and
\begin{equation}
  \label{Prop2}
\frac{d}{dt}L(\zbf(t)) \le -C''  \norm{\zbf-\zbf^*}^2,
\end{equation}
We introduce the stopping time $\tau^K_{\varepsilon} =\inf\{t\ge 0, \Zbf^K\notin\mathcal{B}_{\varepsilon}\}$.
\red{Let $T$ be a positive time to be chosen afterwards.} Thanks to the semi-martingale decomposition of the process $L(\Zbf^K(t))$ we can prove that for every $K$ large enough, there exists $C'''>0$ such that for all $ t\le T\wedge \tau_{\varepsilon}^K$:
\bean
\label{maj_norm}
\norm{\Zbf^K(t)-\zbf^*}^2\le
C&\Bigl[C'\norm{\Zbf^K(0)-\zbf^*}_{PQ}+ \sup_{[0,T]}|M^K_t| -C''\int_0^t\norm{\Zbf^K(s)-\zbf^*}^2-C'''\frac{1}{K}  ds\Bigr],
\eean  
where $M^K_t$ is a local martingale with zero mean which can be written explicitly using compensated Poisson point measures \red{(see Appendix \ref{app:proof_sortie_eq}).}\\
We define for every $\kappa>1/K$, $S_{\kappa}=\inf\{t\ge0, \norm{\Zbf^K(t)-\zbf^*}^2\le 2C'''\kappa\}$ and introduce
\begin{equation}
\label{T_eta}
 T_{\kappa}=\frac{C'(\norm{\Zbf^K(0)-\zbf^*}_{PQ})+
\sup_{[0,T]}|M^K(t)|}{C''C'''\kappa},
\end{equation}
which represents the maximal time that the process $\norm{\Zbf^K-\zbf^*}^2$ can spend above the threshold $2C'''\kappa$ before the time $T\wedge \tau^K_{\varepsilon}$.
The inequality \eqref{maj_norm} becomes for all $t\le S_{\kappa}\wedge T
\wedge\tau_{\varepsilon}^K$ 
\be
\norm{\Zbf^K(t)-\zbf^*}^2\le CC''C'''\kappa T_{\kappa}.
\ee
This equation connects the time spent by the process outside a ball, with the values taken by $\norm{\Zbf^K-\zbf^*}^2$ during this time interval. Therefore if we bound the values of $T_{\kappa}$, we control the process $\norm{\Zbf^K-\zbf^*}^2$ and consequently the exit time $\tau^K_{\varepsilon}$.
To estimate $T_{\kappa}$ we need to control exponentially the values of the martingale $M^K_t$ uniformly on a time interval. To this aim, we use the following lemma.
\begin{lemma}[Graham, M\'el\'eard - Proposition 4.1 \cite{GM97}]
\label{lem:mart}
For every $\alpha>0$ and $T>0$ there exists a constant $V_{\alpha,T}$ satisfying that for all $K$ large enough:
$$\Pro\bigl(\sup_{[0,T\wedge \tau_{\varepsilon}^K]} |M^K_t| >\alpha \bigr)\le
\exp(-KV_{\alpha,T})$$
\end{lemma}

\noindent With this result and \eqref{maj_norm} we study for $\varepsilon''<\varepsilon'<\varepsilon$, the number of back and forth, $k_{\varepsilon}$ between the balls $\mathcal{B}_{\varepsilon''}$ and $\mathcal{B}_{\varepsilon'}$ before the exit of $\mathcal{B}_{\varepsilon}$. 
\red{With an appropriate choice of the parameters $\varepsilon'$ and $T$, w}e establish that $k_{\varepsilon}$ is smaller than a geometric random variable with parameter $\exp(-KV)$, thus
$$\Pro(k_{\varepsilon}>\exp(KV/2))=1-(1-\exp(-KV))^{\exp(KV/2)} \underset{K\to\infty}{\longrightarrow} 1
$$
To conclude it remains to show, using \eqref{maj_norm} again, that these back and forth require a time of order 1.
\end{proof}
\bigskip
Finally we study the behavior of the process while it remains close to the equilibrium $\zbf^*$. The equilibrium $\zbf^*$ can have zero components and we establish that the associated stochastic sub-populations become extinct in a time of order $\log K$. We introduce the stopping time
$$
S^K_{ext}=\inf\{t\ge0, \forall i\in P, N^K_i(t)=0 \text{ and } \forall l\in Q, H^K_l(t)=0\},
$$ 
and set $S^K_{ext}=0$ if both $P$ and $Q$ are empty.
\begin{prop}
\label{prop:extinction}
Let $\varepsilon>\varepsilon''>0$ small enough. If the initial condition $\Zbf^K(0)\in\mathcal{B}_{\varepsilon''}$,then there exists $a>0$ such that
$$\lim_{k\to\infty} \Pro(S^K_{ext}\le a\log K) =1.$$ 
\end{prop}
\begin{proof}
Fix $l\in Q$. We prove the result for the predator population holding trait $y_l$, and the same reasoning can be applied to a prey population holding trait $x_i$, for $i\in P$ (see Theorem 4 in \cite{Champagnat06}).\\ 
Theorem~\ref{thm:exis-uni-eq} ensures that the fitness $F(y_l;\zbf^*)$ is negative.
We define the constant $V_{\varepsilon}$ associated by Theorem~\ref{thm:sortie_eq} to the exit time $\tau^K_{\varepsilon}$ of the ball $\mathcal{B}_{\varepsilon}$.
For every $t\le \tau^K_{\varepsilon}$, the number of predators $H^K_l(t)$ is bounded from below by a continuous time birth and death process $H$ with birth rate $\lambda=r\sum_{i=1}^dB(x_i,y_l) (n^*_i+\varepsilon)$, death rate $\mu=D(y_l)$ and initial condition $H^K_l(0)\le K\varepsilon''$. We choose $\varepsilon$ small enough for the process $H$ to be sub-critical: $\varepsilon < -F(y_l,\zbf^*)/(r\sum_{i=1}^dB(x_i,y_l))$.\\
From classical results on branching processes (see \cite{AthreyaNey} p.109), we obtain that 
$$
\Pro(H(t)=0|H(0)=1 )=1-\frac{\mu-\lambda}{\mu\exp(-(\lambda-\mu)t) + \lambda}.
$$
Since $\forall h_0\in \N$, $\Pro(H(t)>0|H(0)=h_0)=1-\Pro(H(t)=0|H(0)=1)^{h_0}$, we deduce that for every initial condition $0\le h_0\le K\varepsilon''$, 
$$
\Pro(H(t)>0|H(0)=h_0)\le 1- \bigl(1-\frac{\mu-\lambda}{\mu\exp(-(\lambda-\mu)t) + \lambda}\bigr)^{K\varepsilon''}.
$$
We set $1>\delta>0$ and apply the previous inequality to the positive time $t_K^l= (\frac{\delta-1}{\lambda-\mu})\log(K)$. We obtain that $\forall 0\le h_0\le K\varepsilon''$, 
$$
\Pro(H(t_K^l)>0|H(0)=h_0)\le 1- \bigl(1-\frac{\mu-\lambda}{\mu K^{1-\delta} + \lambda}\bigr)^{K\varepsilon''}\underset{K\to\infty}{\longrightarrow} 0.
$$
We conclude the proof by choosing $a\log K$ as the maximal $t^l_K$ for $l\in Q\cup P$.
\end{proof}
 \section{Evolution of the process in a rare mutation time scale}
\label{sec:MUT}
In this section, mutations happen during the prey and predator reproduction events. We observe their impact on the dynamics of the community. The coevolution of the traits depends on the occurrence of mutations and the invasion of the mutant population. We seek conditions for the survival of a mutant population and study the consequences of the fixation of a mutation for the prey-predator community.\\
The individual birth and death rates are defined as in Section~\ref{sec:Model}. 
The mutation events are added as follows
\begin{itemize}
\item when a prey individual with trait $x$ gives birth, the trait of its offspring is affected by a mutation with probability $u_Kp(x)$. The newborn holds a trait $x+l$ where $l$ is distributed according to $\pi(x,l)dl$. Otherwise (with probability $1-u_Kp(x)$) the newborn inherits its parent trait $x$.
\item Similarly for each predator holding a trait $y$. At each reproduction event, with probability $u_KP(y)$ the trait of the offspring is affected by a mutation: it holds the trait $y+l$ where $l$ is distributed according to $\Pi(y,l)dl$. Otherwise the newborn inherits its parent trait $y$.
\end{itemize}
\red{The same parameter $u_K$ scales the mutation frequencies in both prey and predator populations. This assumption is consistent with the fact that the demographic dynamics of both populations happens on the same time scale (Section \ref{sec:sortie-eq}).
When the parameter $u_K$ is small, the mutations are rare. We assume in the sequel that $Ku_K\to0$ as $K\to \infty$. This assumption measures the rarity of the mutations and is consistent with the theory of adaptive dynamics (\cite{Metz92},\cite{DieckmannLaw96}).\\
} 
In subsection \ref{subsec:simu} we illustrate the impact of mutations on the example introduced in section \ref{subsec:example}. In subsection \ref{subsec:mutrare} we consider the limit of the community process under the assumptions of infinite population and rare mutations. We extend the results obtained by Champagnat \cite{Champagnat06} to the prey-predator coevolution. Finally in subsection \ref{subsec:eqcano} we consider a limit when the mutation steps are small. We prove that the coevolution of the prey and predator traits can be described by the deterministic coupled system of differential equations introduced by Dieckmann, Law and Marrow \cite{Marrow96}. This system extends the \textit{canonical equation of adaptive dynamics} to the coevolution of a prey-predator interaction.

\subsection{Simulations}
\label{subsec:simu}
Let us consider again the example introduced in section~\ref{subsec:example} in which prey individuals are characterized by a trait $x=(q_n,q_a)$ where $q_n$ is the quantity of quantitative defenses they produce and $q_a$ the type of qualitative defense they use. The predators are characterized by $y=(\rho,\sigma)$ where $\rho$ reflects the qualitative value they prefer and $\sigma$ is their range. The mutations are distributed according to gaussian distributions, centred in the trait of the parent with covariance matrices $\gamma$ and $\Gamma$ for prey and predators respectively.\\
We illustrate in different cases the impact of mutations on the community. We will observe the convergence on the rare mutation scale toward a pure jump process taking values in the set of couples of finite measures on the trait spaces $\Xcal $ and $\Ycal$ respectively.

\subsubsection{Co-evolution of the qualitative defense $q_A$ and the predator preference $\rho$}
\noindent We first consider the coevolution of the prey trait $q_a$  and of the predator trait $\rho$.  Both traits are associated through the predation function $B$, and the defense trait $q_a$ influences the competition among prey. In these simulations we assume that mutations do not affect the prey trait $q_n$ and the predator trait $\sigma$. We consider three cases: first we assume that no mutation occurs in the predator population (Figure \ref{fig:evolx1}), then the opposite case where mutations only occur in the predator population (Figure \ref{fig:evolx3}), finally we study the coevolution of the traits (Figure \ref{fig:evolx4}).

\noindent In the first case we assume that no mutation occurs in the predator population: $P=0$.
The initial community is composed of $K$ prey individuals holding trait $x=(0.3,0.4)$ and $K$ predators holding trait $y=(0.2,0.6)$. The mutation probability $u_K=5\cdot 10^{-5}$ is small.
Figure \ref{fig:evolx1-1} gives the different values of $q_a$ carried by prey and of $\rho$ carried by predators for all times. We observe that natural selection favours the values of $q_a$ far from $\rho$. The predator population dies out when the defense $q_a$ gets to far away from their preference. The extinction time is represented by a vertical line on the three graphs. As long as predators are present in the community, we observe that the prey traits are concentrated in a single value: the prey population remains monomorphic.\\
In the other graphs, we focus on the demographic dynamics. Figure~\ref{fig:evolx1-pred} gives the dynamics of the number of predators through time.
On Figure~\ref{fig:evolx1-proie} we represent the size of the prey sub-populations with the following traits: the initial trait value $(0.3,0.4)$ in green, $(0.3,0.664)$ in blue, and $(0.3,1.285)$ in pink (the same colors are used on Figure \ref{fig:evolx1-1}). On these graphs we observe the impact of the mutations on the community.
The mutation $(0.3, 0.664)$ is the first to invade the initial community and to replace the resident prey holding trait $(0.3,0.4)$. e observe that before the appearance of this mutation the respective numbers of predators and prey $(0.3,0.4)$ remain stationary. Some mutations have appeared but their population remained small (less than 10 individuals). \red{This phenomenon illustrates the stationarity of the prey and predator population sizes near the deterministic equilibria of the $LVP$ system stated in Theorem \ref{thm:sortie_eq}.}\\
The invasion of the mutation $(0.3,0.664)$ is characterized by a fast extinction of the resident prey population and a fast growth of the mutant population. Meanwhile, the number of predators diminishes to another stationary value. \red{The extinction speed of the resident population is given by Proposition \ref{prop:extinction}}.\\
The invasion of a mutant prey holding trait $(0.3,1.285)$ in the resident community composed of prey holding trait $(0.3,0.664)$ and predators, drives the predators to extinction. The extinction of predators is a direct consequence of the prey phenotypic evolution: it is called an evolutionary murder (see \cite{Dercole06}). Afterwards both prey populations survive. Note that their respective population sizes are similar: they have indeed the same natural birth and death rates and similar ability for competition. In this simulation, the prey population remains dimorphic after the predator extinction and both traits are driven apart by the competition.
 \begin{figure}[h!]
 \begin{minipage}{0.3\linewidth}
 \subfigure[][]{
 \scalebox{0.38}{
 \includegraphics{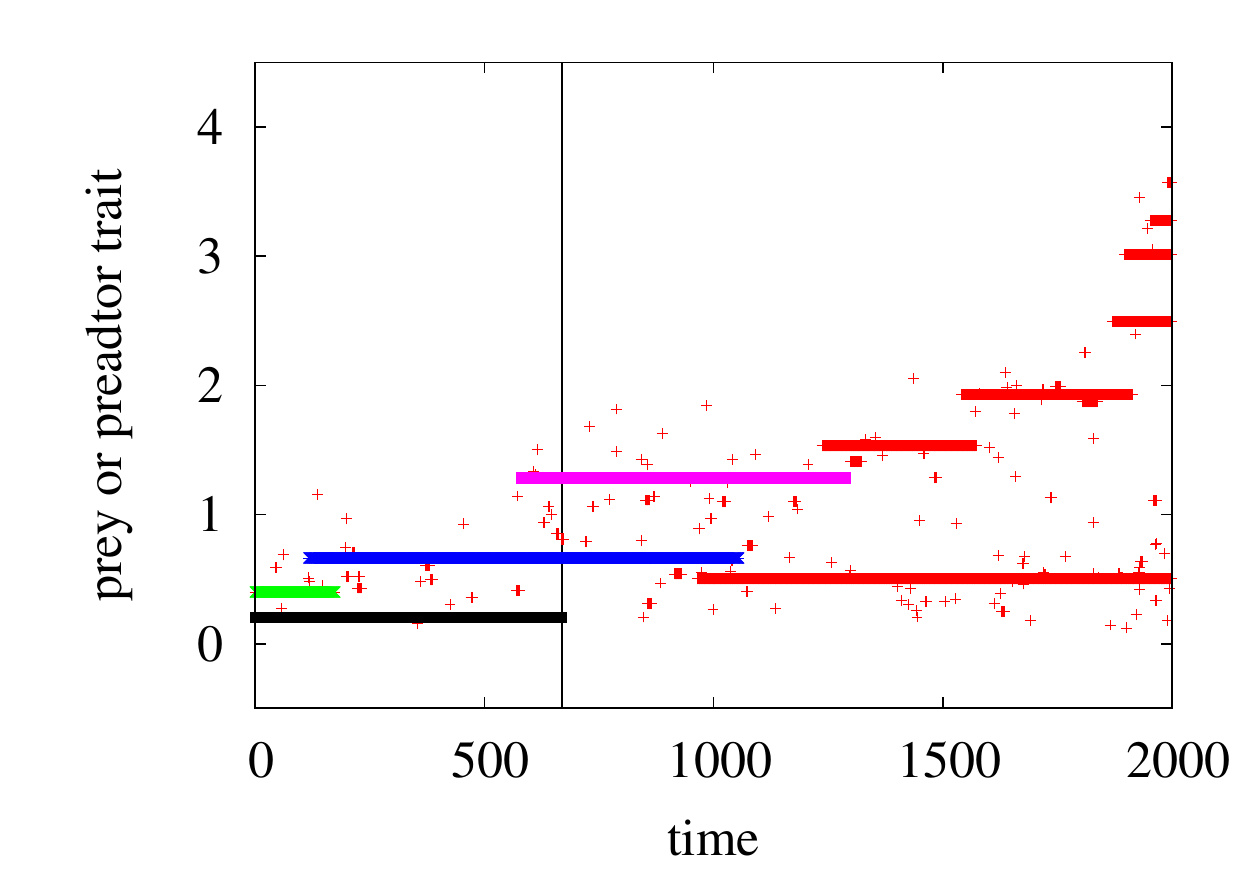}}
 \label{fig:evolx1-1}}
 \end{minipage}
 \hfill
 \begin{minipage}{0.3\linewidth}
 \subfigure[][]{
 \scalebox{0.38}{
 \includegraphics{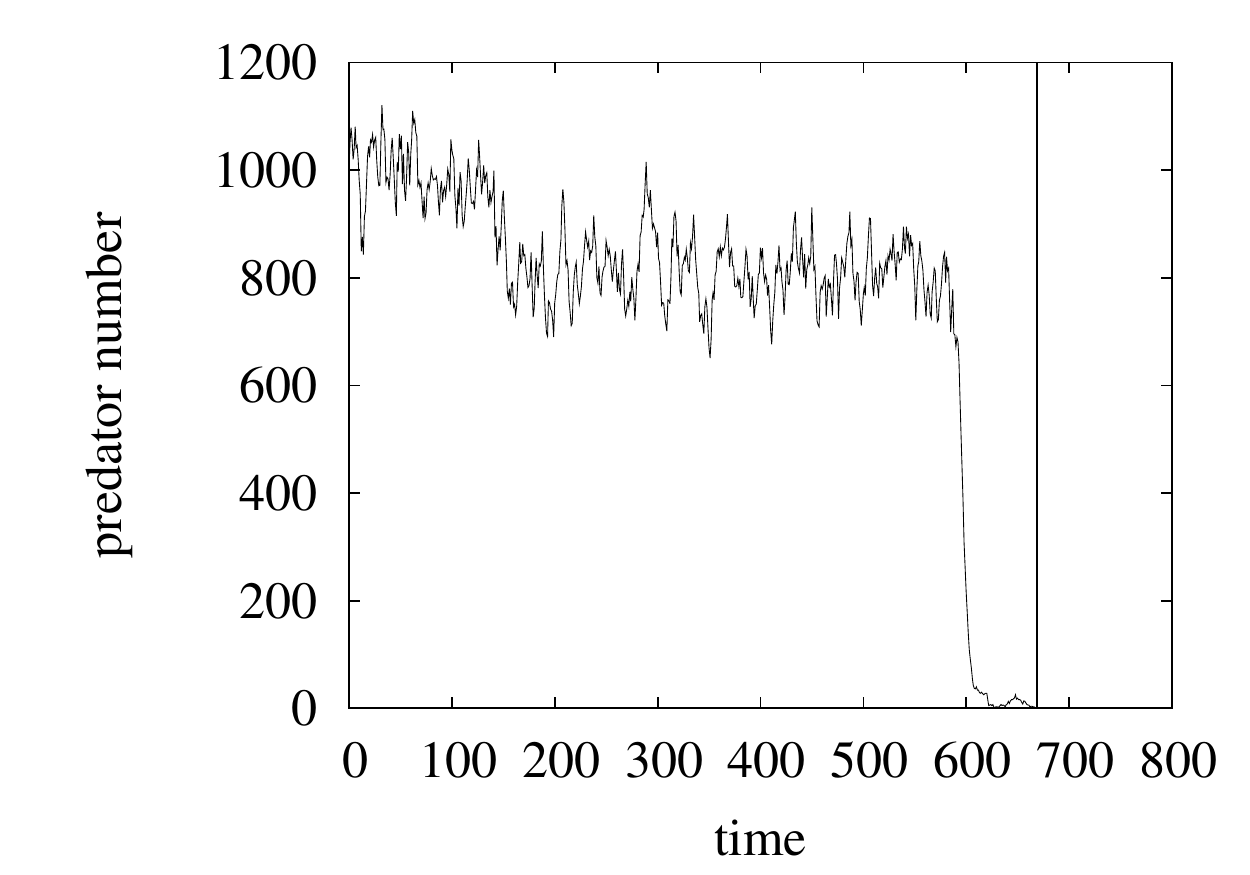}}
 \label{fig:evolx1-pred}
 }
 \end{minipage}
 \hfill
 \begin{minipage}{0.3\linewidth}
 \subfigure[][]{
 \scalebox{0.38}{
 \includegraphics{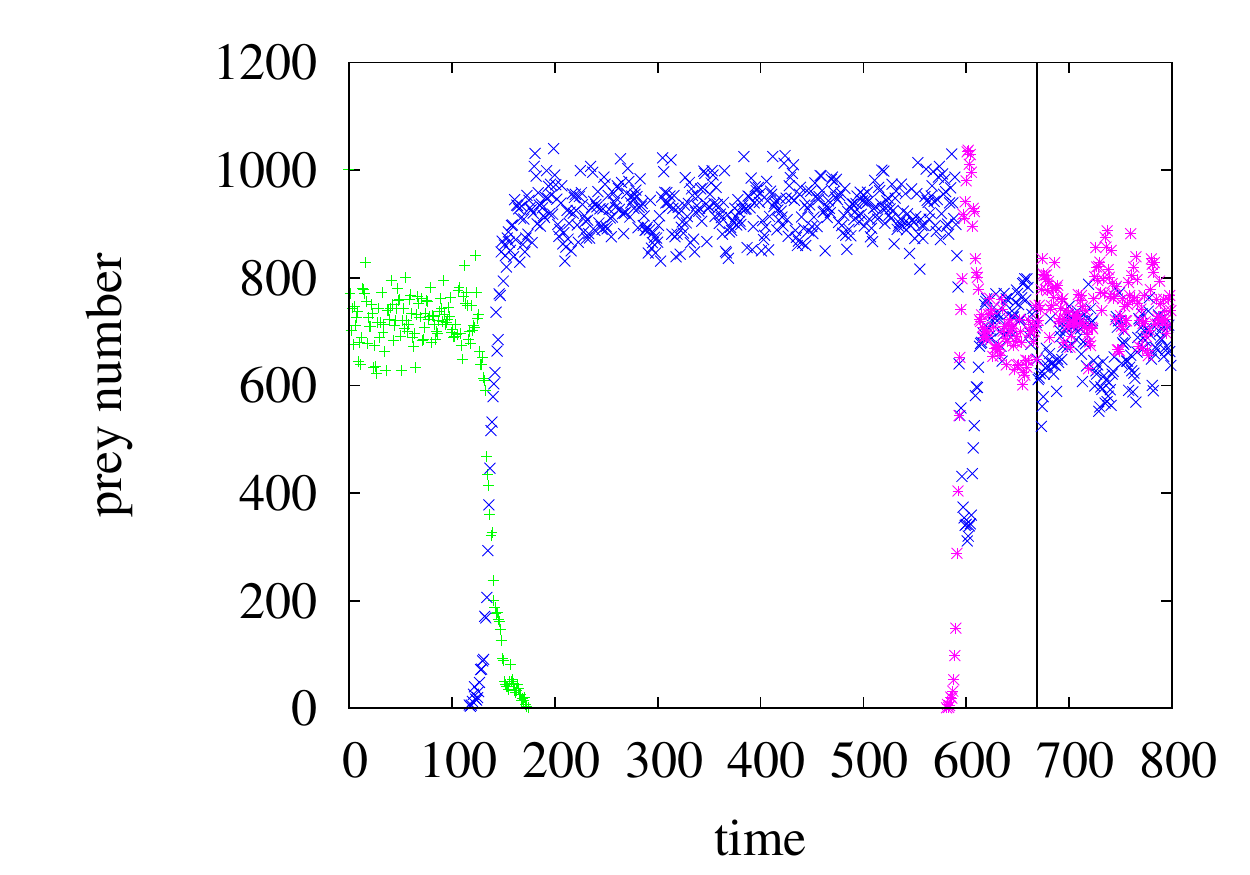}}
 \label{fig:evolx1-proie}
 }
 \end{minipage}
\caption{{Figure \subref{fig:evolx1-1} represents the traits $q_a$ (\textcolor{red}{$+$}, \textcolor{green}{$+$}, \textcolor{blue}{$+$}, \textcolor{magenta}{$+$}) and $\rho$ ($\times$) present in the community through time. Figure \subref{fig:evolx1-pred} gives the dynamics of the number of predators on the time interval $[0:800]$ and Figure \subref{fig:evolx1-proie} gives the dynamics of size of the prey populations holding trait $(0.3,0.4)$ in green, $(0.3,0.664)$ in blue and $(0.3,1.285)$ in pink (the same colors are associated on Figure \subref{fig:evolx1-1}). The vertical line corresponds to the extinction time of predators. The other parameters are $K=1000$, $u_K=5\cdot 10^{-5}$, $p=1$, $P=0$, $\pi(q_a,l)\sim \mathcal{N}(q_a,0.1)$ $b_0=2$, $d_0=0$, $c_0=1.5$, $D=0.5$, $r=0.8$, $\alpha_n=0.1$, $\beta_n=2$.}}
\label{fig:evolx1}
\end{figure}
\noindent This simulation is characteristic of the behavior of the process when the population is large and mutations are rare. As introduced by Champagnat \cite{Champagnat06} there exist two phases: a long phase where the sizes of the sub-populations remain stable, close to the equilibrium values of the deterministic system; a short phase corresponding to the invasion of a mutant trait in the resident population. The successive mutant invasions induce jumps in the traits present in the community as well as in respective sizes of each sub-population.  \red{We describe this jump process in Section \ref{subsec:mutrare}}\\

\noindent We then consider the opposite case where mutations only affect the predator preference $\rho$ and not the prey population (see Figure \ref{fig:evolx3}).
As before, Figure \ref{fig:evolx3-1} represents the traits $q_a$ and $\rho$ in the population. Figure \ref{fig:evolx3-pred} corresponds to the rescaled number of predators holding the traits $(0.2,0.6)$ in black, $(0.339,0.6)$ in green, $(0.531,0.6)$ in pink and $(0.597,0.6)$ in blue (represented with the same colors on Figure \subref{fig:evolx3-1}). The rescaled size of the prey population is drawn on Figure \ref{fig:evolx3-proie}. The initial population is composed of $K$ prey individuals with trait $(0.3,0.6)$ and $K$ predators with trait $(0.2,0.6)$.
We recall that the predator preference corresponds to the value of the qualitative defense that they can avoid or the prey type that they are specifically able to consume (see \cite{Muller04,courtois2012differences}). Predators whose preference $\rho$ is closer to the prey qualitative defense $q_a=0.6$ have an advantage in terms of relative fitness. We observe that the predator population remains monomorphic and that the trait jumps closer to $q_a$ accordingly to the successive invasions of mutants. At each invasion, the sizes of the prey and predator populations jump to the stable equilibrium of the associated $LVP$ system. The last invasion phase is very slow (see Figure \ref{fig:evolx3-pred}). It is due to a very slow convergence toward the equilibrium, of the solutions to the $LVP$ system associated with the traits $x=(0.3,0.6)$, $y_1=(0.531,0.6)$ and $y_2=(0.597,0.6)$. We observe in a general manner that the invasion times of successive mutations increase as $\rho$ comes closer to $q_a$. This reflects the flatenning of the fitness landscape for predators: through time, advantageous mutations become less beneficial with respect to the resident population.\\

 \begin{figure}[h!]
 \begin{minipage}{0.3\linewidth}
 \subfigure[][]{
 \scalebox{0.38}{
 \includegraphics{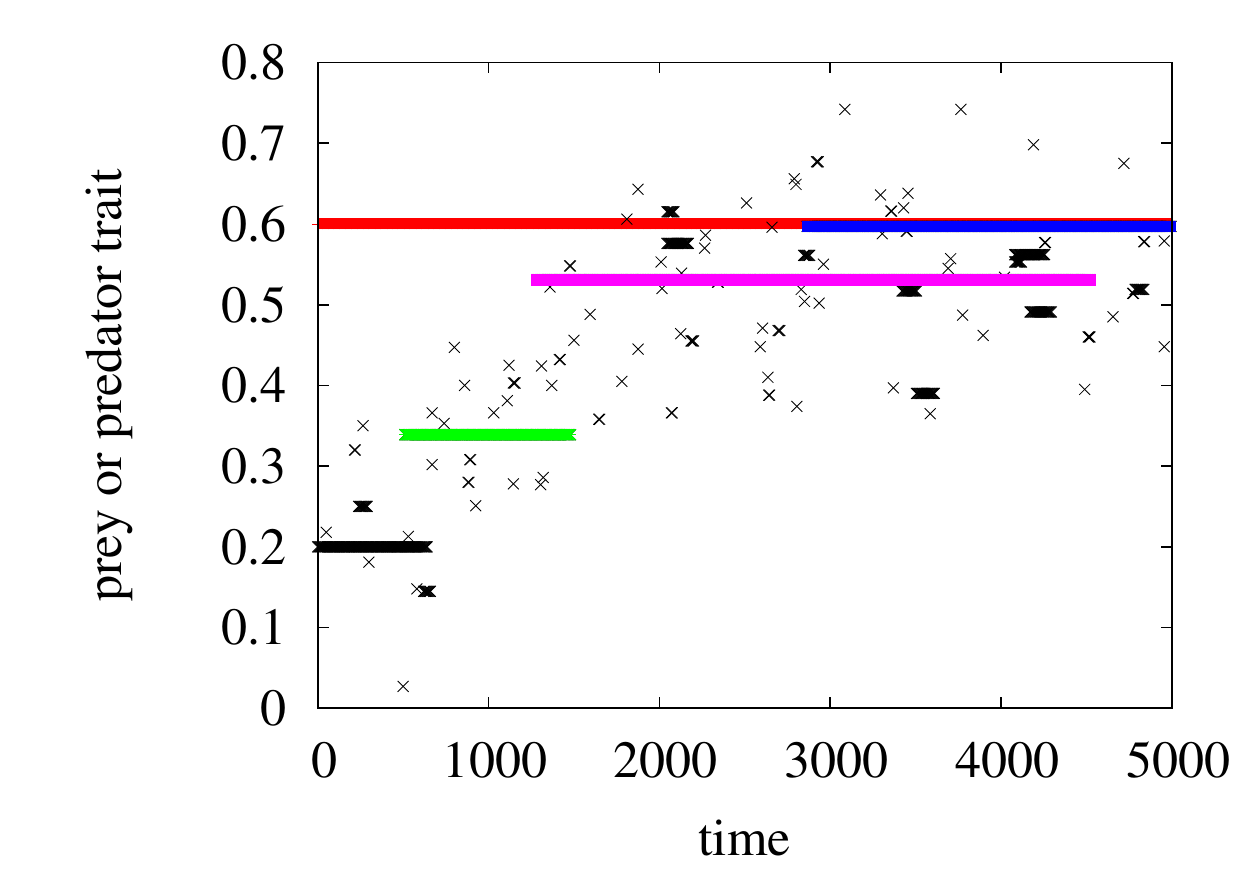}}
 \label{fig:evolx3-1}
 }
 \end{minipage}
 \hfill
 \begin{minipage}{0.3\linewidth}
 \subfigure[][]{
 \scalebox{0.38}{
 \includegraphics{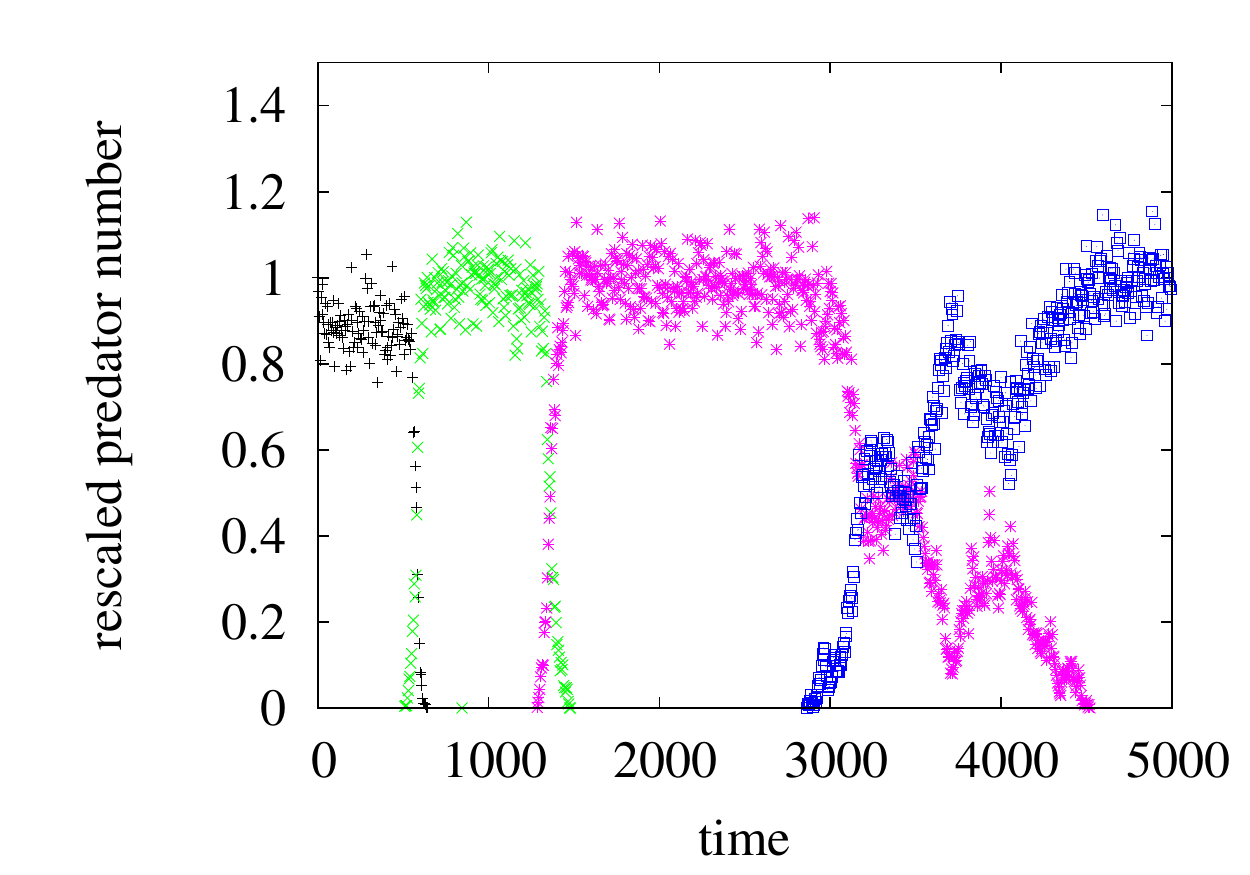}
 \label{fig:evolx3-pred}
 }}
 \end{minipage}
 \hfill
 \begin{minipage}{0.3\linewidth}
 \subfigure[][]{
 \scalebox{0.38}{
 \includegraphics{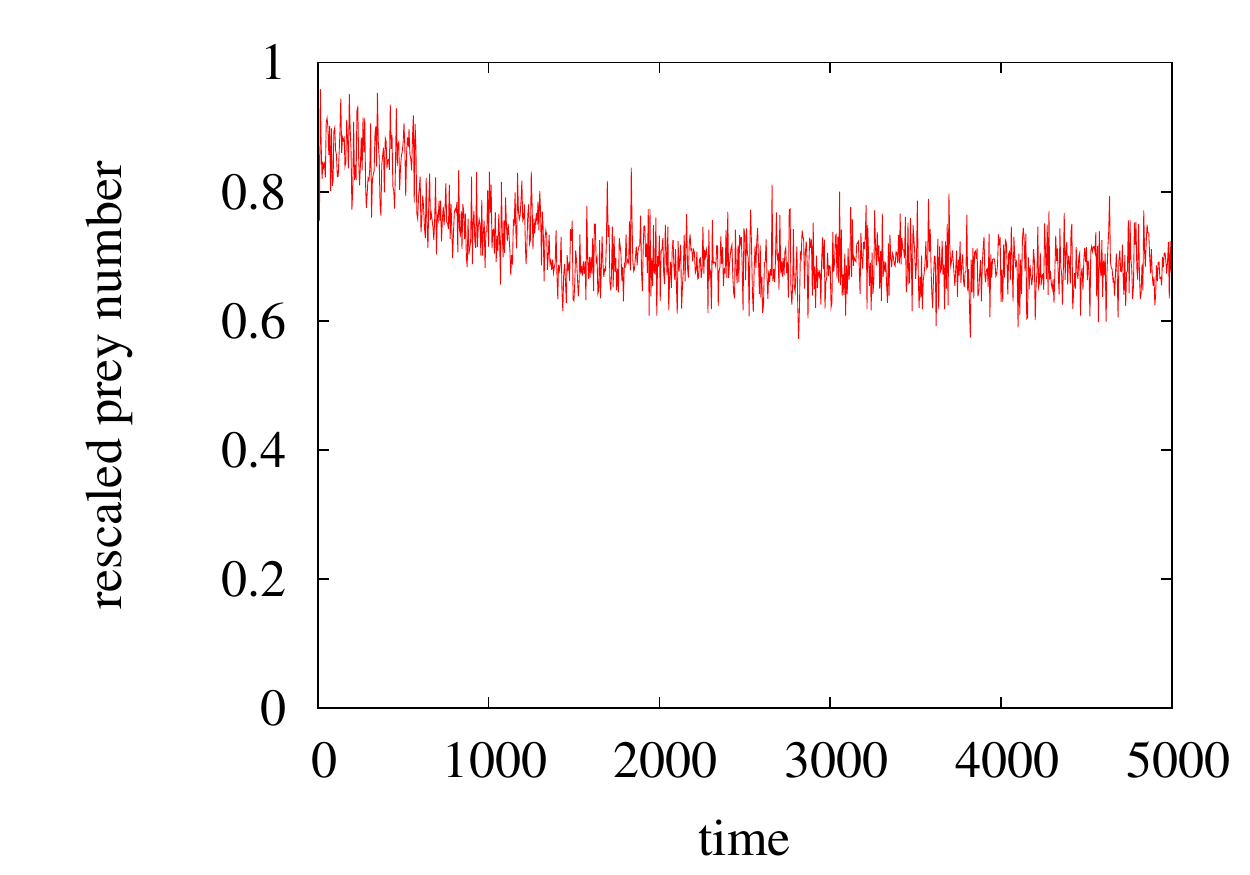}
 \label{fig:evolx3-proie}
 }}
 \end{minipage}
\caption{{Figure \subref{fig:evolx3-1} represents the traits $q_a$ (\textcolor{red}{$+$}) and $\rho$ ($\times$,\textcolor{green}{$\times$}, \textcolor{magenta}{$\times$}, \textcolor{blue}{$\times$}) present in the community through time. Figure \subref{fig:evolx3-pred} gives the dynamics of the rescaled number of predators holding trait $(0.2,0.6)$ in black, $(0.339,0.6)$ in green and $(0.531,0.6)$ in pink and $(0.597,0.6)$ in blue. Figure \subref{fig:evolx3-proie} represents the rescaled size of the prey population through time. The other parameters are $K=1000$, $u_K=1\cdot 10^{-4}$, $p=0$, $P=1$, $\Pi(\rho,l)\sim \mathcal{N}(\rho, 0.01)$, $b_0=2$, $d_0=0$, $c_0=1.5$, $D=0.5$, $r=0.8$, $\alpha_n=0.1$, $\beta_n=2$.}}
\label{fig:evolx3}
\end{figure}
\noindent To observe coevolution, we introduce mutations in both the prey and the predator populations. The prey evolution is constrained by two forces: the intra-specific competition that favours diversification and the predation pressure that drives prey phenotypes away from the predator preferences. We investigate the effect of these two forces on the community when the relative mutation speeds $p$ and $P$ vary. On Figure \ref{fig:evolx4}, we represent the traits $q_a$ (\textcolor{red}{$+$}) and $\rho$ ($\times$) present in the community through time.
On Figure \ref{fig:evolx4-2} $p=P$, we observe that the predator trait jumps close to the value of the defense of the prey population. Afterwards, the prey population becomes polymorphic. This diversity is due to the competition interaction. Finally, as predators do not adapt their preference fast enough, their population dies out. In this case, the competitive force has more impact than the predation pressure and induces a diversification of the prey phenotypes (see \cite{Loeuille02}).
On Figure \ref{fig:evolx7}, we raise the mutation probability of predators: $P=5p$ and choose smaller mutations steps. We observe two phases: in the first one (for $t\in[0:4000]$) the distance between the prey qualitative defense and preference of predators decreases. After this time, both traits seem to evolve simultaneously. This phenomenon recalls the \textit{Red Queen} or \textit{Arm races} observed by biologists (see \cite{Marrow92,Abrams00,Dercole06,becerra2009macroevolutionary}), which corresponds to a parallel variation of the traits of partner species in time.
 
 \begin{figure}[h!]
 
 \begin{minipage}{0.45\linewidth}
 \subfigure[][]{
 \scalebox{0.5}{
 \includegraphics{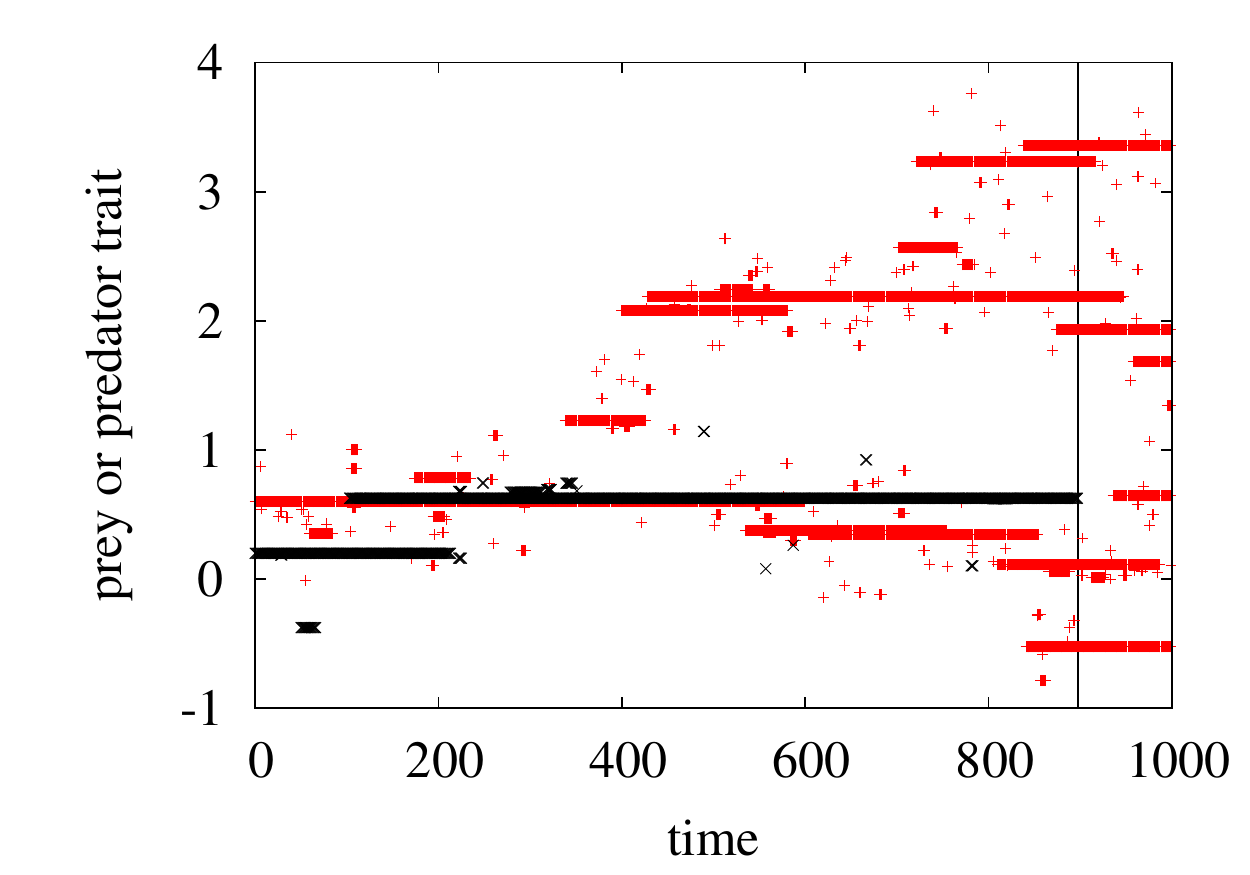}}
 \label{fig:evolx4-2}
 }
 \end{minipage}
 \hfill
 \begin{minipage}{0.45\linewidth}
 \subfigure[][]{
 \scalebox{0.5}{
 \includegraphics{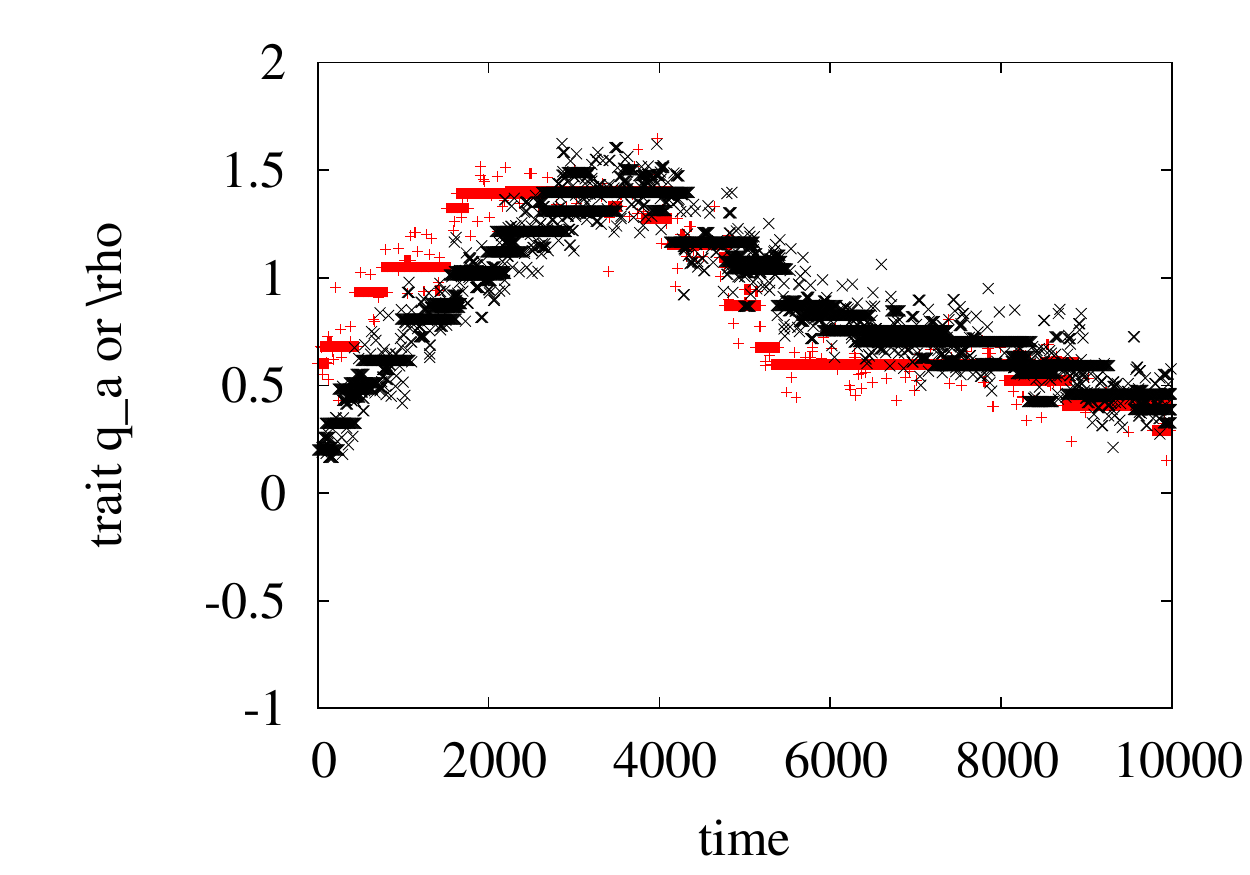}}
 \label{fig:evolx7}
 }
 \end{minipage}
\hfill
\caption{{Both Figures represent the traits $q_a$ (\textcolor{red}{$+$}) and $\rho$ ($\times$) present in the community through time. The mutation probabilities vary: on Figure \subref{fig:evolx4-2}  $p=P=1$, $\pi(q_a,\cdot)\sim \mathcal{N}(q_a,0.1)$, $\Pi(\rho,\cdot)\sim \mathcal{N}(\rho,0.1)$, and on Figure \subref{fig:evolx7} $P=5$, $p=1$ $\pi(q_a,\cdot)\sim \mathcal{N}(q_a,0.01)$, $\Pi(\rho,\cdot)\sim \mathcal{N}(\rho,0.01)$.
The other parameters are $K=1000$, $u_K=10^{-4}$, $b_0=2$, $d_0=0$, $c_0=1.5$, $D=0.5$, $r=0.8$, $\alpha_n=0.1$, $\beta_n=2$.}}
\label{fig:evolx4}
\end{figure}%
\subsubsection{Evolution of the quantitative defense}
\noindent We now model the variations in the quantity $q_n$ of quantitative defense. Unlike the qualitative defenses considered above,  quantitative defenses impact the prey birth rate and not their competitive ability. In these simulations the mutations do not affect the prey trait $q_a$ and the mutation probability of predators is null again. The initial community is composed of $K$ prey individuals holding trait $(0,0.6)$ and of $K$ predators holding trait $(0.2,0.6)$. Figure \ref{fig:evoly-1} represents the traits $q_n$ borne by prey through time. Figure \ref{fig:evoly-proie} gives the dynamics of the rescaled sizes of the prey sub-populations associated with the initial trait $(0,0.6)$ in red, $(0.189,0.6)$ in green, $(0.311,0.6)$ in blue, $(0.703,0.6)$ in pink and $(0.260,0.6)$ in light blue. These traits are represented using the same colors on Figure \ref{fig:evoly-1}. The remaining traits, in black on Figure \ref{fig:evoly-1}, correspond to mutations which did not invade the community. The dynamics of the rescaled number of predators is given on Figure \ref{fig:evoly-pred}. The vertical line corresponds to the predator extinction.\\
Note that the quantity of defense produced by prey increases in the presence of predators and that the number of predators decreases when prey increase their defenses. When prey holding trait $(0.311, 0.6)$ and $(0.703,0.6)$ coexist, the number of predators decreases quickly. We observe long time oscillations that correspond to the behavior of the dynamical systems associated to these three populations. As the competition is constant in the prey population, these simulations do not enter the mathematical framework we described (Assumption \ref{hyp:mat}). These oscillations illustrate that evolution can induce instability in the interaction networks (e.g. \cite{Loeuille10}).  After the extinction of predators, prey producing many defenses are penalized because their reproduction is weaker. The direction of natural selection changes with the extinction of the predators. 
We observe here what is called \textit{apparent competition}: the coexistence of two prey traits with predators relies on the fact that the predation pressure is stronger on the most competitive prey population (see \cite{Amstrong80}).\\
This change in the direction of evolution illustrates a new difficulty induced by coevolution: the same mutation will not have the same impact on the community depending on the presence or the absence of predators. It is thus necessary to consider the coevolution of both populations.

 \begin{figure}[h!]
 \begin{minipage}{0.3\linewidth}
 \subfigure[][]{
 \scalebox{0.38}{
 \includegraphics{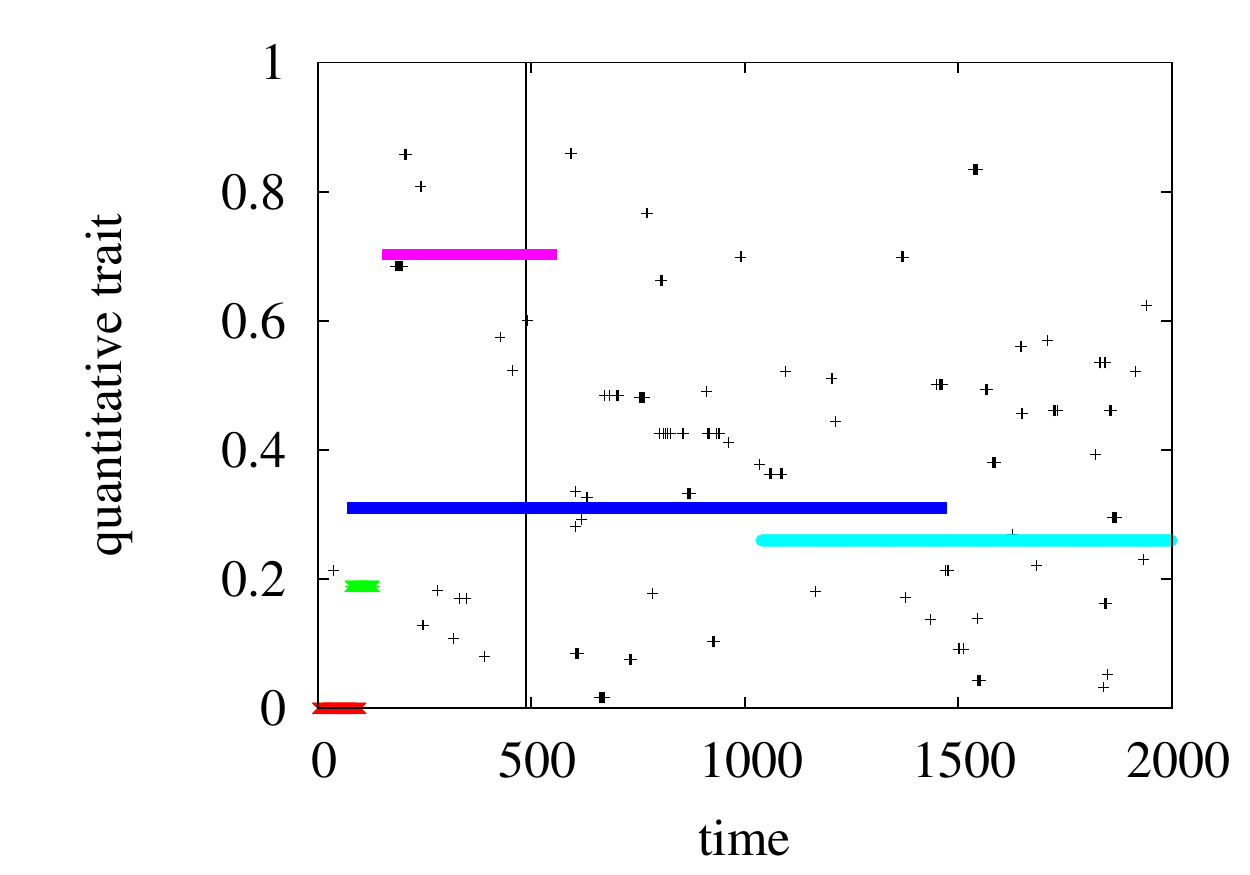}
 \label{fig:evoly-1}
 }}
 \end{minipage}
 \hfill
 \begin{minipage}{0.3\linewidth}
  \subfigure[][]{
 \scalebox{0.38}{
 \includegraphics{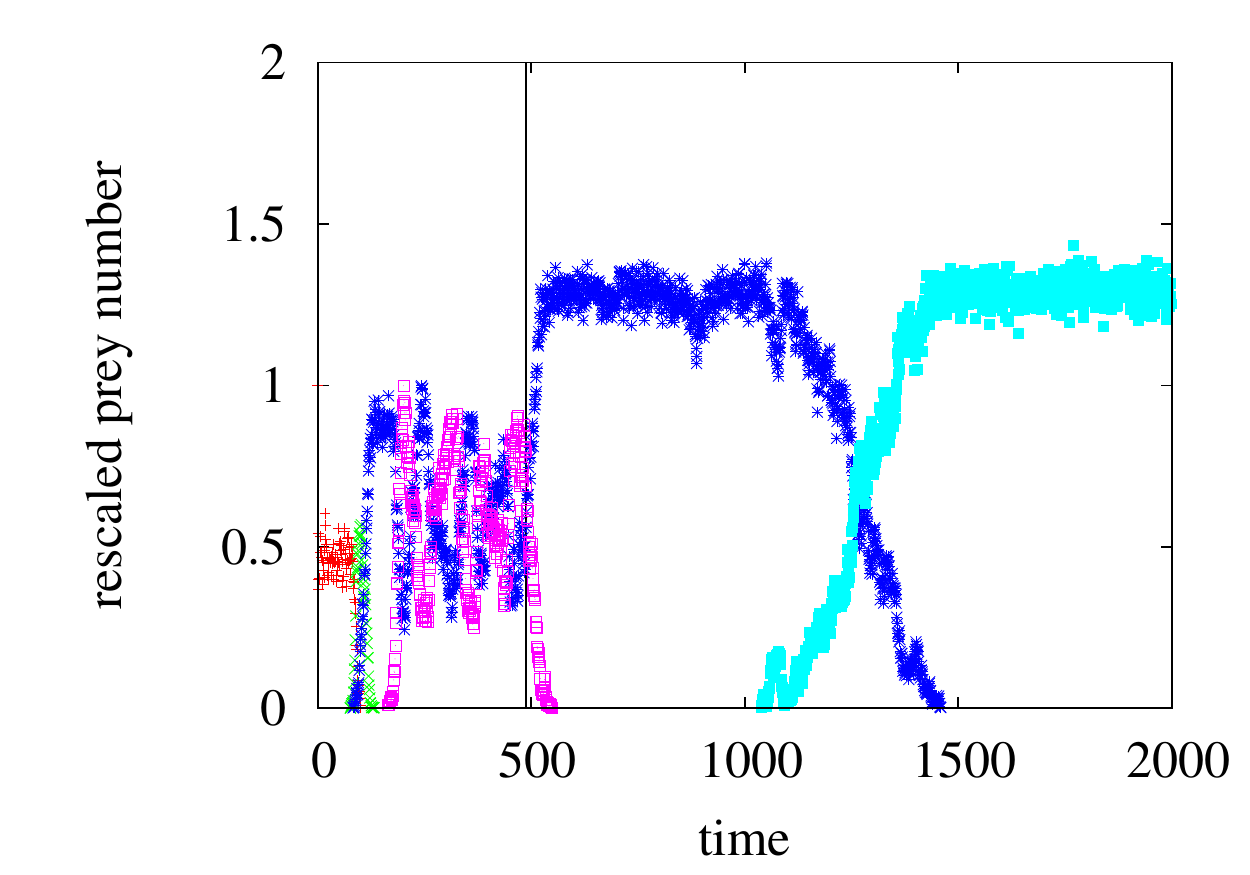}
 \label{fig:evoly-proie}
 }
 }
 \end{minipage}
 \hfill
 \begin{minipage}{0.3\linewidth}
 \subfigure[][]{
 \scalebox{0.38}{
 \includegraphics{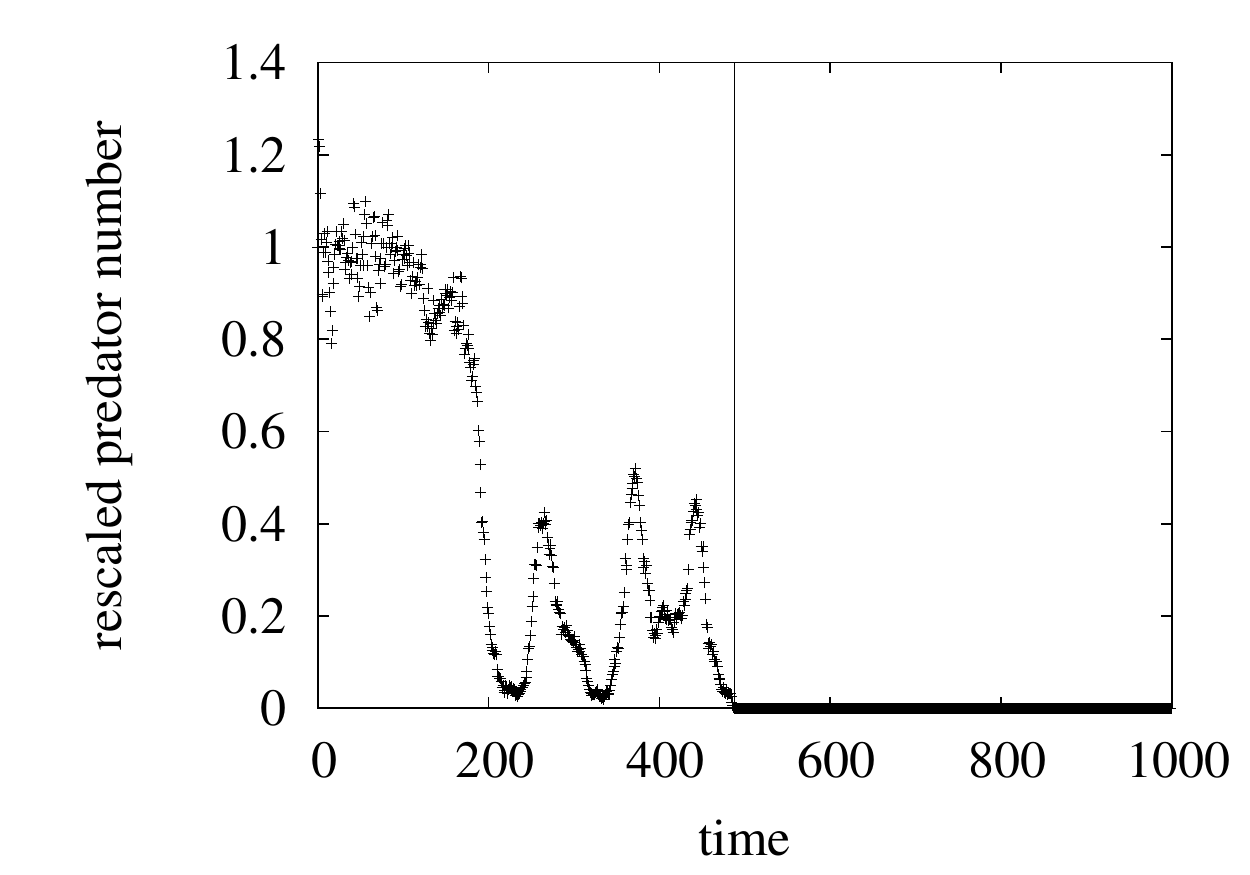}
 \label{fig:evoly-pred}
 }
 }
 \end{minipage}
\caption{{ Figure \subref{fig:evoly-1} gives the values of $q_n$ borne by prey through time. On Figure \subref{fig:evoly-proie} we draw the rescaled sizes of the prey population holding trait $(0,0.6)$ in red, $(0.189,0.6)$ in green, $(0.311,0.6)$ in blue, $(0.703,0.6)$ in pink and $(0.260,0.6)$ in light blue. The dynamics of the rescaled number of predators is given on Figure \subref{fig:evoly-pred}. 
Other parameters are given by  $K=1000$, $u_K= 10^{-4}$, $p=1$, $P=0$, $\pi(q_n,l)\sim \mathcal{N}(q_n,0.1)$ $b_0=2$, $d_0=0$, $c_0=1.5$, $D=0.5$, $r=0.8$, $\alpha_n=0.1$, $\beta_n=2$.}}
\label{fig:evoly}
\end{figure}
\subsection{Limit in the rare mutation time scale and jump process}
%\subsection{Rare mutation framework}
\label{subsec:mutrare}
We consider the limit of the community process in a large population scaling with rare mutations. The number of traits present in the community varies when mutations appear in the community. We represent the community by a couple of empirical measures $(\nu^K(t), \eta^K(t))$:
$$
\nu^K(t)=\frac{1}{K}\sum_{i=1}^{N^K(t)}\delta_{x_i},
\quad \eta^K(t)=\frac{1}{K}\sum_{l=1}^{H^K(t)}\delta_{y_l},
$$
where $\delta_x$ is the Dirac measure at point $x$. This process takes values in the set $\mathcal{M}_F(\Xcal)\times\mathcal{M}_F(\Ycal)$ of couples of finite measures on $\Xcal$ and $ \Ycal$ respectively.\\
We recall that the mutation frequencies in both populations are scaled by a parameter $u_K$ such that $Ku_K\to 0$. \red{This assumption is consistent with the adaptive dynamics framework in which mutations occur when the resident population is at equilibrium (\cite{Metz92},\cite{DieckmannLaw96}). Further assumptions will be given in Theorem \ref{thm:PES} on the exact scaling of the mutation frequency.\\
The fact that the mutation frequency decreases with the population size is not unexpected, considering population genetics arguments.
Indeed, the genetic variation among a population increases with respect to the number of individuals. However, assuming that mutant effects are distributed around 0 with a given variance, large numbers of mutations in large populations eventually produce very similar mutants. Due to this redundancy, the amount of variation produced by the mutation process saturates in large populations (see \cite{frankham1996relationship},\cite{soule1976allozyme},\cite{leimu2006general}). This saturation can be interpreted as a decrease of the outcome of new mutants.\\
}

% The mutational scale $t\to t/Ku_K$ is then a long time scale.
\noindent The next proposition states that mutations cannot occur in a bounded time interval. 
\begin{lemma}
\label{lem:nonaccumulation}
Let us assume Assumptions \ref{hyp:existence}, \ref{hyp:moment}  and that the mutation densities satisfy:
\bean
\label{hyp:mut}
&\forall x\in\Xcal, \forall u\in\R^p,\quad
\pi(x,u)\le \bar{m}(u),\quad
 \int_{\R^p}\bar{m}(u)du<\infty,\\
&\forall y\in\Ycal, \forall v\in\R^P,\quad
\Pi(y,v)\le \bar{M}(v),\quad
 \int_{\R^P}\bar{M}(v)dv<\infty,
\eean
then for every $\delta>0$, there exists $\varepsilon>0$ such that for all $t>0$,
 $$\limsup_{K\to\infty}\Pro\Bigl(\text{a mutation occurs in }\left[\frac{t}{Ku_K},\frac{t+\varepsilon}{Ku_K}\right]\Bigr) \le\delta.
$$
\end{lemma}
\noindent The proof of this Lemma can be easily adapted from the proof of Corollary 2.2 in \cite{CJM13}. It is based on a coupling of the community process before the first mutation time with a multi-type birth and death process whose birth and death rates depend on $u_K$. This process is independent of the mutation events occurring in $(\nu^K,\eta^K)$. As these mutations occur at a rate proportional to $Ku_K$, the probability to observe a mutation in a time interval of length $\varepsilon/Ku_K$ is negligible. \\

\noindent We state the main result of this section. It describes the convergence of the community process in the mutational scale toward a pure jump process. This process extends the \textit{Polymorphic Evolutionary Sequence} introduced by Champagnat and M\'el\'eard \cite{CM11} to a prey-predator network.  \\
\noindent As we have seen in the simulations, the limiting process takes values in the set of the stable equilibria $(\nbf^*(\mathbf{x},\mathbf{y}), \hbf^*(\xbf,\ybf))$ of the deterministic system $LVP(\xbf,\ybf)$ (introduced in \eqref{dpmp} for $\xbf\in\Xcal^d$ and $\ybf\in\Ycal^m$) as long as predators survive. Remark that after the extinction of predators, the behavior of the prey population is well known (see \cite{Champagnat06,CM11}) and the limiting process takes values in the set of equilibria $\bar{\nbf}(\xbf)$ of the $LVC(\xbf)$ system defined in \eqref{dp}.\\
The process describing the successive states of the community is a Markovian jump process $\Lambda=(\Lambda^1,\Lambda^2)$ taking values in $\mathcal{E}$:
\be
\mathcal{E}=\Bigl\{ (\sum_{i=1}^dn_i\delta_{x_i},\sum_{l=1}^mh_l\delta_{y_l}); \mathbf{x}\in\Xcal^d,\mathbf{y}\in\Ycal^m, (\nbf,\hbf)\in\bigl\{(\nbf^*(\mathbf{x},\mathbf{y}), \hbf^*(\xbf,\ybf)),
(\mathbf{\bar{n}}(\mathbf{x}),0)\bigr\} \Bigr\}.
\ee
The dynamics of $\Lambda$ depends on the arrivals of mutations in the prey and the predator populations. A successful mutant invasion modifies both the prey and the predator populations (see Figures \ref{fig:evolx1}, \ref{fig:evolx3} and \ref{fig:evoly}). From any state $(\sum_{i=1}^dn^*_i(\xbf,\ybf)\delta_{x_i},\sum_{l=1}^mh^*_k(\xbf,\ybf)\delta_{y_l})$ where predators are alive
\begin{itemize}
\item
for every $j\in\{1,...,d\}$ the process jumps to the equilibrium associated with the modified vector of traits $((\xbf,x_j+u),\ybf)$ 
at infinitesimal rate:
$$
p(x_j)n^*_j(\xbf,\ybf)b(x_j)\frac{[s(x_j+u;(\nbf^*(\mathbf{x},\mathbf{y}), \hbf^*(\xbf,\ybf))]_+}{b(x_j+u)}\pi(x_j,u)du.
$$
This corresponds to the invasion of a mutant prey population with trait $x_j+u$ in the community.
\item
for every  $k\in\{1,...,m\}$ the process jumps to the equilibrium associated with the modified vector of traits $(\xbf,(\ybf,y_k+v))$ 
at infinitesimal rate:
$$
P(y_k)h^*_k(\xbf,\ybf)\left(\sum_{i=1}^drB(x_i,y_k)n^*_i(\xbf,\ybf)\right)
\frac{[F(y_k+v;(\nbf^*(\mathbf{x},\mathbf{y}), \hbf^*(\xbf,\ybf))]_+}{\sum_{i=1}^drB(x_i,y_k+v)n^*_i(\xbf,\ybf)}\Pi(y_k,v)dv.
$$
This corresponds to the invasion of a  predator population holding the mutant trait $y_k+v$.
\end{itemize}
We recall that the fitness functions $s$ and $F$ are defined in \eqref{fit-Wp} and \eqref{fit-WP} respectively.
\begin{remark} As in Figures \ref{fig:evolx1} and \ref{fig:evoly}, the community jump process $\Lambda$ can reach a state where the predator population dies out. Since the invasion of mutant predators requires the positivity of their invasion fitness (see Theorem \ref{thm:exis-uni-eq}), the predator extinction can only result from the invasion of a mutant prey which diminishes the growth rate of the resident predator. The behavior of the community after the predator extinction is described by the PES introduced in Theorem 2.7 \cite{CM11}. We recall that the infinitesimal jump rate from a state $(\sum_{i=1}^d\bar{n}_i(\xbf,\ybf)\delta_{x_i},0)$  to
$(\sum_{i=1}^d\bar{n}_i((\xbf,x_j+u))\delta_{x_i}+\bar{n}_{d+1}((\xbf,x_j+u))\delta_{x_j+u}),0),
$
is given by
$$
p(x_j)b(x_j)\bar{n}_j(\xbf)\frac{[s(x_j+u;(\bar{\nbf}(\xbf),0)]_+}{b(x_j+u)}\pi(x_j,u)du.
$$
\end{remark}
We now formulate the limiting theorem.
\begin{thm}
\label{thm:PES}
Fix $\xbf\in\Xcal^d$ and $\ybf \in\Ycal^m$.
Let us assume Assumptions \ref{hyp:existence}, \ref{hyp:moment}, \ref{hyp:ODE}, \ref{hyp:fitness}, %\eqref{hyp:mutrare}, 
\eqref{hyp:mut} and that the initial condition $(\sum_{i=1}^{d}n^K_i\delta_{x_i},\sum_{l=1}^m h^K_l\delta_{y_l})$ converges in probability toward $ (\sum_{i=1}^{d}n^*_i\delta_{x_i},\sum_{l=1}^m h^*_l\delta_{y_l})$.  If furthermore
\ben
\label{hyp:mutrare}
\log(K)\ll \frac{1}{Ku_K} \ll \exp(VK), \quad \forall V>0,
\een
then the process $\bigl(\nu^K(\frac{t}{Ku_K}),\eta^K(\frac{t}{Ku_K})\bigr)_{t\ge0}$ converges toward the pure jump process
$\Lambda=((\Lambda^1_t,\Lambda^2_t);t\ge0)$ defined above and whose initial condition is given by $(\sum_{i=1}^dn^*_i(\xbf,\ybf)\delta_{x_i},\sum_{l=1}^mh^*_l(\xbf,\ybf)\delta_{y_l}).$\\
This convergence takes place in the sense of convergence of the finite dimensional distributions for the topology on $\mathcal{M}_F(\Xcal\times\Ycal)$ induced by the total variation norm.
\end{thm}
Assumption \eqref{hyp:mutrare}, introduced by Champagnat \cite{Champagnat06}, reflects the separation between the demographic and the mutational time scales (see Figure \ref{fig:evolx1-pred}, \ref{fig:evolx1-proie}, \ref{fig:evolx3-pred} and \ref{fig:evolx3-proie}). The demographic time scale is of order $\log K$. It corresponds to the evolution of the stochastic process close to its deterministic approximation. The process $\Zbf^K$ enters a neighbourhood of the attractive deterministic equilibrium and the deleterious traits die out (Proposition~\ref{prop:cveq} and Proposition~\ref{prop:extinction}). The mean time between two mutations is of order $1/Ku_K$, therefore the resident population is close to the equilibrium of the associated $LVP$ system when a mutant appears in the community (Theorem \ref{thm:sortie_eq}). \\
The proof derives from the proof of Theorem 1 in \cite{Champagnat06} and from the results obtained in Section \ref{sec:sortie-eq}. The main idea is to study the invasion of a mutant trait in the community. Starting from an initial condition  $(\sum_{i=1}^dn^*_i(\xbf,\ybf)\delta_{x_i},\sum_{l=1}^mh^*_l(\xbf,\ybf)\delta_{y_l})$ at the deterministic equilibrium, the next mutation occurs after an exponential time of parameter
$$
E(\xbf,\ybf)=\sum_{i=1}^d p(x_i)b(x_i)n^*_i(\xbf,\ybf)+\sum_{l=1}^mP(y_l)h^*_l(\xbf,\ybf)\Bigl(r\sum_{i=1}^d B(x_i,y_l)n^*_i(\xbf,\ybf)\Bigr).
$$
The mutant individual then comes from the prey population with trait $x_j$ ($1\le j \le d$) with probability
$$
\frac{p(x_j)b(x_j)n^*_j(\xbf,\ybf)}{E(\xbf,\ybf)},
$$
or from the population of predators holding trait $y_k$ ($1\le k\le m$) with probability
$$
\frac{P(y_k)h^*_k(\xbf,\ybf)\sum_{i=1}^drB(x_i,y_k)n^*_i(\xbf,\ybf)}{E(\xbf,\ybf)}.
$$
In the sequel we consider a mutant trait $y_k+v$ where $v$ is distributed according to $\Pi(y_k,v)dv$. While the number of individuals holding the mutant trait is small, we compare thanks to Theorem \ref{thm:sortie_eq} the size of the mutant population with a continuous time birth and death process with birth rate $r\sum_{i=1}^d B(x_i,y_k+v)n^*_i(\xbf,\ybf)$ and death rate $D(y_k+v)$. Its growth rate is then given by the invasion fitness $F(y_k+v; (\nbf^*(\xbf,\ybf),\hbf^*(\xbf,\ybf)))$ of the mutant trait in the resident population. If the fitness is negative we prove that the mutant population goes extinct similarly as in Lemma \ref{prop:extinction}. Otherwise, the probability that the mutant population reaches a positive density $\varepsilon$ is close to the survival probability of the supercritical branching process which is given by
$$\frac{F(y_k+v; (\nbf^*(\xbf,\ybf),\hbf^*(\xbf,\ybf)))}{r\sum_{i=1}^d B(x_i,y_k+v)n^*_i(\xbf,\ybf)},$$(see \cite{AthreyaNey}, p102). Moreover this phase lasts a time of order $\log K$ (see the proof of Lemma 3 in \cite{Champagnat06}). 
Then using the large population approximation on a finite time interval (Proposition \ref{prop:cventempsfini}), we establish that the process $\Lambda$ jumps to the equilibrium of the system $LVP(\xbf,(\ybf,y_k+v))$.

\subsection{Small mutations: a canonical equations system for coevolution}
\label{subsec:eqcano}
In this subsection we consider a different scaling for the jump process $\Lambda$ where the mutation steps of both populations are of order $\varepsilon$ (see \cite{CM11,CJM13}). We study the limit of the sequence $\Lambda_{\varepsilon}$ in a long time scale $\frac{t}{\varepsilon^2}$ to observe global evolutionary dynamics. We establish that the limiting behavior of the prey and predator traits satisfies a coupled system of differential equations. These equations were heuristically introduced by  Dieckmann, Law and Marrow (1996) \cite{Marrow96}. They extend the \textit{canonical equation of adaptive dynamics} to the coevolution of a prey and predator interaction.\\
In the sequel we assume that every couple of a prey and a predator trait can coexist although two  prey traits cannot coexist. 

\begin{Hyp}
\label{hyp:eqcan}
\begin{enumerate}[a)]
\item For every $(x,y)\in\Xcal\times\Ycal$, predators survive in the equilibrium of $LVP(x,y)$:
\ben
\label{survie}
\frac{b(x)-d(x)}{c(x,x)}>\frac{D(y)}{rB(x,y)}.
\een
\item 
\label{IIF}
Invasion implies fixation: For every $(x,x')\in\Xcal^2$ and $y\in\Ycal$, we have 
\bea
& s(x';(n^*(x,y),h^*(x,y)))<0,\\ 
 \text{or } &s(x';(n^*(x,y),h^*(x,y)))>0\text{ and }s(x;(n^*(x',y),h^*(x',y)))<0.
\eea
 \item The mutation densities $\pi$ and $\Pi$ are Lipschitz continuous on $\Xcal\times\R^p$ and $\Ycal\times\R^P$.
\item The functions $g$ and $G$ defined for $x,x'\in\Xcal$ and $y,y'\in\Ycal$ by
\bean
\label{g_p}
&g(x';(x,y))=p(x)n^*(x,y)b(x)\frac{s(x';(x,y))}{b(x')},\\
&G(y';(x,y))=P(y)h^*(x,y)B(x,y)\frac{F(y';(x,y))}{B(x,y')},
\eean
are continuous and $\mathcal{C}^1$ with respect to their first variable.
\end{enumerate}
\end{Hyp}
\begin{remark}Condition \eqref{survie} compares the equilibrium sizes of the prey populations evolving in the presence or the absence of predators. The predator survival requires that the prey population size decreases in the presence of predators.
\end{remark}
\noindent For every couple of traits $(x,y)$, the equilibrium $(n^*(x,y),h^*(x,y))$ of the system $LVP(x,y)$ given by Theorem \ref{thm:exis-uni-eq} equals
\ben
\label{nheq}
n^*(x,y)=\frac{D(y)}{rB(x,y)}  \text{ and } h^*(x,y)=\frac{1}{B(x,y)}\bigl(b(x)-d(x)-c(x,x)\frac{D(y)}{rB(x,y)}\bigr).
\een
To ease notation, we denote in this section $s(x';(x,y))=s(x';(n^*(x,y),h^*(x,y)))$ and $F(y';(x,y))=F(y';(n^*(x,y),h^*(x,y)))$.\\
Assumption \ref{hyp:eqcan}.\ref{IIF}) and Theorem \ref{thm:exis-uni-eq} entail that two prey types cannot coexist in the equilibrium of the deterministic system $LVP$. Together with the competitive exclusion principle introduced in Section \ref{sec:ODE}, this ensures that two predator populations cannot coexist either. Therefore each mutant invasion  (prey or predator) leads to the replacement of the resident trait. The community is then always composed of a monomorphic prey population and a monomorphic predator population:
\ben
\label{PEStrait}
\Lambda^1_{\varepsilon}(t)=n^*(X_{\varepsilon}(t),Y_{\varepsilon}(t))\delta_{X_{\varepsilon}(t)},\quad \Lambda^2_{\varepsilon}(t)=h^*(X_{\varepsilon}(t),Y_{\varepsilon}(t))\delta_{Y_{\varepsilon}(t)}.
\een
The trait process $(X_{\varepsilon}(t),Y_{\varepsilon}(t))$ is a Markovian jump process taking values in $\Xcal\times\Ycal$ whose infinitesimal generator is given for any measurable bounded function $\phi$ by
\bea
\label{gen:trait}
\mathcal{L}_{\varepsilon}\Phi(x,y)=&
\int_{\Xcal} \Bigl(\Phi(x+\varepsilon u,y)-\Phi(x,y)\Bigr) [g(x+\varepsilon u;(x,y))]_+\pi(x,u)du\\
&+\int_{\Ycal} \Bigl(\Phi(x,y+\varepsilon v)-\Phi(x,y)\Bigr)[G(y+\varepsilon v;(x,y))]_+\Pi(y,v)dv.
\eea

\noindent The following Theorem states the limiting behavior of the process $(X(t/\varepsilon^2),Y(t/\varepsilon^2))$ as $\varepsilon$ goes to $0$. The proof relies on a classical compactness-uniqueness argument that can be immediately extended from \cite{CM11} (Appendix C.).
\begin{thm}
\label{thm:eqcan}
Let us assume Assumptions \ref{hyp:existence}, \ref{hyp:moment}, \ref{hyp:ODE}, \ref{hyp:fitness}, \ref{hyp:eqcan} and that the sequence of initial conditions $(X_{\varepsilon}(0),Y_{\varepsilon}(0))$ is bounded in $\mathbb{L}^2$ and converges in law toward a deterministic vector $(x_0,y_0)$, then for every $T>0$ the process $(X(t/\varepsilon^2),Y(t/\varepsilon^2))$ converges in law in $\mathbb{D}([0,T],\Xcal\times\Ycal)$ toward a couple of deterministic functions $(x(t),y(t))_{t\in[0,T]}$ unique solution of the system of differential equations
\ben
\label{eqcan}
 \left\{
\begin{aligned}
&\frac{d}{dt}x(t)=\int_{R^p} u[u \cdot \nabla_1 g(x(t);(x(t),y(t)))]_+ \pi(x(t),u) du,
\\
&\frac{d}{dt}y(t)=\int_{R^P} v[v \cdot \nabla_1 G(y(t);(x(t),y(t)))]_+ \Pi(y(t),v) dv,
\end{aligned}
\right.
\een
with initial condition $(x_0,y_0)$.
\end{thm}
\noindent This system is strongly coupled through the functions $g$ and $G$.\\
\noindent In the specific case where the mutation measures $\pi$ and $\Pi$ are symmetrical, with covariance matrices $\gamma$ and $\Gamma$, the system \eqref{eqcan} becomes
\be
\label{eqcan-sym}
\left\{
\begin{aligned}
&\frac{d}{dt}x(t)= \frac{1}{2}p(x)\gamma(x)n^*(x,y) \nabla_1 s(x;(x,y)) ,
\\
&\frac{d}{dt}y(t)=\frac{1}{2}  P(y)\Gamma(y) h^*(x,y)\nabla_1 F(y;(x,y)).
\end{aligned}
\right.
\ee
\begin{remark}In the large population limit with rare and small mutations, diversification events of the population can be observed. These evolutionary branching are well understood in the case of the evolution of a single population (see \cite{CM11,CJM13}). They rely on the behavior of the jump process $\Lambda$ when coexistence of two traits occurs.
The prey-predator coevolution make the evolutionary branching properties of the trait processes complex to study. In particular, if two prey traits coexist, the next mutation can lead to the coexistence of two predator traits as well.
 \end{remark}

\subsubsection{Application}
We apply those results to the example introduced in Section \ref{subsec:example} where prey individuals hold a trait $(q_n,q_a)\in\R\times\R_+$ and predators a trait $(\rho,\sigma)\in\R\times\R_+$. We recall that the rate functions are given in \eqref{para} and that the mutation measures are gaussian with respective variance $\gamma$ and $\Gamma$.\\
Derivating the fitness functions with respect to the mutant trait, we obtain 
\bean
\label{grad-W}
\nabla_1s((q_n,q_a);(q_n,q_a,\rho,\sigma))=&
\left(
\begin{array}{c}
\bigl( -\alpha_n b_0\exp(-\alpha_n q_n)+\beta_n\frac{h^*(q_n,q_a,\rho,\sigma)}{B(q_n,q_a,\rho,\sigma)}\bigr) \un_{q_n>0}\\
\frac{q_a-\rho}{\sigma^2}\frac{h^*(q_n,q_a,\rho,\sigma)}{B(q_n,q_a,\rho,\sigma)}
\end{array}
 \right).
\eean
and
\ben
\label{grad-Wp}
\nabla_1F((\rho,\sigma);(q_n,q_a,\rho,\sigma))=
\left(
\begin{array}{c}
\frac{q_a-\rho}{\sigma^2}D\\
\bigl(\frac{(q_a-\rho)^2}{\sigma^3}-\frac{1}{\sigma}\bigr)D\un_{\sigma>0}
\end{array}
 \right)
\een
In particular
$\nabla_1F((\rho,\sigma);(q_n,q_a,\rho,\sigma))=0$ if and only if $\rho=q_a$ and $\sigma= |q_a-\rho|$.\\

\noindent We first study the coevolution of the traits $q_a$ and $\rho$, the values of $\sigma$ and $q_n$ being fixed, as in Figure \ref{fig:evolx4}. The system of differential equations governing the dynamics of $q_a$ and $\rho$ is then
\bean
&\frac{d}{dt}q_a(t)=\frac{\gamma\pi(q_a)}{r} \phi(q_a,\rho)\\
&\frac{d}{dt}\rho(t)=\Gamma\Pi(\rho)\phi(q_a,\rho),
\eean
where
$$
\phi(q_a,\rho)=\frac{D(q_a-\rho)}{\sigma^2}h^*(q_a,\rho),
$$
and the equilibrium $h^*(q_a,\rho)$ is given in \eqref{nheq}.
The function $\phi$ vanishes if $q_a=\rho$ or if the predator population dies out ($h^*(q_a,\rho)=0$).\\
We deduce from the specific form of the system that for all $t\ge0$, $h^*(q_a(t),q_n(t))>0$. Moreover there exist three cases depending on the respective values of the mutation probabilities and variances and on the parameter $r$:
\begin{itemize}
 \item If $r\Gamma\Pi(\rho) >\gamma\pi(q_a)$, the difference $|q_a(t)-\rho(t)|$ decreases with time. This phenomena was observed on Figure \ref{fig:evolx7} on the first part of the graph. 
\item If $r\Gamma\Pi(\rho)=\gamma\pi(q_a)$, both derivatives are equal for all times. The evolution then follows an \textit{arm race dynamics} : both traits evolve continuously and $|q_a(t)-\rho(t)|$ remains constant (see \cite{Marrow92, Abrams00, Dercole06}).
\item If $r\Gamma\Pi(\rho)<\gamma\pi(q_a)$, prey escape the predator influence as the distance between $q_a$ and $\rho$ increases. When $t\to\infty$, the solution converges toward a vector $(q_a^*,\rho^*)$ that doesn't satisfy \eqref{survie}. However the extinction of the predator population is not possible in finite time unlike in the process $\Lambda$ (see Figure~\ref{fig:evolx4-2})
\end{itemize}

\noindent Then we consider, as in Figure \ref{fig:evoly}, the prey strategies for the quantitative defense $q_n$, when the other traits are not affected by mutations:
$$
\frac{d}{dt}q_n(t)=\pi(q_n)\gamma\frac{D}{rB(q_n)}\bigl( (\alpha_n-\beta_n)b_0\exp(-\alpha_n q_n)-\beta_n c_0\frac{D}{rB(q_n)}\bigr)\un_{q_n>0}
$$
In the case where $\alpha_n\ge \beta_n$, meaning that the allocative trade-off between producing a large quantity of defense and having a good reproduction is important, the quantitative defense $q_n(t)$ decreases to $0$.
If $\alpha_n<\beta_n$, the derivative of $q_n$ vanishes at the point 
$$q_n^*= \frac{-1}{\alpha_n+\beta_n} \ln\left(\frac{\beta_n c_0 D}{rB(0)(\beta_n-\alpha_n)b_0}\right).
$$ 
Then either $q_n^*$ is negative and again $q_n(t)\to 0$ or $q_n^*\ge0$ and $q_n(t)$ converges to $q_n^*$ when $t\to\infty$. With the parameters of Figure \ref{fig:evoly}, $q_n^*\approx0.58$. We observed first an increase of $q_n$ and then the extinction of predators. Thus, an important question is whether or not the predator population dies out as $t\to\infty$. An easy calculation gives that
$$h^*(q_n^*)>0 \iff \frac{\beta_n}{\beta_n-\alpha_n}>1,$$ which is always true if $0<\alpha_n<\beta_n$. We deduce that the evolution of the quantitative defense does not drive the predators to extinction. This prediction contradicts the extinction observed in Figure \ref{fig:evoly}. However, in this simulation Assumption E.\ref{IIF} is not satisfied and the predator extinction is due to the coexistence of two prey types.

\section{Discussion}
We introduced three different objects to describe the prey-predator community: a deterministic system $LVP$ in \eqref{dpmp}, a stochastic jump process $\Lambda$ in Section \ref{sec:MUT} and a couple of two canonical equations in \eqref{eqcan}. These processes correspond to three different limits of the individual based process introduced in Section \ref{sec:Model}.\\
The jump process $\Lambda$ describes the dynamics of the community when mutations are rare. It describes the successive equilibria of the community. In this sense, it justifies a simulation method developped in Ecology to study the phenotypic evolution of communities (see \cite{Loeuille05, Loeuille08b,brannstrom10}). In these articles, the community evolves as the solution of a system of differential equations. Each equation of the system describes the dynamics of a sub-population. When a mutation occurs (at a very low rate), it increases the number of sub-populations and thus a new equation is added to the system. Their method gives the successive equilibria of the community similarly to the jump process $\Lambda$, however, it does not take into account the demographic stochasticity as every mutant with a positive fitness invades the community. \\
Our model highlights the implications of coevolutionary dynamics for the ecological dynamics of the community and its maintenance in time (see Section \ref{subsec:simu}). Particularly, we show that such consequences depend on the trait under scrutinity and on the costs that are associated to these traits. For instance, the two categories of defenses have different implications in this regard.
If the evolution of qualitative defenses is fast enough, it can lead to the disappearance of the predators as in Figure \ref{fig:evolx1} and \ref{fig:evolx4}, a phenomenon called "evolutionary murders" (as the evolution of a species in the community eventually kills another species). We note that such evolutionary murders do not happen when one considers the evolution of quantitative defenses. Likewise, the evolution of predators does not lead to the extinction of prey. Therefore, our study highlights how evolutionary murder phenomenons, already known in ecology (\cite{brannstrom2012modelling, georgelin14, Dercole06})
depend evolving species and types of traits that evolve.
Even in the absence of species extinction, we note that the coevolution also modifies the strength of the interactions between species and can thus lead to the reinforcement of an interaction (see Figure \ref{fig:evolx3}). Or as observed in Figure \ref{fig:evolx1}, evolution can induce the disappearance of an interaction (through diminishing the competition between two plants).
Interactions then progressively weaken and become "ghosts from the past", as commonly observed in phylogenetic or evolutionary studies (e.g. \cite{tobias2013species,bennett2013increased}).\\
Such variations in interaction strength can have important consequences for the overall stability of the system. Indeed, in food webs, stability analyses suggest that distributions of interaction strengths including weak interactions have a stabilizing effect on the dynamics of the community \cite{mccann1998weak} with important implications for the conservation of species and for the delivery of ecosystem services. 
The question of the links between evolution and stability of the network is therefore crucial. As shown in Figure \ref{fig:evoly}, evolution can induce instability in the network so that small perturbations of a population may lead to the extinction of one or several populations (cf \cite{Loeuille10}). \\
The jump process contains the various behaviors present in ecological communities however we only have little predictive information on the composition of the community at all times. Therefore it can be interesting to consider the canonical system \eqref{eqcan}. This process represents the dynamics of the traits present in the community under strong assumptions on the small size of the mutation steps and on the non-coexistence of different traits of prey and predators. 
The strong influence of prey on predators and vice versa can be well understood when we consider the equilibria of this system.
We only consider one-dimensional traits ($p=P=1$).  If we consider the specific case of an equilibrium $(x^*,y^*)$ such that 
\ben
\label{equilibre-can}
\partial_1s(x^*;(x^*,y^*))=0 \text{ and }\partial_1F(y^*;(x^*,y^*))=0 
\een
This equilibrium corresponds to a two-dimensional version of the \textit{Evolutionary strategies} introduced in \cite{metz1996adaptive} for the one-dimensional canonical equation. A natural question about this equilibrium is a condition for its stability. The Jacobian matrix at a point $(x,y)$ is given by:
\ben
\label{jac-can}
\left( \begin{array}{cc}
n^*(x,y)(\partial_{11}s(x;(x,y))+\partial_{12}s(x;(x,y))) & n^*(x,y)(\partial_{13}s(x;(x,y))\\
h^*(x,y)\partial_{12}F(y;(x,y)) &h^*(x,y)(\partial_{11}F(y;(x,y))) +\partial_{13}F(y;(x,y)) )\\
\end{array}\right)
\een
Note that the conditions
$$\partial_{11}s(x^*;(x^*,y^*))+\partial_{12}s(x^*;(x^*,y^*)<0 ,$$
and
$$ \partial_{11}F(y^*;(x^*,y^*))) +\partial_{13}F(y^*;(x^*,y^*)) <0,
$$are not sufficient nor necessary to ensure the stability of the equilibrium $(x^*,y^*)$. These two conditions correspond to the local stability of the equilibrium $x^*$ when we consider the evolution of the prey trait in the presence of a fixed predator trait $y^*$ and conversely for the evolution of the predator trait in the presence of prey individuals holding the fixed trait $x^*$ (see \cite{CM11}). \\
The branching properties of the community are complex to study. Indeed, they rely on a precise study of the jump process $\Lambda$ after the first coexistence of two traits. As we have seen in the simulations, the coexistence in the prey population can lead to extinction of predators  (see Figure \ref{fig:evolx1} and \ref{fig:evoly}), or to the coexistence of different trait of predators (see Figure \ref{fig:evolx7}).

Throughout this work we considered the same time scales for both prey and predators. Note that while this hypothesis of similar evolutionary time scales allows a first grasp on the effects of coevolution on the ecological dynamics of such interactions, strong asymmetries actually occur in nature. Taking again the example of plant-herbivore interactions, large asymmetries of demographic and evolutionary time scales can arise when the two partners have large differences in terms of body size and generation time (eg, tree-insect interactions such as \cite{Umea} or, at the other extreme, grass-large herbivore interactions \cite{bakker2006herbivore}). 
We will consider such asymmetries of time scales in a future work.

\appendix
\red{
\section{Construction of a trajectory of the prey-predator community process}
\label{app:poisson}
We construct a trajectory of the prey-predator community process as solution of a system stochastic differential equations driven by Poisson point measures (see \cite{FM04},\cite{CFM08}). We introduce two families of independent Poisson point measures on $(\R_+)^2$ with intensity $dsd\theta$: $(R_j)_{1\le j\le d+m}$ for the prey and predators reproduction events and $(M_j)_{1\le j\le d+m}$ for the death events. Then, $\forall 1\le i\le d$ and $\forall 1\le l\le m$
\bean
\label{pop-poisson}
N^K_i(t)&=N^K_i(0) +\int_0^t\int_{\R_+} 
  \un_{\theta\le b(x_i)N^K_i(s-)}
  R_i(ds,d\theta) 
\\
&-\int_0^t\int_{\R_+}
  \un_{\theta\le  \lambda(x_i,\Zbf(s-))N^K_i(s-) }
  M_i(ds, d\theta) ,\\
H^K_l(t)&=H^K_l(0) +  \int_0^t\int_{\R_+} 
  \un_{\theta\le rH^K_l(s-)\bigl(\sum_{i=1}^d \frac{B(x_i,y_l)}{K}N^K_i(s-) \bigr)}
  R_{d+l}(ds,d\theta) \\
&-\int_0^t\int_{\R_+}
  \un_{\theta\le  D(y_l)H^K_l(s-) }
  M_{d+l}(ds, d\theta).
\eean
Let us explain briefly these equations. We focus on the prey population $N^K_i$ with trait $x_i$. A trajectory is constructed using two Poisson point measures $R_i$ and $M_i$. The measure $R_i$ handles the reproduction events and $M_i$ the death events. A Poisson point measure $R$ on $(\R_+)^2$ with intensity $dsd\theta$ charges a countable set of points $\Omega=\{(s_u,\theta_u),u\in\N  \}$ (with  mass $1$ on each point) (e.g. \cite{watanabe1981stochastic} Chapter I.8 for a complete definition).
% and satisfies that for every $A_1$ and $A_2$ disjoints Borel sets of $(\R_+)^2$
%\begin{itemize}
% \item $\#(\Omega\cap A_1)$ and $\#(\Omega\cap A_2)$ are independent
%\item $\#(\Omega\cap A_1)$ is a random variable distributed according to a Poisson law of parameter $\int_{(\R_+)^2} \un_{A_1}(s,\theta) dsd\theta$.
%\end{itemize}
Then
$ \int_0^t\int_{\R_+} \un_{\theta\le b(x_i)N^K_i(s-)} R_i(ds,d\theta) $ only counts the points $(s^i_u,\theta^i_u)_{u\in\N}$ such that $s_u^i\le t$ and $\theta^i_u \le b(x_i)N^K_i(s^i_u-)$. Thus, we select the points of $R_i$ which correspond to birth events of the prey population. The other integrals have similar interpretations.
\medskip\\
The existence of solutions of \eqref{pop-poisson} is justified by Proposition \ref{thm:majunif}i).
From this construction, we deduce the expression of the prey and the predator population sizes:
\be
\label{taille-poiss}
\begin{aligned}
N^K(t)&=N^K(0) +\sum_{i=1}^d \left[ \int_0^t\int_{\R_+} 
  \un_{\theta\le b(x_i)N^K_i(s-)}
  R_i(ds,d\theta)\right.\\
& -\left.\int_0^t\int_{\R_+}
  \un_{\theta\le  \lambda(x_i,\Zbf(s-))N^K_i(s-) }
  M_i(ds, d\theta)\right] ,\\
H^K(t)&=H^K(0) +\sum_{l=1}^m \Bigl[ \int_0^t\int_{\R_+} 
  \un_{\theta\le rH^K_l(s-)\bigl(\sum_{i=1}^d \frac{B(x_i,y_l)}{K}N^K_i(s-) \bigr)}
  R_{d+l}(ds,d\theta) \\
&-\int_0^t\int_{\R_+}
  \un_{\theta\le  D(y_l)H^K_l(s-) }
  M_{d+l}(ds, d\theta)\Bigr].
\end{aligned}
\ee}
\section{Proof of Proposition \ref{thm:majunif}}
\label{app:majunif}
(i) For the first part, we compare the prey population with a population evolving in the absence of predators. Let us denote by $(\widetilde{N}^K_1,\dots,\widetilde{N}^K_d)$ the sizes of the prey sub-populations evolving without predators and set $\widetilde{N}^K=\sum_{i=1}^d\widetilde{N}^K_i$. Using the description given in Appendix \ref{app:poisson}, we construct the processes $N^K$ and $\widetilde{N}^K$ on the same probability space in such a way that $\forall t\ge0$, $\widetilde{N}^K(t)\ge N^K(t)$ almost surely.
Fournier, M\'el\'eard (2004) (Theorem 5.3) and Champagnat (2006) (Lemma 1) established that
\ben
\label{moment-champ-mod}
\sup_K\E\Bigl(\sup_{t\in[0,T]} \Bigl(\frac{\widetilde{N}^K(t)}{K} \Bigr)^3 \Bigr)<\infty \quad \text{and }\quad \E\Bigl(\sup_K  \sup_{t\in[0,T]} \Bigl(\frac{\widetilde{N}^K(t)}{K} \Bigr)^3 \Bigr)<\infty.
\een
The process $N^K$ then satisfies the same moment properties.
To study the number of predators, we define $\tau_n=\inf\{t\ge0,H^K(t)\ge n\}$. By neglecting the death events, we obtain that
$$
\E\Bigl(\sup_{t\in[0,T\wedge\tau_n]}(\frac{H^{K}(t)}{K})^3 \Bigr)\le \E\Bigl((\frac{H^{K}(0)}{K})^3\Bigr)
+4\E\Bigl(\int_0^{T\wedge\tau_n}  (1+(\frac{H^{K}(s)}{K})^{2})\frac{H^{K}(s)}{K} \frac{N^K(s)}{K} r\bar{B}ds \Bigr) ,
$$
where we used that $(1+x)^3-x^3\le 4(1+x^2)$, $\forall x\ge0$. Since the process $\widetilde{N}^K$ is independent of the number $H^K$ of predators we get that
$$
\begin{aligned} 
 \E\Bigl(\sup_{t\in[0,T\wedge\tau_n]}(\frac{H^{K}(t)}{K})^3\Bigr) 
&\le \phi(T) +2r\bar{B}\int_0^{T} \E\Bigl(\sup_{t\in[0,s\wedge \tau_n]}\bigl(\frac{H^{K}(t)}{K}\bigr)^{3}\Bigr)\E\Bigl(\sup_{t\in[0,s]} \frac{\widetilde{N}^K(t)}{K}\Bigr) ds,\\
\end{aligned}
$$
where $\phi(T)=\E\Bigl((\frac{H^{K}(0)}{K})^3\Bigr)+2r\bar{B}T\E\left( \sup_{t\in[0,T]}\frac{\widetilde{N}^K(t)}{K}\right)$.
By Gronwall's Lemma and \eqref{moment-champ-mod}, we obtain that
\ben
\label{gron}
\E\Bigl(\sup_{t\in[0,T\wedge\tau_n]}(\frac{H^{K}(t)}{K})^3\Bigr) \le C(T)
\een
which concludes point (i) and proves the existence of $\Zbf^K$ for all times.\\
(ii) The second part is much more difficult since using such a coupling is not possible: the constant $C(T)$ obtained in \eqref{gron} goes to $\infty$ as $T\to\infty$.
In the sequel we study the behavior of the time derivative of $\E\Bigl(\bigl(  \frac{N^K(t)+H^K(t)}{K} \bigr)^2\Bigr)$. We gather together the terms related to predation and bound the other terms using Assumption \ref{hyp:existence} to obtain
\bean
\label{E-carre-maj}
\frac{d}{dt}\E&\Bigl(\bigl( \frac{N^K(t)}{K}+\frac{H^K(t)}{K} \bigr)^2\Bigr)\le\E\bigl(K\Psi(\Zbf^K(t))\bigr),
\eean
where 
\bean
\label{psi}
&\Psi(\Zbf^K)= \sum_{i=1}^d\sum_{l=1}^m \frac{H^K_l}{K}\frac{N^K_i}{K}B(x_i,y_l)\\
&\quad \times \Bigl[(\frac{N^K+H^K-1}{K})^2-(\frac{N^K+H^K}{K})^2+r(\frac{N^K+H^K+1}{K})^2-r(\frac{N^K+H^K}{K})^2 \Bigr]\\
&+\frac{N^K}{K}\bar{b} \Bigl[(\frac{N^K+H^K+1}{K})^2-(\frac{N^K+H^K}{K})^2\Bigr]+\cu\frac{ (N^K)^2 }{K^2} \Bigl[(\frac{N^K+H^K-1}{K})^2-(\frac{N^K+H^K}{K})^2\Bigr]\\
&+\frac{H^K}{K} \Du \Bigl[(\frac{N^K+H^K-1}{K} )^2-(\frac{N^K+H^K}{K})^2\Bigr].
\eean
The function $\Psi$ is the sum of three terms that we handle separately. The first term gathers together all the predation effects. The second term (sum of the second and third terms) only depends on the prey population. The last term is related to the death of predators.
We start with the first term. To remove the dependence on the traits, we search for conditions on the term between square brackets to be non positive. This is equivalent to consider the sign of
$(1-\frac{1}{n+h})^2-1+r(1+\frac{1}{n+h})^2-r,$ for $(n,h)\in\N^2\setminus\{(0,0)\}$.
It is non positive as soon as $n+h\ge \frac{(1+r)}{2(1-r)}=n_1$. Thus if $N^K>n_1$,
\bean
\label{1}
\sum_{i=1}^d\sum_{l=1}^m \frac{H^K_l}{K}\frac{N^K_i}{K}&B(x_i,y_l)\left(\frac{N^K+H^K}{K}\right)^2 \Bigl[(1-\frac{1}{N^K+H^K })^2-1+r(1+\frac{1}{N^K+H^K})^2 -r \Bigr]\\
\le&\frac{N^KH^K}{K^2}\left(\frac{N^K+H^K}{K}\right)^2\Bu \Bigl[(1-\frac{1}{N^K+H^K })^2-1+r(1+\frac{1}{N^K+H^K})^2 -r \Bigr],
\eean
which is non positive.\\
\noindent For the second term, let us remark that if $N^K>K \frac{2\bar{b}}{\cu}=Kn_2$, then
\bean
\label{2}
N^K\bar{b}\Bigl(\bigl(1&+\frac{1}{N^K+H^K}\bigr)^2-1\Bigr) +\frac{\cu}{K}(N^K)^2\Bigl(\bigl(1-\frac{1}{N^K+H^K}\bigr)^2-1\Bigr)
\\&\le N^K\bar{b}\Bigl[\Bigl(\bigl(1+\frac{1}{N^K+H^K}\bigr)^2-1\Bigr) +2\Bigl(\bigl(1-\frac{1}{N^K+H^K}\bigr)^2-1\Bigr)\Bigr],
\eean
We set $n_0=\max(n_1,n_2)$. If $N^K\ge Kn_0$, we obtain by combining \eqref{1} and \eqref{2} that:
\bean
\label{3}
\Psi(\Zbf^K)&\le \frac{1}{K}\Bigl(\frac{N^K+H^K}{K}\Bigr)^2 \left[
N^K\bar{b}\Bigl[\Bigl(\bigl(1+\frac{1}{N^K+H^K}\bigr)^2-1\Bigr) +2\Bigl(\bigl(1-\frac{1}{N^K+H^K}\bigr)^2-1\Bigr)\Bigr]\right.\\
&+\left.H^K\Du \Bigl[(1 -\frac{1}{N^K+H^K})^2-1\Bigr] \right] .
\eean
Finally the term between square brackets in \eqref{3} is smaller than $-\min(\bar{b},\Du)$, as soon as $N^K\ge Kn_0$ for $n_0$ large enough. Thus $\forall t\ge0$,
\bea
\frac{d}{dt}\E&\Bigl(\Bigl( \frac{N^K(t)+H^K(t)}{K}\Bigr)^2\Bigr)
\\&\le
\E\Bigl(-\min(\bar{b},\Du)\Bigl( \frac{N^K(t)+H^K(t)}{K}\Bigr)^2\un_{\{N^K(t)>Kn_0\}} +K\Psi(\Zbf^K(t))\un_{\{N^K\le Kn_0\}} \Bigr).
\eea
We now consider the event $\{N^K\le Kn_0\}$. On this event we aim at bounding from above the function $\Psi$ with 
\bean
\Psi&(\Zbf^K)\le \frac{1}{K}(\frac{N^K+H^K}{K})^2\Phi^K(\frac{N^K}{K},\frac{H^K}{K}).
\eean
Since for $(n,h)\in\N^2\setminus\{(0,0)\}$,
\bea
(1-\frac{1}{n+h})^2-1 +r(1+\frac{1}{n+h})^2-r
= -2(1-r)\frac{1}{n+h}+(1+r)\frac{1}{(n+h)^2},
\eea
and Assumption \ref{hyp:existence}, we set for every $(u,v)\in (\R_+)^2\setminus\{(0,0)\}$,
\ben
\label{Phi}
\Phi^K(u,v)= \frac{2u}{u+v}(\bar{b}-\cu u)-\frac{2v}{u+v}((1-r)\Bu u+\Du)+\frac{u}{K(u+v)^2}(\bar{b}+\cu u+(1+r)\Bb v)+\Du\frac{v}{K(u+v)^2}.
\een
We seek a condition on $v$ to obtain that $\Phi^K(u,v)\le -D$,  $\forall K\ge0$, $\forall 0\le u\le n_0$. This inequality can be written as a polynomial
\ben
\label{poly}
v^2 \alpha(u)+ v\beta(u,K)+\gamma(u,K)\le 0,
\een
where the coefficients are given by
\be \left\{
\begin{aligned}
\alpha(u)&=-2(1-r)\Bu u-\Du,\\
\beta(u,K)&=2u(\bar{b}-\cu u)-2u^2(1-r)\Bu +\frac{u}{K}(1+r)\bar{B}+\frac{\Du}{K},\\
\gamma(u,K)&=\frac{u}{K}(\bar{b}+\cu u +\Du u^2+2u^2(\bar{b}-\cu u)).\\
\end{aligned}\right.
\ee
As $\alpha(u)<0$, this polynomial remains negative for every $v$ greater than its largest real root. If the polynomial \eqref{poly} has real roots, then we can bound from above the largest one with
$$
\frac{|\beta(u,K)|+\sqrt{\beta(u,K)^2-4\alpha(u)\gamma(u,K)}}{-2\alpha(u)}.
$$  
The coefficient $\beta(u,K)$ decreases with $K$, thus for every $K\ge1$ and $0\le u\le n_0$,
\bea
2u(\bar{b}-\cu u)-2u^2(1-r)\Bu\le \beta(u,K)&\le2u(\bar{b}-\cu u)-2u^2(1-r)\Bu +u(1+r)\bar{B}+\Du\\
-2(1+\cu) n_0^2\Bu\le \beta(u,K)&\le 2n_0 \bar{b}+n_0(1+r)\Bu +\Du.
\eea
In the case where  $\gamma(u,K)<0$, the discriminant $\Delta(u,K)=\beta(u,K)^2-4\alpha(u)\gamma(u,K)$ is bounded by $|\beta(u,K)|$.
Otherwise $\Delta\le \beta(u,K)^2+8((1-r)\Bu u+\Du )u (\bar{b}+\cu u +\Du u^2+2u^2(u-\cu u))$ which can be bounded uniformly for $u\in[0,n_0]$.
Thus there exists $h_0$ independent on $K$ such that 
$$
\forall n\le Kn_0,\quad \forall h>Kh_0,\quad \Phi^K\bigl(\frac{n}{K},\frac{h}{K}\bigr)\le -D.
$$
Finally
\bea
\frac{d}{dt}\E\Bigl(\bigl(&\frac{N^K(t)+H^K(t)}{K} \bigr)^2\Bigr)\le\E\Bigl(K\Psi(\Zbf^K(t))\un_{\{N^K(t)\le Kn_0,H^K(t)\le Kh_0\}} \Bigr)\\
&+\E\Bigl(-C( \frac{N^K(t)+H^K(t)}{K})^2(\un_{\{N^K(t)>Kn_0\}}+\un_{\{N^K(t)\le Kn_0, H^K(t)>Kh_0\}}\Bigr),
\eea
with $C>0$.
To conclude it remains to bound the expectation of $\Psi$ on the event $\{N^K\le Kn_0 \text{ and }H^K\le Kh_0\}$. Keeping only the positive terms we obtain that
\bea
\E\Bigl(&K\Psi(\Zbf^K(t))=\un_{\{N^K(t)\le Kn_0,H^K(t)\le Kh_0\}} \Bigr)
\le \E\Bigl(\un_{\{N^K(t)\le Kn_0,H^K(t)\le Kh_0\}}\\
&\times(\bar{b}N^K(t)+r\bar{B} N^K(t)\frac{H^K(t)}{K})\Bigl((\frac{N^K(t)+H^K(t)}{K}+\frac{1}{K} )^2-(\frac{N^K(t)+H^K(t)}{K} )^2\Bigr)\Bigr)\\
&\le\sum_{n=0}^{Kn_0}\sum_{h=0}^{Kh_0}( \frac{n+h}{K})^2(\bar{b}n+r\bar{B} n\frac{h}{K})\Bigl(( 1+\frac{1}{n+h})^2-1)\\
&\le \sum_{n=0}^{Kn_0}\sum_{h=0}^{Kh_0}( \frac{n+h}{K})^2 3(\bar{b}+r\bar{B} \frac{h}{K}),
\eea
where the last inequality derives from $(1+u)^2-1\le 3u$, for all $u\in[0,1]$.\\
Combining all these results
\bea
\frac{d}{dt}\E\left(( \frac{N^K(t)+H^K(t)}{K})^2\right)
& \le \E\Bigl(-C( \frac{N^K(t)+H^K(t)}{K})^2)\Bigr)\\
&+ \int_0^{n_0}\int_{0}^{h_0}\Bigl(3( n+h)^2(C+\bar{b})+ (n+h)^3r\Bu \Bigr)dh dn \\
&\le C'-C\E\Bigl((\frac{N^K(t)+H^K(t)}{K} )^2\Bigr),
\eea
with $C'>0$.
We solve this inequality to get that
\bea
\E\left(( \frac{N^K(t)+H^K(t)}{K})^2\right)&\le C' +\left(\E\Bigl( ( \frac{N^K(0)+H^K(0)}{K} )^2\Bigr)-C'\right)e^{-Ct}.
\eea
which gives the uniform bound.

\section{Proof of Theorem~\ref{thm:ALCP}}
\label{app:proofALCP}
The proof relies on the expression of Linear Complementarity Problems as variational inequality problems.
\begin{definition}
The variational inequality problem associated with a function $f:\R^u\to\R^u$ and a subset $E\subset\R^u$ seeks a vector $z\in E$ such that
\ben
\label{VI}
\forall a\in E,\quad (a-z)^T f(z)\ge0.
\een
\end{definition}
\noindent The existence of solutions is not true in a general setting but we are interested in a specific framework where the subset $E$ is compact and convex. 
\begin{thm}
\label{thm:proj}
Let $E$ be a non empty compact convex of $\mathbb{R}^u$ and $f$ continuous function, then the variational inequality problem associated to $(f,E)$ admits a solution.
\end{thm}
The proof of Theorem \ref{thm:proj} is rather classical and requires to express a solution as a fix point of a projection of the subset $E$ (see \cite{cottleLCP} Theorem 3.7.1).
With this result we can prove the Theorem \ref{thm:ALCP}.
\begin{proof}[Proof of Theorem~\ref{thm:ALCP}]
Let us recall that a solution to the Linear complementarity problem associated to the couple $(\tM,\tq)$ defined in \eqref{ALCP} is a vector $\zbf=(\nbf,\hbf)\in\mathbb{R}^d\times\mathbb{R}^m$ such that: for every $1\le i\le d$ and $1\le l\le m$,
\ben
\label{cond-n}
n_i\ge0,\quad (q+M\nbf+B\hbf)_i\ge0,\quad (\nbf)^T(q+M\nbf+B\hbf)=0 
\een
and
\ben
\label{cond-h}
h_l\ge0,\quad (D-B^T\nbf)_l\ge0, \quad (\hbf)^T(D -B^T\nbf)=0
\een
These conditions \eqref{cond-n} entail that the vector $\nbf$  is a solution to $LCP(M,q+B\hbf)$. \\
Note that if $\nbf\in\R^d$ is solution to the restricted problem $LCP(M,q)$ satisfying moreover $(-B^T\nbf+D)_l\ge0$ for all $1\le l\le m$, then the vector $(\nbf,0)$ is solution to $LCP(\tM,\tq)$.
Similarly we seek a suitable vector $\nbf$ and adjust it thanks to the vector $\hbf$.\\
We consider the variational inequality problem associated to the set
$$
E=\{\nbf\in (\mathbb{R}_+)^d,\quad \forall 1\le l\le m \quad (D-B^T\nbf)_l\ge0   \},
$$
and the continuous function $f(\nbf)=q+M\nbf$. \\
Since $D$ is non negative, the set $E$ is not empty. Moreover $E$ is convex, closed and bounded thus compact. Theorem~\ref{thm:proj} ensures the existence of a solution $\nbf^*$ to this problem. Note that \eqref{VI} can be written as
 $$\forall a\in E, \quad a^Tf(\nbf^*)\ge (\nbf^*)^Tf(\nbf^*).$$ 
Thus $\nbf^*$ minimizes the function $a\to a^Tf(\nbf^*)$ on $E$.
Therefore
\begin{itemize}
 \item either $\nbf^*$ is in the interior of $E$ and is therefore a global minimizer of the function $a\to a^Tf(\nbf^*)$ on $\R^d$ and $(\nbf^*,0)$ is a solution to $LCP(\widetilde{M},\widetilde{q})$. 
\item otherwise we can define the Lagrange multipliers for this problem. There exist $d+m$ non negative real $h_1,....,h_{d+m}$ such that $\forall 1\le i\le d $, $\forall 1\le l\le m $,
$$ 
(q+M\nbf^*)_i=h_i-\sum_{l=1}^mB_{il}h_{d+l},\quad h_in_i^*=0,\text{ and } h_{d+l}(-B^T \nbf^*+D)_k=0.
$$
The first condition entails that $h_i=(q+M\nbf^*)_i+\sum_{l=1}^mB_{il}h_{d+l}$ and therefore the vector $(\mathbf{n}^*,h_{d+1},\dots,h_{d+m})$ is a solution to $LCP(\widetilde{M},\widetilde{q})$. 
\end{itemize}
\end{proof}

\section{Proof of Theorem~\ref{thm:sortie_eq}}
\label{app:proof_sortie_eq}
\red{A perturbation $\Zm^K=(\mathcal{N}^K_1,\cdots,\mathcal{N}^K_d,\mathcal{H}^K_1,\cdots,\mathcal{H}^K_m)$ of the prey-predator community process is defined by $2$ families of $d+m$ real-valued random processes $(u^K_i)_{1\le i \le d+m}$ and $(v^K_i)_{1\le i \le d+m}$ which are predictable with respect to the filtration $\mathcal{F}_t$ generated by the processes $\Zbf^K$.  Both families are uniformly bounded by a parameter $\kappa>0$.\\
The perturbation $\Zm^K$ is solution of the following system of stochastic differential equations driven by the Poisson point measures $R_i$ and $M_i$ introduced in Appendix \ref{app:poisson}.
\bean
 \label{Z-mod}
\mathcal{Z}^K(t)&=\mathcal{Z}^K(0) +\sum_{i=1}^d \Bigl[ \int_0^t\int_{\R_+} \frac{e_i}{K}
  \un_{\theta\le b(x_i)\mathcal{N}^K_i(s-) +u^K_i(s)}
  R_i(ds,d\theta) \\
&-\int_0^t\int_{\R_+}\frac{e_i}{K}
  \un_{\theta\le  \mathcal{N}^K_i(s-))\lambda(x_,\Zm^K(s-)) +v^K_i(s)}
  M_i(ds, d\theta)\Bigr] \\
& +\sum_{l=1}^m \Bigl[ \int_0^t\int_{\R_+} \frac{e_{d+l}}{K}
  \un_{\theta\le r\mathcal{H}^K_l(s-)\Bigl(\sum_{i=1}^d \frac{B(x_i,y_l)}{K}\mathcal{N}^K_i(s-) \Bigr)+u^K_{d+l}(s)} 
  R_{d+l}(ds,d\theta) \\
&-\int_0^t\int_{\R_+}\frac{e_{d+l}}{K}
  \un_{\theta\le  D(y_l)\mathcal{H}^K_l(s-) +v^K_{d+l}(s)}
  M_{d+l}(ds, d\theta)\Bigr].
\eean
where $(e_1,\dots,e_d,e_{d+1},\dots,e_{d+m})$ is the canonical basis of $\R^{d+m}$.}\\
The proof relies on the study of the stochastic process $L(\Zm^K)$ where $L$ is the Lyapunov function for the system $LVP(\xbf,\ybf)$ introduced in \eqref{Lyap-L} with an appropriate choice of $\gamma$.
The function $L$ is the sum of two functions $V$ and $W$. $V$ defined in \eqref{Lyap-V} is linear in the coordinate $n_i$, $i\in P$ and $h_l$, $l\in Q$ and strictly convex in the other coordinates. Moreover, its Hessian matrix at $\zbf^*$ is diagonal. $W$ defined \eqref{Lyap-G} is a quadratic form in $(\zbf-\zbf^*)$. This justifies the inequality \eqref{Prop1}:
\bea
%\label{Prop1}
 \norm{\zbf-\zbf^*}^2&\le\sum_{i\notin P}|n_i-n_i^*|^2 +\sum_{i\in P}|n_i|+\sum_{l\notin Q}|h_l-h_l^*| +\sum_{l\in Q}|h_l|\\
&\le  C\Bigl(L(\zbf)-L(\zbf^*)\Bigr)\le
CC'(\sum_{i\notin P}|n_i-n_i^*|^2 +\sum_{i\in P}|n_i|+\sum_{l\notin Q}|h_l-h_l^*| +\sum_{l\in Q}|h_l|),
\eea
where $P$ and $Q$ have been defined in \eqref{PQ}. We set in the following 
$$\norm{\zbf-\zbf^*}_{PQ}=\sum_{i\notin P}|n_i-n_i^*|^2 +\sum_{i\in P}|n_i|+\sum_{l\notin Q}|h_l-h_l^*| +\sum_{l\in Q}|h_l|
$$
The derivative of $L(\zbf(t))$ given in \eqref{deriv-L} can be bounded from above in the neighbourhood of $\zbf^*$ by
\bea
\frac{d}{dt}L(\zbf(t)) \le -C_1\norm{\nbf(t)-\nbf^*}^2 -C_1\Bigl(\sum_{i\in P} n_i(t) +\sum_{l\in Q} h_l(t) \Bigr)-C_1 \sum_{i\notin P}\Bigr( \sum_{l\notin Q} B_{il} (h_l(t)-h^*_l)\Bigl)^2,
\eea
for a positive real number $C_1$.
If we set 
$$C_2=\inf\{ \sum_{i\notin P}\Bigr( \sum_{l\notin Q} B_{il} (h_l-h^*_l)\Bigl)^2, \hbf\in(\R_+)^m, \norm{\hbf-\hbf^*}=1\}>0,
$$
then
\bea
\frac{d}{dt}L(\zbf(t)) \le -C_1\norm{\nbf(t)-\nbf^*}^2 -C_1\Bigl(\sum_{i\in P} n_i(t) +\sum_{l\in Q} h_l(t) \Bigr)-C_1C_2 \sum_{l\notin Q}(h_l(t)-h^*_l)^2.
\eea
We then obtain \eqref{Prop2}:
\begin{equation*}
%  \label{Prop2}
\frac{d}{dt}L(\zbf(t)) \le -C''  \norm{\zbf-\zbf^*}^2.
\end{equation*}
We introduce $\tau_{\varepsilon}^K=\inf\{t\ge0, \Zm^K(t)\notin B_{\varepsilon}\}$. In the sequel we prove that there exist $\varepsilon''<\varepsilon$ and $V>0$ such that if $\Zm^K(0)\in\mathcal{B}{_{\varepsilon''}}$, then
\ben
\label{but}
\lim_{K\to\infty} \Pro\Bigl( \tau^K_{\varepsilon}>e^{KV}\Bigr) =1.
\een
For every $t\le \tau_{\varepsilon}^K$,
 \bea
 L(&\Zm^K(t))=L(\Zm^K(0))  +M^K(t)\\
 &+\int_0^t\sum_{i=1}^d \Bigl( L(\Zm^K(s)+\frac{e_i}{K})-L(\Zm^K(s))
 \Bigr)\Bigl(\Nm^K_i(s)b(x_i)+u^K_i(s)\Bigr)ds \\
 &+\int_0^t\sum_{i=1}^d\Bigl( L(\Zm^K(s)-\frac{e_i}{K})-L(\Zm^K(s)) \Bigr)\Bigl(\Nm^K_i(s)\lambda(x_i,\Zm^K(s))+v^K_i(s)\Bigl)
    ds\\
 &+ \int_0^t\sum_{l=1}^m\Bigl( L(\Zm^K(s)+\frac{e_{d+l}}{K})-L(\Zm^K(s))
 \Bigr) \Bigl(\Hm^K_l(s)\Bigl(r\sum_{i=1}^dB(x_i,y_m)\frac{\Nm^K_i(s)}{K}\Bigr)+u^K_{d+l}\Bigr)ds 
  \\
 &+ \int_0^t\sum_{l=1}^m\Bigl( L(\Zm^K(s)-\frac{e_{d+l}}{K})-L(\Zm^K(s)) \Bigr)\Bigl(\Hm^K_l(s)D(y_l)+v^K_{d+l}\Bigr)ds  .
 \eea
where $M^K_t$ is a local martingale which can be expressed with respect to the compensated Poisson point measures $(\widetilde{R}_i)_{1\le i\le d+m}$ and $(\widetilde{M}_{i})_{1\le i\le d+ m}$:
\bean
\label{mart-taille}
&M^K(t)=\sum_{i=1}^d \left[\int_0^t\int_0^{\infty}\Bigl[L(\Zm^K(s-)+\frac{\delta_{e_i}}{K})-L(\Zm^K(s-)) \Bigr]\right.
\un_{\theta\le b(x_i)\Nm^K_i(s-) + u^K_i(s)} \widetilde{R}_i(ds,d\theta) \\
&+\int_0^t\int_0^{\infty}
\Bigl[L(\Zm^K(s-)-\frac{\delta_{e_i}}{K})-L(\Zm^K(s-)) \Bigr] \left.
  \un_{\theta\le\Nm^K_i(s-)\lambda(x_i,\Zm^K(s-)) +v^K_i(s)}
 \widetilde{M}_i(ds,d\theta)\right]\\
&+\sum_{l=1}^m\left[\int_0^t\int_0^{\infty}
\Bigl[L(\Zm^K(s-)+\frac{\delta_{e_{d+l}}}{K})-L(\Zm^K(s-)) \Bigr]\right.
 \un_{\theta\le \Hm^K_l(s-)\Bigl(r\sum_{i=1}^d\frac{B(x_i,y_l)}{K}\Nm^K_i(s-)\Bigr) +u^K_{d+l}(s)}
  \widetilde{R}_{d+l}(ds,d\theta)\\
 &+\int_0^t\int_0^{\infty}
\Bigl[L(\Zm^K(s-)-\frac{\delta_{e_{d+l}}}{K})-L(\Zm^K(s-))  \Bigr]
\left.\un_{\theta\le D(y_l)\Hm^K_l(s-) +v^K_{d+l}(s)} \widetilde{M}_{d+l}(ds,d\theta)\right].
\eean
For every $t\le \tau_{\varepsilon}^K$ and $1\le i\le d$ we give the second order expansion of the terms
$$
 L(\Zm^K(t)+\frac{e_i}{K})-L(\Zm^K(t))=\frac{1}{K}\frac{\partial}{\partial e_i} L(\Zm^K(t))+\frac{1}{2}
 \int_{0}^{\frac{1}{K}}(\frac{\Nm^K_i(t)}{K}+\frac{1}{K}-u)\frac{\partial^2}{\partial e_i^2} L\Bigl(\Zm^K(t)-(\frac{\Nm^K_i(t)}{K}-u )e_i\Bigr)du.
  $$
We obtain a similar equality for the derivative with respect to $e_{d+l}$ for $1\le l\le m$. \\
Let us remark that $\sup\{\frac{\partial^2}{\partial e_j^2} L(u,v) ,(u,v)\in\mathcal{B}_{\varepsilon}\}<\infty$ for $\varepsilon$ small enough, for all $1\le j \le d+m$. Therefore the integrated term is of order $1/K^2$ for large $K$. The impact of the perturbed terms can be bounded similarly using the first derivative. Thus
 \bea
 L&(\Zm^K(t))= L(\Zm^K(0))  +M^K(t)\\
 &+\int_0^t \sum_{i=1}^d \frac{\partial L(\Zm^K(s))}{\partial e_i}
\frac{\Nm^K_i(s)}{K}\Bigl[b(x_i)-d(x_i)-\sum_{j=1}^dc(x_i,x_j)\frac{\Nm^K_j(s)}{K}-\sum_{l=1}^m\frac{\Hm^K_l(s)}{K}B(x_i,y_l)
 \Bigr] 
    ds\\
 &+ \int_0^t \sum_{l=1}^m  \frac{\partial L(\Zm^K(s))}{\partial e_{d+l}} \frac{\Hm^K_l(s)}{K} \Bigl[r\sum_{i=1}^dB(x_i)\frac{\Nm^K_i(s)}{K}-D(y_l)\Bigr]ds  +
 \mathcal{O}\Bigl(\frac{t}{K}\Bigr)+\mathcal{O}\bigl(\kappa t\bigr).
 \eea
Note that if $\zbf(t)$ is a solution of $LVP(\xbf,\ybf)$ then:
 \bea
 \pderiv{L(\zbf(t))}{t}&=\sum_{i=1}^d \frac{\partial}{\partial e_i} L(\zbf(t))n_i(t)
 \Bigl[b(x_i)-d(x_i)-\sum_{j=1}^dc(x_i,x_j)n_j(t)-\sum_{k=1}^mB(x_i,y_k)h_k(t) \Bigr]\\
 & + \sum_{l=1}^m \frac{\partial}{\partial e_{d+l}}L(\zbf(t))
 h_l(t)\Bigl[r\sum_{i=1}^dB(x_i,y_l)n_i(t)-D(y_l)\Bigr].
 \eea
We denote by $\pderiv{L(\Zm^K(t))}{t}$ the derivative along the solution $\zbf$ such that $\zbf(t)=\Zm^K(t)$. Then for $\kappa\ge 1/K$:
 \begin{align*}
 L(\Zm^K(t))= &L(\Zm^K(0))  +M^K(t)+ \int_0^t \pderiv{L(\Zm^K(s))}{t}ds +\mathcal{O}\bigl(\kappa t\bigr).
  \end{align*}
Using inequalities \eqref{Prop1} and \eqref{Prop2} we obtain that there exists $C'''>0$, such that if $ t\le T\wedge
 \tau_{\varepsilon}^K$ then
 \bean
\label{maj_normbis}
 \norm{\Zm^K(t)-\zbf^*}^2&\le
 C\Bigl[C'(\norm{\Zm^K(0)-\zbf^*}_{PQ}+ \sup_{t\in[0,T]}|M^K(t)|
 -C''\int_0^t\bigl(\norm{\Zm^K(s)-\zbf^*}^2-C'''\kappa\bigr) ds\Bigr].\\
 \eean
This inequality is the main tool of the proof. It connects the time spent by the process above a given threshold with the values it takes during this time interval.\\
\noindent
We define $S_{\kappa}=\inf\{t\ge0,\norm{\Zm^K(t)-\zbf^*}^2 \le 2C'''\kappa\}$.
Then for every $t\le S_{\kappa}\wedge T \wedge \tau_{\varepsilon}^K$:
 \be
 \norm{\Zm^K(t)-\zbf^*}^2\le C\Bigl[C'(\norm{\Zm^K(0)-\zbf^*}_{PQ})+
 \sup_{[0,T]}|M^K(t)|-C''C'''\kappa t \bigr) \Bigr].
 \ee
 As the l.h.s. is nonnegative we define
 \begin{equation}
 \label{Teta-bis}
  T_{\kappa}=\frac{C'(\norm{\Zm^K(0)-\zbf^*}_{PQ}+
 \sup_{[0,T]}|M^K(t)|}{C''C'''\kappa}\ge0,
 \end{equation}
which can be seen as the maximal time spent by the process $\norm{\Zbf^K(t)-\zbf^*}^2$ above $2C'''\kappa$ before the time $T\wedge \tau_{\varepsilon}^K$.
Therefore for every $t\le S_{\kappa}\wedge T
 \wedge\tau_{\varepsilon}^K$:
 \be
 \norm{\Zm^K(t)-\zbf^*}^2\le CC''C'''\kappa T_{\kappa}.
 \ee
To control the norm $\norm{\Zm^K(t)-\zbf^*}^2$ it remains to control $T_{\kappa}$ and thus the martingale $M^K$. To obtain the uniform bound, we use the exponential bound given by Lemma~\ref{lem:mart}.
On the event
\ben
\label{evenement}
\Bigl\{ T_{\kappa}\le T\wedge\frac{\varepsilon^2}{2CC''C'''\kappa} \Bigr\},
\een
then $\sup_{[0,S_{\kappa}]}(\norm{\Zm^K(t)-\zbf^*}^2) \le \frac{\varepsilon^2}{2},$ and in particular
 $S_{\kappa}\le \tau_{\varepsilon}^K\wedge T_{\kappa}.$\\
Moreover applying \eqref{maj_normbis} on the same event we get 
\begin{equation}
\label{maj_norm2}
\sup_{[0,T\wedge \tau_{\varepsilon}^K]}(\norm{\Zbf^K(t)-\zbf^*}^2)\le
CC''C'''\kappa(T+T_{\kappa})\le \frac{\varepsilon^2}{2}+CC''C'''\kappa T.
\end{equation}
Thus if furthermore $\kappa<\varepsilon^2/(2CC''C'''T)$ then $\tau_{\varepsilon}^K> T$.\\
\newline
\noindent These results lead to the Theorem. Let $\varepsilon'>0$ such that
 $\varepsilon''<\varepsilon'/2<\varepsilon'<\varepsilon$. \\
We introduce a sequence of stopping times that describes the back and forth of the process $\Zm^K$ between the balls $\mathcal{B}_{\varepsilon''}$ and $\mathcal{B}_{\varepsilon'/2}$ (see Figure
 \ref{fig:tempsarret}).
Set $\tau_0=0$ and for every $k\ge 1$ such that $\tau_k<\tau_{\varepsilon}^K$:
 \bean
 \tau'_k=&\inf\bigl\{t\ge \tau_{k-1}: \Zm^K(t)\notin B_{\varepsilon'/2} 
 \bigr\},\\
 \tau_k=&\inf\bigl\{t\ge \tau'_{k}: \Zm^K(t)\in B_{\varepsilon''} \text{ ou
 } \Zm^K(t)\notin B_{\varepsilon}   \bigr\}.
 \eean
We denote by $k_{\varepsilon}$ the number of back and forths before the exit: 
$$
k_{\varepsilon}=\inf\{k\in \N,
 \tau_k=\tau_{\varepsilon}^K\}.
 $$
In the sequel we bound $k_{\varepsilon}$ from below.
\begin{figure}
\centering
\includegraphics[scale=0.7]{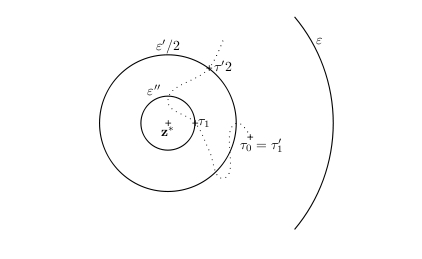}
\caption{\small{A trajectory of $\Zm^K$ in the neighbourhood of $\mathbf{z}^*$ for $d=m=1$.}}\label{fig:tempsarret}
\end{figure}

\noindent We consider an initial condition $\Zm^K(0)\in \mathcal{B}_{\varepsilon'}$. We set $\kappa=(\varepsilon'')^2/2C'''$ and apply the previous results. The time $\tau_1$ corresponds to the first return in $\mathcal{B}_{\varepsilon''}$ therefore it is equal to the time $S_{\kappa}$ introduced before. We deduce from the previous computations that on the event \eqref{evenement}
\begin{align*}
\Pro\bigl( \tau_{1}<\tau_{\varepsilon}^K\bigr)
& = \Pro\bigl( \sup_{[0,\tau_{1}]} \norm{\Zm^K(t)-\zbf^*}^2 <\varepsilon^2 \bigr)\ge\Pro\bigl( T_{\kappa}\le T\wedge\frac{\varepsilon^2}{2CC''C'''\kappa} \bigr).
\end{align*}
We replace $T_{\kappa}$ by its value \eqref{Teta-bis} to get that
 \bea
  \Pro\bigl( T_{\kappa}>T\wedge\frac{\varepsilon^2}{2CC''C'''\kappa} \bigr)
 &=\Pro\Bigl(  \sup_{[0,T]}|M^K(t)| >\bigl(C''C'''\kappa T\wedge\frac{\varepsilon^2}{2C}
 \bigr)-C'(\norm{\Zm^K(0)-\zbf^*}_{PQ})\Bigr)\\
 &\le\Pro\Bigl(  \sup_{[0,T]}|M^K(t)| >\bigl(C''C'''\kappa
 T\wedge\frac{\varepsilon^2}{2C} \bigr)-C'\varepsilon'\Bigr),
 \eea
where we used that $\Zm^K(0)\in\mathcal{B}_{\varepsilon'}$ to obtain the last inequality.\\
If we choose $T=2C'\varepsilon'/C''C'''\kappa$ and $\varepsilon'$ such that
 $2C'\varepsilon'<\frac{\epsilon^2}{2C}$ then the inequality becomes
 \begin{align*}
  \Pro\bigl( T_{\kappa}>T\wedge\frac{\varepsilon^2}{2CC''C'''\kappa} \bigr)
 &\le\Pro\bigl(  \sup_{[0,T]}|M^K(t)|
 >C'\varepsilon'\bigr).
 \end{align*}
We finally use Lemma \ref{lem:mart} to obtain
 \begin{align*}
  \Pro\bigl( T_{\kappa}>T\wedge\frac{\varepsilon^2}{2CC''C'''\kappa} \bigr)
 &\le \exp(-KV),
 \end{align*}
where $V>0$ only depends on $\varepsilon'$ and $\varepsilon''$.\\
Since this inequality remains true as long as the initial condition is in $B_{\varepsilon'}$ we deduce that
 \begin{align}
  \sup_{\Zm^K(0)\in B_{\varepsilon'}}\Pro\Bigl( 
 \tau_1<\tau_{\varepsilon}^K\Bigr)\ge 1-\exp(-KV).
 \end{align}
Applying the strong Markov property at the stopping time $\tau_k$ for $k\ge1$
 \begin{align*}
  \sup_{\Zm^K(0)\in B_{\varepsilon'}}\Pro\Bigl( 
 \tau_k<\tau_{\varepsilon}^K|\tau_{k-1}<\tau_{\varepsilon}^K\Bigr)\ge
 1-\exp(-KV).
 \end{align*}
therefore we can bound $k_{\varepsilon}$ from below by a random variable distributed according to a geometric law of parameter $\exp(-KV)$. Then
\ben
\label{maj-keps}
\lim_{K\to\infty} \Pro(k_{\varepsilon}>\exp(KV/2))=1.
\een
It remains to prove that these back and forths do not happen too fast. We establish that the time intervals $\tau_k-\tau_{k-1}$ are of order $1$ for $k\ge2$. To this aim we search for $T'$ such that for every $k\ge2$,
 $\Pro(\tau'_k-\tau_{k-1}>T')>0$.
Using the strong Markov property again, it is sufficient to prove that
 $\inf_{\Zm^K(0)\in B_{\varepsilon''}}\Pro(  \tau'_1>T')>0$:
 \begin{align}
  \inf_{\Zm^K(0)\in B_{\varepsilon''}}\Pro(  \tau'_1>T')=\inf_{\Zm^K(0)\in B_{\varepsilon''}}\Pro\Bigl( \sup_{[0,T'\wedge
 \tau'_1]}\norm{\Zm^K(t)-\zbf^*}^2<\frac{\varepsilon'^2}{4}\Bigr).
 \end{align}
We deduce from \eqref{maj_norm2} with $\varepsilon=\varepsilon'/2$ that on the event $\{T_{\kappa}\le
 T'\wedge\frac{\varepsilon'^2}{8CC''C'''\kappa}\}$:
 \begin{align*}
  \sup_{[0,T'\wedge \tau'_1]}(\norm{\Zm^K(t)-\zbf^*}^2)^2\le
 CC''C'''\kappa(T'+T_{\kappa})\le \frac{\varepsilon'^2}{8}+CC''C'''\kappa T'.
 \end{align*}
Setting $T'=2C'\varepsilon''/C''C'''\kappa$ and $\varepsilon''$ such that $2C'\varepsilon''<\varepsilon'^2/4C$, we get that 
 \begin{align*}
  \sup_{[0,T'\wedge \tau'_1]}\norm{\Zm^K(t)-\zbf^*}^2<
 \frac{\varepsilon'^2}{4},
 \end{align*}
and thus $\tau'_1>T'$.\\
Lemma \ref{lem:mart} ensures again that for any initial condition in $ \mathcal{B}_{\varepsilon''}$:
 $$
 \begin{aligned}
  \Pro\Bigl(T_{\kappa}> T'\wedge\frac{\varepsilon'^2}{8CC''C'''\kappa}  \Bigr)\le\Pro\Bigl( \sup_{[0,T']}|M^K(t)| > C'\varepsilon'' \Bigr)  
  \underset{{K\to\infty}}{\longrightarrow}0.
 \end{aligned}
 $$
and thus
$$\inf_{\Zm^K(0)\in B_{\varepsilon''}}\Pro\Bigl( \sup_{[0,T'\wedge
 \tau'_1]}\norm{\Zm^K(t)-\zbf^*}^2<\frac{\varepsilon'^2}{4}\Bigr)   \underset{{K\to\infty}}{\longrightarrow} 1.$$
Finally \eqref{but} is deduced from \eqref{maj-keps}.

\medskip
\noindent\textbf{Acknowledgements:} The authors are grateful to Fr\'ed\'eric Bonnans who provided insight and expertise on Linear Complementarity Problems. This article benefited from the support
of the ANR MANEGE (ANR-09-BLAN-0215) and from the Chair ``Mod\'elisation Math\'ematique
et Biodiversit\'e" of Veolia Environnement - \'Ecole Polytechnique - Museum National d'Histoire
Naturelle - Fondation X.
\bibliographystyle{plain}
\begin{spacing}{0.5}
%\footnotesize{
\bibliography{Biblio.bib}
%}
\end{spacing}
\end{document}